\documentclass[runningheads]{llncs}

\usepackage[numbers]{natbib}
\usepackage{url}
\usepackage{amssymb}
\usepackage{amsmath}
\usepackage{stmaryrd}
\usepackage{syntax}
\usepackage{semantic}
\usepackage{url}
\usepackage{xcolor}
\usepackage{nicefrac}
\newcommand{\relu}{\unlhd}
\newenvironment{restate}[1]%
{\begin{trivlist}\item[]{\normalsize\bf Restatement of #1}\hspace*{4mm}\it}%
{\end{trivlist}}
\newcommand{\modownarrow}{\mathord{\downarrow}}

\begin{document}

%\title{(On the difficulty of) Non-terminating probabilistic programs in machine learning}
\title{Weakest Preexpectation Semantics \\ for Bayesian Inference\thanks{This work is supported by the ERC Advanced Grant Project FRAPPANT (project number 787914).}}
\subtitle{Conditioning, Continuous Distributions and Divergence}
\author{Marcin Szymczak \and Joost-Pieter Katoen}
\institute{Software Modelling and Verification Group \\ RWTH Aachen University \\ 52056 Aachen, Germany}

\maketitle

\begin{abstract}
We present a %the first 
semantics of a probabilistic while-language with soft conditioning and continuous distributions which handles programs diverging with positive probability. To this end, we extend the probabilistic guarded command language (pGCL) with draws from continuous distributions and a score operator. The main contribution is an extension of the standard weakest preexpectation semantics to support these constructs. As a sanity check of our semantics, we define an alternative trace-based semantics of the language, and show that the two semantics are equivalent. Various examples illustrate the applicability of the semantics.
\end{abstract} 

\section{Introduction}

%Intro: current research on PP for ML
Research on semantics of probabilistic languages for machine learning 
\cite{DBLP:journals/corr/BorgstromGGMG13, R2, DBLP:conf/esop/TorontoMH15, DBLP:conf/lics/HeunenKSY17,
DBLP:conf/icfp/BorgstromLGS16} has so far
focused almost exclusively on almost-surely terminating programs. These programs terminate on all possible inputs with probability one. 
This seems a reasonable assumption, because
not only most probabilistic models used in practice terminate with probability one, but programs which may diverge
with a positive probability also make not much sense in the context of probabilistic inference.

%Motivation
However, one cannot simply assume that in the context of statistical probabilistic programming, 
divergence is a non-existing issue which can be ignored completely. For one thing, models which are
not guaranteed to terminate actually exist, and are not merely degenerate cases. Even if one cannot 
apply inference in this case, being able to reason about such programs is important, as it helps to define 
suitable approximations and check their correctness. Moreover, the line dividing almost surely terminating
and possibly diverging programs can sometimes be very thin and a small change to some parameter values
may make a terminating program diverge.

%Motivating example
To make a case for potentially diverging probabilistic programs, consider a variation of the tortoise
and hare problem described by Icard~\cite{DBLP:conf/cogsci/Icard17} as a simple problem in intuitive physics:
a tortoise is walking at some low constant speed and a hare, which was initially behind the tortoise
and moves forward with random, fast strides, is trying to catch it.
Assuming that the tortoise is moving at a constant speed of $1{+}\mathtt{e}$ cm per second (where $\mathtt{e}$ is some small constant), and each second 
the hare moves with probability $\frac{1}{4}$ by a random Gaussian-distributed distance, being 4cm on average,
%at a random speed between $1$ and $7$ cm per second in each time step, 
we would like to calculate the average time after which the hare will catch the tortoise. We can model this problem by the following probabilistic program:

%\begin{verbatim}
%t := 5.0; 
%h := 0.0;
%dt := 0.05;
%time := 0.0;
%while (h < t)
%{
%  t := t + 1*dt + e;
%  if (flip(0.25*dt)) 
%    h := h + Uniform(1,7)*dt;
%  time := time + dt;
%}
%return time
%\end{verbatim}
\begin{verbatim}
t := 5.0; 
h := 0.0;
time := 0.0;
while (h < t)
{
  t := t + 1 + e;
  if (flip(0.25)) 
    h := h + Gaussian(4,2);
  time := time + 1;
}
return time
\end{verbatim}

%\noindent where $\mathtt{dt}$ is the time step. 
\noindent where $\mathtt{flip}(p)$ returns $\mathtt{true}$ with probability $p$ and $\mathtt{false}$
with probability $1-p$ and $\mathtt{Gaussian}(\mu,\sigma)$ draws a random value from the Gaussian
distribution with mean $\mu$ and variance $\sigma$.
It can be proven that if $\mathtt{e} = 0$,
the program terminates with probability one, but if $\mathtt{e}> 0$, the program
may diverge with positive probability, no matter how small $\mathtt{e}$ is. In other words,
if the tortoise moves at a speed strictly greater than one, the hare may never catch it.

The above program is a simple forward simulation, which does not use conditioning at all. However,
we may also invert the problem and ask what was the tortoise's head start given that the hare caught
the tortoise in around one minute. This could be modelled by the following probabilistic program:
%
%\begin{verbatim}
%t := Gaussian(5.0,2.0);   
%h := 0.0;
%dt := 0.05;
%time := 0.0;
%while (h < t)
%{
%  t := t + 1*dt + e;
%  if (flip(0.25*dt)) 
%    h := h + Uniform(1,7)*dt;
%  time := time + dt;
%}
%score(Gaussian_pdf(time, 10.0, 60.0));
%return t
%\end{verbatim}

\begin{verbatim}
t := Gaussian(5,2);   
h := 0.0;
time := 0.0;
while (h < t)
{
  t := t + 1 + e;
  if (flip(0.25)) 
    h := h + Gaussian(4,2);
  time := time + 1;
}
score(Gaussian_pdf(time, 10.0, 60.0));
return t
\end{verbatim}
\noindent where $\mathtt{score}$ intuitively multiplies the probability of the current program run by its argument,
and $\mathtt{Gaussian\_pdf}(\mu, \sigma, x) = \frac{1}{\sqrt{2\pi}\sigma} e^{- \frac{1}{2 \sigma^2} (\mu - x)^2}$ is the value of the density function of the Gaussian distribution
with mean $\mu$ and variance $\sigma$ at point $x$.
Now, although we assume that the hare caught the tortoise, the program may still diverge with a positive probability.
In order to reason about programs like this, we need a framework which supports soft 
conditioning---as modelled by $\mathtt{score}$ in our setting---and is able to handle diverging programs.

%TODO: Double check and improve it
As a more complicated example, let us consider the inverse intuitive physics example from \cite{probmods2}\footnote{available online under \url{http://probmods.org/chapters/conditioning.html}.}.
In this model, using noisy approximate Newtonian dynamics, a ball is falling on the ground from a certain height, potentially hitting some fixed obstacles on the 
way. Given the observed final position of the ball, we want to find the distribution on initial locations of the ball. Similarly
to the above example, this model is implemented as a simulation of the ball's movement from the
random initial position (sampled from the prior), followed by soft conditioning on the ball's final position. Depending on
the shapes and locations of the obstacles and the size of the floor on which the ball is supposed to land,
the program may not terminate---the ball may get stuck in the air, blocked by obstacles, or may fail to land on the floor
and keep falling indefinitely.

%THE ALTERNATIVE MODEL:
%// makes a floor with evenly spaced buckets
%var bins = function (xmin, xmax, width) {
%  return ((xmax < xmin + width)
%          // floor
%          ? {shape: 'rect', static: true, dims: [400, 10], x: 175, y: 500}
%          // bins
%          : [{shape: 'rect', static: true, dims: [1, 10], x: xmin, y: 490}].concat(bins(xmin + width, xmax, width))
%         )
%}
%
%// add three fixed circles
%var world = [{shape: 'circle', static: true, dims: [60], x: 60, y: 200},
%             {shape: 'circle', static: true, dims: [40], x: 170, y: 200},
%             {shape: 'circle', static: true, dims: [30], x: 300, y: 300}].concat(bins(-1000, 1000, 25))
%
%var randomBlock = function () {
%  return {shape: 'circle', static: false, dims: [10], x: uniform(0, worldWidth), y: 0}
%}
%
%physics.animate(1000, [randomBlock()].concat(world))

Issues with program divergence may also appear when implementing models which are not designed to be possibly diverging---be it because of mistakes in the implementation or intricacies and subtleties of the model itself.
For instance, the implementation of the Pitman-Yor process~\cite{Ishwaran01} on \url{forestdb.org}, an online repository of probabilistic models in Church \cite{DBLP:conf/uai/GoodmanMRBT08} and WebPPL \cite{dippl}, occasionally fails.
According to a note on the website, a possible cause is that the program may not almost surely terminate.
%which maintainers of the website attribute to the fact that the implementation possibly does not terminate with 
%probability 1.

Another issue related to program divergence is that some implementations of sampling-based algorithms do not handle detected divergence correctly---instead of throwing an appropriate error message, they simply ignore diverging runs after a given number of steps, which leads to misleading inference results. 
For instance, consider the following WebPPL program taken from~\cite{DBLP:conf/lics/OlmedoKKM16}:
\begin{verbatim}
var three_calls = function () {
  if (flip(0.5)) {
     return 0;
  }
  else {
    return 1 + three_calls() + three_calls() + three_calls();
  }
}
var model = function () {three_calls()}
\end{verbatim}
This program does not almost surely terminate and its expected outcome is infinite. 
However, applying WebPPL's enumeration (exact) inference algorithm with a bounded maximum number of executions to this program gives a distribution assigning a probability of over 0.8 to outcome 0 and minuscule probabilities to other outcomes. 
No warning about the maximum number of evaluation steps being reached is given.

%diverging runs are quietly ignored, as the execution is finally cut off without communicating the issue to the 
%user. This may lead to the inference engine 

%Current research on non-terminating programs
There has been research on the semantics of non-terminating probabilistic programs \cite{Morgan96, Morgan96b,
DBLP:journals/pe/GretzKM14}.
However, this research was mostly aimed at analysing randomised algorithms, rather than Bayesian inference. As
a consequence, most languages used in this line of research have no features such as  continuous distributions
and soft conditioning, which are the cornerstone of Bayesian probabilistic programming. While some authors
consider non-terminating programs in the context of Bayesian reasoning~\cite{Jansen15, Olmedo18, %Batz18-no divergence here!, 
KaminskiPhD}, they normally
restrict their attention to discrete programs with hard conditioning by means of Boolean predicates.
So far, to our knowledge, the only work which comprehensively treats non-termination in the context
of semantics of Bayesian probabilistic programming with continuous distributions is \cite{BichselGV18}.
This paper defines a semantics which calculates the probability of divergence and the probability of
failing a hard constraint explicitly. Soft constraints are not considered for diverging programs, as
the authors argue that the probability of divergence normalised by soft constraints may be undefined
for some programs if unbounded scores are allowed. The authors do not attempt to restrict the
language so that scores would make sense for diverging programs.

%TODO explain what we do diferently!

%In this paper, we discuss how the possibility of non-termination affects the semantics of Bayesian 
%probabilistic programs. We aim to define the first semantics of a probabilistic language supporting both continuous which is designed to handle diverging programs correctly.

%In this paper, we aim to define the first semantics of a probabilistic language supporting both continuous distributions
%and hard and soft conditioning which is designed to handle diverging programs correctly. 

%In this paper, we investigate how the possibility of divergence affects the semantics of Bayesian probabilistic
%programs, or, equivalently, how the addition of 

In this paper, we investigate how the addition of continuous distributions and soft conditioning, necessary for
most  machine learning applications, affects the semantics of potentially diverging procedural probabilistic programs.
We discuss why
dealing with divergence in programs with soft conditioning is very difficult (if at all possible) and why one cannot
expect any sampling-based semantics to fully correspond to the intuitive meaning of a potentially diverging program.
Nevertheless, we also aim to define the first semantics of a probabilistic language supporting both continuous distributions
and hard and soft conditioning which is designed to handle diverging programs. We discuss the strengths
and limitations of this semantics and state in what sense it can be considered correct.

We provide both a denotational weakest preexpectation semantics \`a la Kozen \cite{DBLP:journals/jcss/Kozen81} and McIver and Morgan~\cite{Morgan96b} %??
together with an operational sampling-based semantics, and prove that the two semantics are equivalent. 
Hence, \emph{this paper extends the standard weakest preexpectation framework to programs with continuous distributions and soft conditioning while being able to treat program divergence.}

\section{A Bayesian probabilistic while-language} \label{section:language}

\newcommand{\ccpgclns}{\texttt{PL}}
\newcommand{\ccpgcl}{\ccpgclns\ }

We start off by presenting the syntax of a simple probabilistic while-language, simply called \ccpgclns, which will be used throughout this paper.
Besides the usual ingredients such as skip and diverge statements, assignments, sequential composition, conditional statements and guarded loops, the language contains three additional constructs: (a) \emph{random draws} from \emph{continuous distributions}, (b) \emph{observations} encoding hard conditioning, and (c) a $\emph{score}$ function used for soft conditioning. These forms of conditioning are central to Bayesian inference.
\begin{figure}[htb]
\centering
\begin{minipage}{10cm}
\begin{grammar}
<$C$>  ::=
%\alt 
$\mathtt{skip}$ \hfill no-operation
\alt $\mathtt{diverge}$ \hfill divergence
\alt $x := E$ \hfill variable assignment
\alt $x :\approx U$ \hfill random variable assignment
\alt $\mathtt{observe}(\phi)$ \hfill hard conditioning
%\alt $\mathtt{for}\ i\ \mathtt{in}\ 1..x \{C\}$
\alt $\mathtt{score}(E)$ \hfill soft conditioning
\alt $C_1;C_2$ \hfill sequential composition
\alt $\mathtt{if}(\phi)\{ C\}$ \hfill conditional
\alt $\mathtt{while}(\phi)\{C\}$ \hfill guarded loop
\end{grammar}
\end{minipage}
\caption{Syntax of \ccpgcl}
\label{figure:syntax-ccpgcl}
\end{figure}

The syntax is presented in Fig.~\ref{figure:syntax-ccpgcl} where $C$, $C_1$, and $C_2$ are programs, $x$ is a program variable, $U$ denotes the continuous uniform distribution on the unit interval, $\phi$ is a predicate over the program variables, and $E$ is an arithmetic expression over the program variables.
We do not specify the syntax of expressions $E$ and predicates $\phi$---we assume these may
be arbitrary, as long as the corresponding evaluation functions are measurable (as explained later).

A few remarks concerning the syntax are in order.
In order to simplify the approximation of while loops (as used later), we use the $\mathtt{if}$ operator without an $\mathtt{else}$ clause.
%% rather than the usual $\mathtt{if}(\phi)\{ C_1\} \mathtt{else}\{C_2\}$. 
This does not change the expressiveness of the language.
For the same reason, the explicit $\mathtt{diverge}$ statement is used as syntactic sugar for $\mathtt{while}(\mathtt{true})\{\mathtt{skip}\}$.
In random assignments, we only allow sampling from the uniform distribution $U$ on the unit interval $[0,1]$. 
This does not limit the expressiveness of the language, as samples from an arbitrary continuous distribution can be obtained by sampling from the unit interval and applying the inverse cumulative distribution function (inverse $\mathtt{cdf}$) of the given distribution to the generated sample.
%
%For instance, in order to draw a sample from the Gaussian distribution with mean $\mathtt{mu}$ and variance $\mathtt{sigma}$, we can do the following:
For instance, we can generate a sample from the Gaussian distribution with mean $\mathtt{mu}$ and variance $\mathtt{sigma}$ as follows 
\begin{verbatim}
u := U;
x := Gaussian_inv_cdf(mu,sigma,u);
\end{verbatim}
\noindent where $\mathtt{Gaussian\_inv\_cdf}(\mu, \sigma, u)$ returns the value of the inverse cumulative
distribution function of the Gaussian distribution with mean $\mu$ and variance $\sigma$ at point $u$---in other words,
$\mathtt{Gaussian\_inv\_cdf}(\mu, \sigma, u)$ is a value $v \in \mathbb{R}$ such that 
$\int_{-\infty}^{v} \mathtt{Gaussian\_pdf}(\mu, \sigma, x)\, dx = u$ 
\footnote{Note that the value of $\mathtt{Gaussian\_inv\_cdf}(\mu, \sigma, u)$ is technically only defined
for $u \in (0,1)$, but we can safely extend it to $[0,1]$ by setting $\mathtt{Gaussian\_inv\_cdf}(\mu, \sigma, 0)$ 
and $\mathtt{Gaussian\_inv\_cdf}(\mu, \sigma, 1)$ to some arbitrary value (say, $0$), as the probability of
drawing $0$ or $1$ from the continuous uniform distribution on $[0,1]$ is zero, anyway.}. 

Random draws from discrete probability distributions can also be encoded by uniform draws from the unit interval, see e.g. \citep{park08sampling}.
%For instance, the statement $\mathtt{flip(0.25*dt)}$ as used in the introduction is a shorthand for $\ldots$ \marginpar{TBC}.
For instance, the statement $\mathtt{if}(\mathtt{flip(0.25)})\{C\}$ as used in the introduction is a shorthand for
$\mathtt{u} :\approx \mathtt{U}; \mathtt{if}(\mathtt{u < 0.25})\{C\}$.

Let us briefly describe the semantics of the three new syntactic constructs at an intuitive level; the rest of this paper is devoted to make this precise.
The execution of the random variable assignment $x :\approx U$ incorporates taking a sample from the uniform distribution $U$ and assigning this sample to the program variable $x$.
The $\mathtt{observe}(\phi)$ statement is similar to the \texttt{assert}$(\phi)$ statement: it has no effect for program runs satisfying the predicate $\phi$, but program runs violating $\phi$ are invalid.
Such invalid runs are discontinued (aka: stopped).
The crucial difference to the \texttt{assert} statement is that probabilities of valid program runs are normalised with respect to the total probability mass of all valid runs. 
For instance, the only valid runs of program
\begin{verbatim}
x := 0; y := 0;
if (flip(0.5)) 
    x := 1;
if (flip(0.5))
    y := 1;
observe(x+y=1)
\end{verbatim}
are $\mathtt{x}{=}0, \mathtt{y}{=}1$ and $\mathtt{x}{=}1, \mathtt{y}{=}0$.
Although in absence of the \texttt{observe}-statement the probability of each such run is $\nicefrac 1 4$, their probability now becomes $\nicefrac 1 2$ due to normalising $\nicefrac 1 4$ with the probability of obtaining a valid run, i.e., $\nicefrac 1 2$.
(As discussed extensively in~\cite{Olmedo18}, the semantics becomes more tricky when program divergences are taken into account.)
As runs are abandoned that violate the predicate $\phi$, this is called \emph{hard} conditioning.

In contrast, the statement $\mathtt{score}(E)$ models \emph{soft} conditioning.
As effect of executing this statement the probability of the current program run is scaled (i.e. multiplied) by the current value of the expression $E$.
The higher the value of $E$, the more likely the combination of random variables sampled so far is considered to be.
%
%For instance, consider the following program:

To illustrate how soft conditioning works, suppose that we have a function $\mathtt{softeq}(a,b) = e^{-(a-b)^2}$,
whose value is $1$ if both arguments are the same and moves closer to $0$ as the arguments move further apart.
%
%Suppose also that $\mathtt{Gaussian\_inv\_cdf}(\mu, \sigma, u)$ returns the value of the inverse cumulative
%distribution function of the Gaussian distribution with mean $\mu$ and variance $\sigma$ at point $u$---in other words,
%$\mathtt{Gaussian\_inv\_cdf}(\mu, \sigma, u)$ is a value $v \in \mathbb{R}$ such that 
%$\int_{-\infty}^{v} \mathtt{Gaussian\_pdf}(\mu, \sigma, x)\, dx = u$. 
%
Now, consider the following program:
\begin{verbatim}
u1 := U;
x := Gaussian_inv_cdf(0,2,u1);
u2 := U;
y := Gaussian_inv_cdf(1,2,u2);
score(softeq(x,y));
\end{verbatim}
%where $\mathtt{softeq}$ is a function defined as $\mathtt{softeq}(a,b) = e^{-(a-b)^2}$,
%whose value is $1$ if both arguments are the same and moves closer to $0$ as the arguments move further apart.
The use of $\mathtt{score}$ has the effect that program runs in which $\mathtt{x}$ and $\mathtt{y}$ are closer to each other
are more likely.

%\section{Weakest preexpectation semantics}
\section{Denotational semantics} \label{section:denotational-semantics}
We will now define the semantics of \ccpgcl in a weakest precondition style manner. 
This semantics builds upon the semantics of the probabilistic guarded command language $\mathtt{pGCL}$~\cite{MM05:AbstractionRefinementProofForProbabilisticSystems} extended with hard conditioning as defined in~\cite{Olmedo18}.
The key object $\mathtt{wp}|[C |](f)(\sigma)$ defines the expected value of a function $f$ with respect to the probability distribution of final states of program $C$, provided the program starts in the initial state $\sigma$. 
The key difference to~\cite{MM05:AbstractionRefinementProofForProbabilisticSystems,Olmedo18} is that dealing with continuous distributions requires some sort of integration, and the integrated functions must be well-behaved.

Defining a denotational semantics of a language allowing unbounded computations requires the use of domain theory,
which helps to ensure that all semantic functions used are well-defined. Some basic definitions from domain theory,
which are needed to understand this paper, are included in Appendix~\ref{appendix:domain-theory}.

Probability theory with continuous random variables is usually formalised using measure theory and the semantics of \ccpgcl follows this route.
For the sake of completeness, the main relevant ingredients of measure theory are summarised in Appendix~\ref{app:basics-measure-theory}.
We start off by defining a measurable state space, and the domain of measurable expectations---the quantitative analogue of predicates.
After shortly defining the (standard) semantics of expressions and predicates, we define a weakest preexpectation semantics of \ccpgcl and subsequently generalise this towards a weakest liberal preexpectation semantics that takes program divergence explicitly into account.

\subsection{Measurable space of states}

In the same vein as~\cite{MM05:AbstractionRefinementProofForProbabilisticSystems,Olmedo18}, the semantics $\mathtt{wp}|[C |](f)(\sigma)$ will be defined as the expected value of the measurable
function $f$ mapping states to nonnegative reals (extended with $\infty$). In order to reason about measurable functions on states, we first define a measurable space of program states.

\newcommand{\failure}{\lightning}
\newcommand{\diverge}{\uparrow}

\newcommand{\statespace}{\Omega_{\sigma}}
\newcommand{\statesa}{\Sigma_{\sigma}}

\newcommand{\fullstatespace}{\hat{\Omega}_{\sigma}}
\newcommand{\fullstatesa}{\hat{\Sigma}_{\sigma}}

\newcommand{\extposreals}{\overline{\mathbb{R}}_{+}}

Let $\mathcal{N}$ be a countable set of variable names ranged over by $x_i$.
A program state maps program variables to their current value.
Formally, state $\sigma$ is a set $\{(x_1, V_1), \dots, (x_n, V_n)\}$ of pairs of unique variable names $x_i$ and their corresponding values $V_i \in \mathbb{R}$.
The set $\statespace$ has the following form:
\[
\Omega_{\sigma} = \biguplus_{n \in \mathbb{N}} \left ( \left \{ \{(x_1, V_1), \dots, (x_n, V_n)\} \ |\ 
\forall i \in 1..n\ x_i \in \mathcal{N}, V_i \in \mathbb{R}. \forall j \neq i \  x_i \neq x_j\right \} \right )
%x_1, \dots, x_n \in \mathcal{N}, V_1, \dots, V_n \in \mathbb{R} \uplus \mathbb{Z}, \forall i,j \in 1..n \ i \neq j =>   x_i \neq x_j \right \} \right )
\]
The state space $\statespace$ is equipped with the functions:
$\mathtt{dom}(\cdot) \colon \statespace -> P(\mathcal{N})$, 
returning the domain of a state (i.e., the set of variables which are assigned values), and
$\mathtt{elem}(\cdot ,\cdot) \colon \statespace \times \mathcal{N} -> \mathbb{R} \uplus \{ \bot \}$
such that $\mathtt{elem}(\sigma, x)$ (for convenience, abbreviated $\sigma(x)$)
returns the value assigned to variable $x$ in $\sigma$ or $\bot$ if $x \notin \mathtt{dom}(\sigma)$.
%
%Possible concrete instantiations of $\mathcal{N}$ and $\statespace$ and corresponding functions include:
%
% $\mathcal{N}$ being an arbitrary countable set and each element of $\statespace$ being represented by a set of 
%pairs of (distinct) variable names and values. Formally:
The functions $\mathtt{dom}$ and $\mathtt{elem}$ are defined for $\sigma = \{(x_1, V_1), \dots, (x_n, V_n)\}$ as:
\[
\begin{array}{rcl}
\mathtt{dom}(\sigma) & = & \{ x_1, \dots, x_n \} \\[1ex]
\mathtt{elem}(\sigma,y) \ = \ \sigma(y) & = & 
\begin{cases}
V_i & \text{if}\ y=x_i \text{ for some } i \\
\bot & \text{otherwise.}
\end{cases}
\end{array}
\]
%%
%This comes later, when we define the operational semantics:
%We subsequently define $\fullstatespace$ to be the set of \emph{all} states---that is, 
%$\fullstatespace = \statespace \uplus \{\failure, \diverge \}$.
Let the metric $d_\sigma$ on $\statespace$ be defined as follows:
\[
d_\sigma(\sigma_1, \sigma_2) =
\begin{cases}
%0 & \text{if}\ \sigma_1(k) = \sigma_2(k) = ()\\
\sum_{x \in \mathtt{dom}(\sigma_1)} |\sigma_1(x) - \sigma_2(x)| & \text{if}\ \mathtt{dom}(\sigma_1) = \mathtt{dom}(\sigma_2) \\
%\infty & \text{if}\ \sigma_1(k) = ()\ \text{xor}\ \sigma_2(k) = () \text{ or types do not match}\\
%d(\sigma_1(k), \sigma_2(k)) & \text{otherwise.}
\infty & \text{otherwise}
\end{cases}
\]
%%
%\[
%d_\sigma(\sigma_1, \sigma_2) =
%\begin{cases}
%%0 & \text{if}\ \sigma_1(k) = \sigma_2(k) = ()\\
%\Sigma_{x \in \mathtt{dom}(\sigma_1)}  d_T(\sigma_1(x), \sigma_2(x)) & \text{if}\ \mathtt{dom}(\sigma_1) = \mathtt{dom}(\sigma_2) \\
%%\infty & \text{if}\ \sigma_1(k) = ()\ \text{xor}\ \sigma_2(k) = () \text{ or types do not match}\\
%%d(\sigma_1(k), \sigma_2(k)) & \text{otherwise}
%\infty & \text{otherwise}
%\end{cases}
%\]
%
%\noindent where
%\[
%d_T(V_1, V_2) = 
%\begin{cases}
%0 & \text{if } V_1 = V_2 = () \\
%%d_{\mathbb}{R}(V_1, V_2)  
%|V_1 - V_2| & \text{if } V_1, V_2 \in \mathbb{R} \\
%%|V_1 - V_2| & \text{if } V_1, V_2 \in \mathbb{Z} \\
%\infty & \text{otherwise}
%\end{cases}
%\]
%%
It is easy to verify that $d_\sigma$ is indeed a metric. 
Note that on the subset of states with a fixed domain $\{x_1, \dots, x_n\}$, $d_\sigma$ is essentially the Manhattan distance.

%We extend this metric to $\fullstatespace$ by defining:
%
%\[
%\hat{d}_\sigma(\sigma_1, \sigma_2) =
%\begin{cases}
%0 & \text{if } \sigma_1 = \sigma_2  = \failure \\
%0 & \text{if } \sigma_1 = \sigma_2 = \diverge \\
%d_\sigma(\sigma_1, \sigma_2)& \text{if } \sigma_1, \sigma_2 \in \statespace \\
%\infty & \text{otherwise}
%\end{cases}
%\]

\begin{lemma}
The metric space $(\statespace, d_\sigma)$ is separable.
\end{lemma}
\begin{proof}
%We need to find a countable dense subset of $\fullstatespace$ w.r.t. $\hat{d}_\sigma$.
Consider a subset $\statespace^{\mathbb{Q}}$ of $\statespace$ where all values are rational. Then
the set $\statespace^{\mathbb{Q}}$ is countable and it can be easily verified that it is a dense subset of $\statespace$. Hence,
$(\statespace,  d_\sigma)$ is separable.
 \qed \end{proof}
%\marginpar{what is $\hat{d}_\sigma?$}

Finally, let $\statesa$ be the Borel $\sigma$-algebra on $\statespace$ induced by the metric $d_\sigma$.
The pair $(\statespace, \statesa)$ is our measurable space of states.

%\paragraph{Note about $\statespace$ (second definition)} For a fixed domain $D = \{x_1, \dots, x_n\}$, we can define
%the set of states in this domain as $\statespace^D  =\{\{(x_1, V_1), \dots, 
%(x_n, V_n)\}\ |\ V_1, \dots, V_n \in \mathbb{R} \uplus \mathbb{Z} \}$.
%Then the full set of proper states is $\statespace = \biguplus_D \statespace^D$.
%Note that the set of all domains is countable. The metric space $(\statespace, d_\sigma)$ is a disjoint union of
%spaces $(\statespace^D, d_\sigma)$.

\subsection{Domain of measurable expectations}
As the weakest preexpectation semantics of \ccpgcl is defined in terms of an operator transforming measurable functions, we need to show that measurable functions from $\statespace$ to $\extposreals = \mathbb{R}_{+} \cup \{+\infty\}$ form a valid domain. More specifically,
these functions must form a \emph{$\omega$-complete partial order} (whose definition is included in Appendix~\ref{appendix:domain-theory}).
Similarly, we need to show that the domain of \emph{bounded} measurable expectations,
which will be used in the weakest liberal preexpectation semantics, is valid.
Fortunately, these facts follow immediately from basic properties of measure theory.

\begin{lemma} \label{lemma:meas-cpo}
The set of measurable functions $f\colon \statespace -> \extposreals$ with point-wise ordering
forms an $\omega$-complete partial order ($\omega$-cpo). Similarly, the set of bounded measurable functions
$f \colon \statespace -> [0,1]$ forms an $\omega$-cpo.
\end{lemma}
\begin{proof}
The bottom element of the set of measurable functions $f\colon \statespace -> \extposreals$ is the function
$\lambda \sigma . 0$, mapping every state to $0$. It is known that any increasing chain of functions 
with co-domain $\extposreals$ has a supremum,  so this also holds for chains of measurable functions.
The fact that point-wise supremum of measurable functions to $\extposreals$ is measurable is a standard
result in measure theory. %TODO: add citation?
The argument for bounded measurable functions is the same.
 \qed \end{proof}

\subsection{Expression and predicate evaluation}
%We do not specify the syntax of expressions $E$ and predicates $\phi$---we assume these may
%be arbitrary, as long as the corresponding evaluation functions are measurable (as explained later).\

The semantics makes use of two evaluation functions, $\sigma(E)$ and $\sigma(\phi)$,
which evaluate the real-valued expression $E$ and predicate $\phi$, respectively, in state $\sigma$.
We assume that for each $E$, the evaluation function on states $\lambda (\sigma, E) . \sigma(E)$ %: \statespace -> \mathbb{R}$ 
is measurable
and, similarly, for all $\phi$, the function $\lambda (\sigma, \phi) . \sigma(\phi)$ %: \statespace -> \{\mathtt{true}, \mathtt{false}\}$ 
is 
measurable\footnote{This assumption requires a $\sigma$-algebra on expressions and predicates. This can be defined as a Borel
$\sigma$-algebra induced by a simple metric on syntactic terms, as in~\cite{DBLP:conf/icfp/BorgstromLGS16}.}.
We also assume that the evaluation functions are total---this means that in case of evaluation errors, such
as some variable in $E$ not being in the domain of $\sigma$, some value (typically $0$ or $\mathtt{false}$) still
needs to be returned. We convert truth values to reals by Iverson brackets $[\cdot]$: $[\mathtt{true}] = 1$ and $[\mathtt{false}] = 0$.
We write $E$ for $\lambda \sigma . \sigma(E)$ and $\phi$ for $\lambda \sigma . \sigma(\phi)$ if it is clear from the context
that $E$ or $\phi$ denotes a function.

\subsection{Weakest preexpectation semantics}

We now have all ingredients in place to define the weakest preexpectation semantics of \ccpgclns. 
This semantics is defined by the operator $\mathtt{wp}|[C |](\cdot)$, which takes a measurable function $f$ from $\statespace$ to $\extposreals$---called the \emph{postexpectation}---and returns a measurable function in the same domain---called the \emph{preexpectation})---which, for every initial state $\sigma_0$, computes the expected value of $f$ after executing the program $C$ starting in state $\sigma_0$. 
In other words, if $f \colon \statespace -> \extposreals$ is a measurable function on states and $\sigma_0 \in \statespace$ is the initial state, then $\mathtt{wp}|[C |](f)(\sigma_0)$ yields the expected value of $f(\sigma)$, where $\sigma$ is a final program state of $C$.

\begin{figure}[t]
\centering
%\begin{minipage}{2.3cm}
\begin{eqnarray*}
\mathtt{wp}|[ \mathtt{skip} |](f) &=& f \\[1ex]
\mathtt{wp}|[ \mathtt{diverge} |](f) &=& 0 \\[1ex]
\mathtt{wp}|[ x := E |](f) &=& \lambda \sigma . f(\sigma[x \mapsto \sigma(E)]) \\[1ex]
\mathtt{wp}|[ x :\approx U |](f) & \ = \ & \lambda \sigma . \int_{[0,1]} 
%(\lambda v. f(\sigma[x \mapsto v ]) ) \, \lambda(dv)\\
f(\sigma[x \mapsto v ])  \, \lambda(dv)\\[1ex]
\mathtt{wp}|[ \mathtt{observe}(\phi)|](f)  &=&\lambda \sigma . [\sigma(\phi)]{\cdot}f(\sigma)\\[1ex]
\mathtt{wp}|[ \mathtt{score}(E) |](f) &=& \lambda \sigma . [\sigma(E) \in (0,1]]{\cdot}\sigma(E){\cdot} f(\sigma)\\[1ex]
\mathtt{wp}|[ C_1;C_2 |](f) &=& \mathtt{wp}|[ C_1|](\mathtt{wp}|[C_2 |](f))\\[1ex]
\mathtt{wp}|[\mathtt{if}(\phi)\{ C\} |](f) &=& [\phi]{\cdot}\mathtt{wp}|[C|](f) + [\neg \phi]{\cdot}f \\[1ex]
\mathtt{wp}|[ \mathtt{while}(\phi)\{C\} |](f) &=& \mathtt{lfp}\ X . [\neg \phi]{\cdot}f + [\phi]{\cdot} \mathtt{wp}|[C|](X)
\end{eqnarray*}
%\end{minipage}
\caption{Weakest preexpectation semantics of \ccpgcl}
 \label{figure:wp-ccpgcl}
\end{figure}

The wp-semantics of \ccpgcl is defined by structural induction and is shown in Fig.~\ref{figure:wp-ccpgcl}. 
The semantics of most constructs matches the wp-semantics in \cite{Olmedo18}, with the distinction that it is defined on the domain of nonnegative measurable functions on states, rather than arbitrary nonnegative functions. 
Let us briefly explain the individual cases one by one.
\begin{description}
\item[Skip.]
The $\mathtt{skip}$ statement leaves the expectation $f$ unchanged.
\item[Divergence.]
The expectation of any function $f$ with respect to the $\mathtt{diverge}$ expression is $0$, as no final state at which $f$ can be evaluated is ever reached by the program.
\item[Assignment.]
For assignment $x := E$, the semantics just evaluates $E$, updates $x$ with the new value
in the state and passes this updated state to the expectation.
\item[Random draw.]
The expected value of a measurable function $f$ on states with respect to the uniform random assignment $x :\approx U$, applied to the initial state $\sigma$, is the Lebesgue integral of $f(\sigma[x \mapsto v])$ (as a function of $v$) with respect to the Lebesgue measure $\mu_L$ on $[0,1]$\footnote{
%$\lambda \sigma$ denotes a function with formal parameter $\sigma$, whereas $\lambda(dv)$ denotes the Lebesgue measure on $[0,1]$.
The Lebesgue measure is usually denoted by $\lambda$ in the literature. We write $\mu_L$ instead to avoid confusion with 
the use of $\lambda \sigma$ to define a function with formal parameter $\sigma$.
}.
By the Fubini-Tonelli theorem, $\mathtt{wp}|[ x :\approx U |](f)$ is itself a measurable function.
\item[Hard.]
The $\mathtt{observe}$ statement defines hard conditioning---it states that all runs of the program which do not satisfy $\phi$ should be discarded and should not affect the expectation of $f$.
\item[Soft.]
Scoring multiplies the expectation by the argument to $\mathtt{score}$, expected to evaluate
to a number in $(0,1]$. 
\item[Sequencing.]
The semantics of a sequence $C_1;C_2$ of two commands is just the composition of the semantics of respective commands--- the semantics of $C_2$ with respect to the given input function $f$ is the input to the semantics of~$C_1$.
\item[Conditional.]
The semantics of an $\mathtt{if}(\phi)\{C\}$-expression is, for initial states satisfying the condition $\phi$, the semantics of the body $C$. For other states, the semantics is equivalent to the $\mathtt{skip}$ statement, as the expression does not do anything.
%
%As $\mathtt{diverge}$ is essentially an abbreviation of $\mathtt{while}(\mathtt{true})\{\mathtt{skip}\}$
%
\item[Loops.]
The semantics of a $\mathtt{while}$-loop is defined as the least fixpoint of a function which
simply returns the input continuation $f$ if $\phi$ is false (corresponding to exiting the loop)
and applies the semantics of the body to the argument $X$ otherwise (which corresponds to 
performing another iteration). 
As explained in the following paragraph (cf. Lemma~\ref{lemma:while-kleene}), %TODO: Where?
this has the desired effect that the semantics of a $\mathtt{while}$-loop
is equivalent to the semantics of the infinite unfolding of the loop.
\end{description}

%$\mathtt{wp}|[C |](f)(\sigma_0)$

\paragraph{Well-definedness and key properties of $\mathtt{wp}$.}
When defining the weakest preexpectation semantics of \ccpgclns, we implicitly assumed
that all mathematical objects used are well defined. Specifically, we assumed that
the $\mathtt{wp}$ transformer preserves measurability and that the least fixpoint in the semantics
of $\mathtt{while}$-loops exists. These properties can be proven by structural induction on
the program $C$. 
%TODO: do we need the following?????
The key observations used in the proof (which would not be needed in the discrete case) are
that $\lambda \sigma . \int_{[0,1]} (\sup_i f_i)(\sigma[x \mapsto v]) \,\mu_L(dv)
= \lambda \sigma .\ \sup_i \int_{[0,1]} f_i(\sigma[x \mapsto v]) \,\mu_L(dv)$
by Beppo Levi's theorem and that 
$ \lambda \sigma . \int_{[0,1]}  f(\sigma[x \mapsto v ])  \, \mu_L(dv)$ is measurable
(as a function of $\sigma$) by the Fubini-Tonelli theorem.
As both $\omega$-continuity (as defined in Appendix~\ref{appendix:domain-theory}) and measurability are required for $\mathtt{wp}|[C|]$ to be well defined,
we need to prove both properties simultaneously, so that the induction hypothesis
is strong enough.

\begin{lemma} \label{lemma:wp-continuous-measurable}
For every program $C$:
\begin{enumerate}
\item 
the function $\mathtt{wp}|[C|](\cdot)$ is $\omega$-continuous, and
\item
for every measurable $f \colon \statespace -> \extposreals$, $\mathtt{wp}|[C|](f)(\cdot)$ is measurable.
\end{enumerate}
\end{lemma}
\begin{proof}
By induction on the structure of $C$. %% Details in the long version.
 \qed \end{proof}
The continuity of $\mathtt{wp}|[C|]$ also ensures that the expression $ [\neg \phi]{\cdot}f + [\phi]{\cdot}\mathtt{wp}|[C|](X)$ in the semantics of $\mathtt{while}$ loops is continuous as a function of $X$. Applying Kleene's Fixpoint Theorem immediately gives us the following result:
\begin{lemma} \label{lemma:while-kleene}
Let $f \colon \statespace -> \extposreals$ be measurable, $C$ be a \ccpgcl program and $\phi$ be a predicate.
Let ${}^{\mathtt{wp}}_{\langle \phi, C \rangle} \Phi_f(X) = [\neg \phi]{\cdot}f + [\phi]{\cdot}\mathtt{wp}|[C|](X)$.
Then $\mathtt{lfp}\ X . {}^{\mathtt{wp}}_{\langle \phi, C \rangle} \Phi_f(X)$ exists and is equal to
$\sup_n {}^{\mathtt{wp}}_{\langle \phi, C \rangle} \Phi_f^n(0)$. Thus,
$\mathtt{wp}|[ \mathtt{while}(\phi)\{C\} |](f)$ exists and
\[
\mathtt{wp}|[ \mathtt{while}(\phi)\{C\} |](f) = \sup_n {}^{\mathtt{wp}}_{\langle \phi, C \rangle} \Phi_f^n(0)
\]

%Then the expression: 
%\[ \mathtt{lfp}\ X . [\neg \phi]{\cdot}f + [\phi]{\cdot}\mathtt{wp}|[C|](X)\] 
%\noindent exists, and is equal to
%$\sup_n {}^{\mathtt{wp}}_{\langle \phi, C \rangle} \Phi_f^n(0)$, where
%\[ {}^{\mathtt{wp}}_{\langle \phi, C \rangle} \Phi_f(X) = [\neg \phi]{\cdot}f + [\phi]{\cdot}\mathtt{wp}|[C|](X)\].

%
%\[
%\mathtt{wp}|[ \mathtt{while}(\phi)\{C\} |](f) =
%\mathtt{lfp}\ X . [\neg \phi]{\cdot}f + [\phi]{\cdot}\mathtt{wp}|[C|](X)= \sup_n {}^{\mathtt{wp}}_{\langle \phi, C \rangle} \Phi_f^n(0)
%\]
%\noindent where ${}^{\mathtt{wp}}_{\langle \phi, C \rangle} \Phi_f(X) = [\neg \phi]{\cdot}f + [\phi]{\cdot}\mathtt{wp}|[C|](X)$.
\end{lemma}

\subsection{Examples} 

Having defined the weakest preexpectation semantics, we explain it using a few examples. We first
introduce two simple examples, which illustrate the key concepts, and then show how the semantics can be applied
to the tortoise and hare program from the introduction.

\paragraph{Notation.}  To distinguish between program variables and metavariables, we write the former in fixed-width
font (such as $\mathtt{x1}$) and the latter in the usual italic form (such as $x_1$). 
In functions where only the original, non-updated input state appears in the body, we sometimes make the state 
implicit by removing ``$\lambda \sigma .$'' and replacing variable lookups of the form $\sigma(\mathtt{x})$ by
just variables. For instance we write $\lambda \sigma . \sigma(\mathtt{x}) + \sigma(\mathtt{y})$ simply as
$\mathtt{x} + \mathtt{y}$.

\begin{example} \label{example:blr}
%TODO: very simple loop-free example

Let us first consider a very simple instance of Bayesian linear regression. We want to fit  a linear function approximately to two
points $(0,2)$ and $(1,3)$, assuming that the coefficients of the function have Gaussian prior distributions.
A \ccpgcl implementation of such a regression, using the $\mathtt{softeq}$ distance squashing function mentioned
at the end of Section~\ref{section:language}, has the following form:
\begin{verbatim}
u1 := U;
a := Gaussian_inv_cdf(0,2,u1);
u2 := U;
b := Gaussian_inv_cdf(0,2,u2);
score(softeq(a*0 + b, 2));
score(softeq(a*1 + b, 3));
\end{verbatim}

Let us now suppose that we want to calculate the expected value of the square of the$\mathtt{a}$ coefficient
(recall that we can only compute expectations of nonnegative functions). We can do that
by computing the weakest preexpectation of $\lambda \sigma . \sigma(\mathtt{a})^2$ (written $\mathtt{a}^2$ in short) with respect
to the above program --- that is, $\mathtt{wp}|[ C |](\mathtt{a}^2)$, where $C$ is the given program. In the following
derivation (as well as subsequent examples in this chapter), we adopt the notation used in \cite{KaminskiPhD},
where the function directly below a statement $C$ is a postexpectation, and the function directly above $C$ is the
corresponding preexpectation. That is, a block of the form:
\begin{eqnarray*}
{\color{red} //} && {\color{red}  f_1} \\
&&C\\
{\color{red} //} && {\color{red} f_2}
\end{eqnarray*}
\noindent states that $f_1 = \mathtt{wp}|[C|](f_2)$. We also use the letter $\mathtt{G}$ as an abbreviation for the Gaussian inverse cdf 
%\footnote{We integrate the Gaussian inverse over the interval $(0,1)$ rather than $[0,1]$ as it is ${-}\infty$ for 0 and $\infty$ for 1.}.
\footnote{We can integrate the Gaussian inverse over the interval $(0,1)$ instead of $[0,1]$, because the value
of the Lebesgue integral at a single point does not contribute to the result}.

%$\mathtt{softeq}(a,b) = e^{-(a-b)^2}$,
We can derive the expected value of $\mathtt{a}^2$ as shown below.
Note that the expressions between
program lines are functions on program states, written using the implicit notation explained before.

\begin{eqnarray*}
{\color{red} //} && {\color{red} \int_{(0,1)} \int_{(0,1)} e^{-( \mathtt{G}(0, 2, v_2) - 2)^2-( \mathtt{G}(0, 2, v_1) + \mathtt{G}(0, 2, v_2) - 3)^2} {\cdot}\mathtt{G}(0, 2,v_1)^2\, \mu_L(dv_2)\mu_L(dv_1)} \\
&&\mathtt{u1 := U};\\
{\color{red} //} && {\color{red} \int_{(0,1)} e^{-( \mathtt{G}(0, 2, v_2)- 2)^2-( \mathtt{G}(0, 2,\mathtt{u1}) + \mathtt{G}(0, 2, v_2) - 3)^2} {\cdot}\mathtt{G}(0, 2,\mathtt{u1})^2\, \mu_L(dv_2)} \\
&&\mathtt{a := Gaussian\_inv\_cdf(0,2,u1)};\\
{\color{red} //} && {\color{red} \int_{(0,1)} e^{-( \mathtt{G}(0, 2, v_2) - 2)^2-(\mathtt{a} + \mathtt{G}(0, 2, v_2) - 3)^2} 
{\cdot}\mathtt{a}^2\, \mu_L(dv_2)} \\
&&\mathtt{u2 := U};\\
{\color{red} //} && {\color{red} e^{-( \mathtt{G}(0, 2, \mathtt{u2})- 2)^2 -(\mathtt{a} + \mathtt{G}(0, 2, \mathtt{u2})- 3)^2} {\cdot}\mathtt{a}^2} \\
&&\mathtt{b := Gaussian\_inv\_cdf(0,2,u2)};\\
{\color{red} //} && {\color{red}   e^{-( \mathtt{b}- 2)^2-(\mathtt{a} + \mathtt{b} - 3)^2} {\cdot}\mathtt{a}^2} \\
&&\mathtt{score(softeq(a*0 + b, 2))};\\
{\color{red} //} && {\color{red} e^{-(\mathtt{a} + \mathtt{b}- 3)^2} {\cdot}\mathtt{a}^2} \\
&&\mathtt{score(softeq(a*1 + b, 3))};\\
{\color{red} //} && {\color{red} \mathtt{a}^2}
\end{eqnarray*}

We observe that the expected value of $\mathtt{a}^2$ is independent on the initial state, which is not surprising as the program
has no free variables. For any initial state $\sigma$, the expected value $\mathtt{wp}|[ C |](\mathtt{a}^2)(\sigma)$ of $\mathtt{a}$ is:
\[
\int_{(0,1)} \int_{(0,1)} e^{-( \mathtt{G}(0, 2, v_2) - 2)^2 -( \mathtt{G}(0, 2,v_1) + \mathtt{G}(0, 2, v_2) - 3)^2}
 {\cdot}\mathtt{G}(0, 2,v_1)^2\, \mu_L(dv_2)\mu_L(dv_1).
\]

We can also represent this expression as a double integral with respect to the Gaussian probability distribution 
$\mathcal{D}_{\mathtt{G}}$ with mean $0$ and variance $2$, using the fact that a continuous probability distribution is
a pushforward of the Lebesgue measure by the inverse cdf of the given distribution:

\[
\mathtt{wp}|[ C |](\mathtt{a})(\sigma) =
\int \int e^{-( x_2 - 2)^2 -( x_1+ x_2 - 3)^2} {\cdot}x_1^2\, \mathcal{D}_{\mathtt{G}}(dx_2)\mathcal{D}_{\mathtt{G}}(dx_1).
\]

This expression can also be represented as a double integral of Gaussian densities (denoted $\mathtt{G}_{pdf}$) over $\mathbb{R}$:

\[
%% \mathtt{wp}|[ C |](\mathtt{a})(\sigma) =
\int \int e^{-( x_2 - 2)^2 -( x_1+ x_2 - 3)^2} \mathtt{G}_{pdf}(0,2,x_1) \mathtt{G}_{pdf}(0,2,x_2) {\cdot}x_1^2\, \mu_L(dx_2) \mu_L(dx_1).
\]

\end{example}

\begin{example} \label{example:diverge-wp}
Let us now consider a very simple example of a potentially diverging program with continuous variables
and soft conditioning. This example may be rather contrived and does not represent any machine learning model,
but it illustrates well how the semantics works. Take the following program $C$:
\begin{verbatim}
b := 0;
k := 0;
while (b=0)
{
  u := U;
  k := k+1;
  if(u < 1/(k+1)^2)
  {
    b := 1;
    score(k/(k+1));
  }
}
\end{verbatim}

Suppose we want to compute $\mathtt{wp}|[C|](1)$, that is, the weakest preexpectation of the
constant function $\lambda \sigma . 1$ with respect to the program $C$. This may be interpreted as the
probability that the program terminates, weighted by the scores.

%While-loop
%We begin by calculating the characteristic function of the while-loop. By plugging the predicate $\mathtt{h} < \mathtt{t}$
%%and function $\lambda \sigma . \sigma(t)$ (which we abbreviate to $t$) 
%into the formula for the characteristic function, we get 
%${}^{\mathtt{wp}}_{\langle h < t , C' \rangle} \Phi_{f}(X) = [\mathtt{h} \geq \mathtt{t}]{\cdot}f +[\mathtt{h} < \mathtt{t}]{\cdot}\mathtt{wp}|[C'|](X)$,
%where $C'$ is the body of the while loop. By Lemma~\ref{lemma:while-kleene}, we can then conclude that
%$\mathtt{wp}|[\mathtt{while}(h<t)\{C'\}|](f)
%= \sup_n\, {}^{\mathtt{wp}}_{\langle h < t , C' \rangle} \Phi^{n}_{f}(0)$.
%
%We now need to compute $\mathtt{wp}|[C'|](X)$. To simplify presentation, we assume that $X$ does not
%depend directly on variables $\mathtt{u2}$ and $\mathtt{u3}$---we can show by a simple induction that this
%holds for $X = {}^{\mathtt{wp}}_{\langle h < t , C' \rangle} \Phi^{n}_{f}(0)$ for any $n$ (if $f$ is independent 
%of $\mathtt{u2}$ and $\mathtt{u3}$), and we only need to  apply
%$\mathtt{wp}|[C'|](\cdot)$ to functions $X$ of this form.

As the program has a while-loop, we need to find the  characteristic function 
${}^{\mathtt{wp}}_{\langle b = 0, C' \rangle} \Phi_1$ of this loop (whose body we denote by $C'$), with respect to the constant
postexpectation $1$. In this case, the characteristic function is 
$$
{}^{\mathtt{wp}}_{\langle b = 0 , C' \rangle} \Phi_{1}(X) = [\mathtt{b} \neq 0] +
[\mathtt{b} = 0]{\cdot}\mathtt{wp}|[C'|](X).
$$ 
We first need to compute $\mathtt{wp}|[C'|](X)$:
\begin{eqnarray*}
{\color{red} //} && {\color{red} \lambda \sigma .   %\frac{1}{(\sigma(\mathtt{k}){+}2)^2} 
       \frac{\sigma(\mathtt{k}){+}1}{(\sigma(\mathtt{k}){+}2)^3} {\cdot} X(\sigma[k \mapsto \sigma(k){+}1][\mathtt{b} \mapsto 1])
+  \frac{(\sigma(\mathtt{k}){+}2)^2 - 1}{(\sigma(\mathtt{k}){+}2)^2} {\cdot} X(\sigma[k \mapsto \sigma(k){+}1]) }\\
&&{\color{red} =}\\
{\color{red} //} && {\color{red} \lambda \sigma .  \int_{[0,1]}  \left [v < \frac{1}{(\sigma(\mathtt{k}){+}2)^2} \right] {\cdot}
       \frac{\sigma(\mathtt{k}){+}1}{\sigma(\mathtt{k}){+}2} {\cdot} X(\sigma[k \mapsto \sigma(k){+}1][\mathtt{b} \mapsto 1])}\\
{\color{red} //} && {\color{red} \qquad + \left [v \geq \frac{1}{(\sigma(\mathtt{k}){+}2)^2} \right] {\cdot} X(\sigma[k \mapsto \sigma(k){+}1]) \, \mu_L(dv)}\\
&&\mathtt{u := U;}\\
{\color{red} //} && {\color{red} \lambda \sigma . \left [\sigma(\mathtt{u}) < \frac{1}{(\sigma(\mathtt{k}){+}2)^2} \right] {\cdot}
       \frac{\sigma(\mathtt{k}){+}1}{\sigma(\mathtt{k}){+}2}{\cdot} X(\sigma[k \mapsto \sigma(k){+}1][\mathtt{b} \mapsto 1])}\\
&&{\color{red} \qquad + \left [\sigma(\mathtt{u}) \geq \frac{1}{(\sigma(\mathtt{k}){+}2)^2} \right] {\cdot} X(\sigma[k \mapsto \sigma(k){+}1]) }\\
&&\mathtt{k := k+1;}\\
{\color{red} //} && {\color{red} \lambda \sigma . \left [\sigma(\mathtt{u}) < \frac{1}{(\sigma(\mathtt{k}){+}1)^2} \right] {\cdot}
       \frac{\sigma(\mathtt{k})}{\sigma(\mathtt{k}){+}1} {\cdot} X(\sigma[\mathtt{b} \mapsto 1])
+ \left [\sigma(\mathtt{u}) \geq \frac{1}{(\sigma(\mathtt{k}){+}1)^2} \right] {\cdot} X(\sigma) }\\
&&\mathtt{if(u < 1/(k+1)^2)}\\
&&\mathtt{\{}\\
{\color{red} //} && {\color{red} \lambda \sigma . \frac{\sigma(\mathtt{k})}{\sigma(\mathtt{k}){+}1} {\cdot} X(\sigma[\mathtt{b} \mapsto 1])}\\
&&\mathtt{\ \ b := 1;}\\
{\color{red} //} && {\color{red} \lambda \sigma . \frac{\sigma(\mathtt{k})}{\sigma(\mathtt{k}){+}1} {\cdot} X(\sigma)}\\
&&\mathtt{\ \ score(k/(k+1));}\\
{\color{red} //} && {\color{red} X}\\
&&\mathtt{\}}\\
{\color{red} //} && {\color{red} X}
\end{eqnarray*}

To simplify the presentation, we assumed in the last step that $X$ does not
depend directly on the variable $\mathtt{u}$---we can show by a simple induction that this
holds for $X = {}^{\mathtt{wp}}_{\langle b = 0 , C' \rangle} \Phi^{n}_{1}(0)$ for any $n$, 
and we only need to  apply $\mathtt{wp}|[C'|](\cdot)$ to functions $X$ of this form.
By plugging 
$$ 
\begin{array}{lcl}
\mathtt{wp}|[C'|](X) & = & \displaystyle \lambda \sigma .
 \frac{\sigma(\mathtt{k}){+}1}{(\sigma(\mathtt{k}){+}2)^3} {\cdot} X(\sigma[k \mapsto \sigma(k){+}1][\mathtt{b} \mapsto 1]) \\
 &  & \displaystyle \ + \ \frac{(\sigma(\mathtt{k}){+}2)^2 - 1}{(\sigma(\mathtt{k}){+}2)^2}  {\cdot} X(\sigma[k \mapsto \sigma(k){+}1])
\end{array}
$$ 
 into the equation for the characteristic function, we get 
 \[
 \begin{split}
 {}^{\mathtt{wp}}_{\langle b = 0, C' \rangle} \Phi_{1}(X) = \lambda \sigma .
[\sigma(b) \neq 0] + [\sigma(b)=0] 
(\frac{\sigma(\mathtt{k}){+}1}{(\sigma(\mathtt{k}){+}2)^3} X(\sigma[k \mapsto \sigma(k){+}1][\mathtt{b} \mapsto 1]) \\
+  \frac{(\sigma(\mathtt{k}){+}2)^2{-}1}{(\sigma(\mathtt{k}){+}2)^2}  X(\sigma[k \mapsto \sigma(k){+}1]))
\end{split}
\] 
We can now calculate subsequent terms of the sequence
${}^{\mathtt{wp}}_{\langle b = 0, C' \rangle} \Phi^{n}_{1}(0)$, whose supremum is the semantics
of the while-loop:
\allowdisplaybreaks
\begin{eqnarray*}
{}^{\mathtt{wp}}_{\langle \mathtt{b} = 0, C' \rangle} \Phi^{0}_{1}(0) &=& 0 \\
{}^{\mathtt{wp}}_{\langle \mathtt{b} = 0 , C' \rangle} \Phi^{1}_{1}(0) &=& [\mathtt{b} \neq 0] \\
{}^{\mathtt{wp}}_{\langle \mathtt{b} = 0 , C' \rangle} \Phi^{2}_{1}(0) &=& [\mathtt{b} \neq 0]
+ [b=0] \frac{k{+}1}{k{+}2} {\cdot} \frac{1}{(k{+}2)^2}\\
{}^{\mathtt{wp}}_{\langle \mathtt{b} = 0 , C' \rangle} \Phi^{3}_{1}(0) &=& [\mathtt{b} \neq 0]
+ [b=0] \frac{k{+}1}{k{+}2} {\cdot} \left ( \frac{1}{(k{+}2)^2} + \frac{1}{(k{+}3)^2} \right )\\
{}^{\mathtt{wp}}_{\langle \mathtt{b} = 0 , C' \rangle} \Phi^{4}_{1}(0) &=& [\mathtt{b} \neq 0]
+ [\mathtt{b}=0] \frac{k{+}1}{k{+}2} {\cdot} \left ( \frac{1}{(k{+}2)^2} + \frac{1}{(k{+}3)^2} + \frac{1}{(k{+}4)^2} \right )\\
&\dots&
\end{eqnarray*}
\noindent
It follows that ${}^{\mathtt{wp}}_{\langle b = 0, C' \rangle} \Phi^{n}_{1}(0)$ can be represented in
a closed form for any $n$:
\[
{}^{\mathtt{wp}}_{\langle b = 0 , C' \rangle} \Phi^{n}_{1}(0) = [b \neq 0]
+ [b=0] {\cdot} \frac{k+1}{k+2} {\cdot} \left ( \sum_{i=2}^{n}  \frac{1}{(k+i)^2} \right )
\]
The correctness of this formula can be proven by a simple induction on $n$ (which we omit here). This means that the semantics of the while-loop has the form: 
\begin{eqnarray*}
\mathtt{wp}|[\mathtt{while}(\mathtt{b}=0)\{C'\}|](1) & = &
\sup_n{}^{\mathtt{wp}}_{\langle b = 0 , C' \rangle} \Phi^{n}_{1}(0) \\
& = & 
[b \neq 0] + [b=0] \frac{k{+}1}{k{+}2} {\cdot} \left ( \sum_{i=2}^{\infty}  \frac{1}{(k{+}i)^2} \right )
\end{eqnarray*}
We can now use this result to compute the postexpectation of $\sigma . 1$ with respect to the full
program (where the while-loop is the program $C'$, whose semantics has already been calculated):
\begin{eqnarray*}
{\color{red} //} && {\color{red} \frac{\pi^2}{12} - \frac{1}{2} }\\
&& {\color{red}  =}\\
{\color{red} //} && {\color{red} 
\frac{1}{2} {\cdot} \sum_{i=2}^{\infty}  \frac{1}{i^2}}\\
&&\mathtt{b :=0;}\\
{\color{red} //} && {\color{red} 
[\mathtt{b} \neq 0]+ [\mathtt{b}=0] \frac{1}{2} {\cdot} \left ( \sum_{i=2}^{\infty}  \frac{1}{i^2} \right )}\\
&&\mathtt{k=0;}\\
{\color{red} //} && {\color{red} 
[\mathtt{b} \neq 0]+ [b=0] \frac{\mathtt{k}{+}1}{\mathtt{k}{+}2} {\cdot} \left ( \sum_{i=2}^{\infty}  \frac{1}{(\mathtt{k}{+}i)^2} \right )}\\
& & \mathtt{C'} \\
%%&&\mathtt{while (b=0)}\\
%%&&\mathtt{\{}\\
%%&&\mathtt{\ \ u := U;}\\
%%&&\mathtt{\ \ k := k+1;}\\
%%&&\mathtt{\ \ if(u > 1/(k+1)^2)}\\
%%&&\mathtt{\ \ \{}\\
%%&&\mathtt{\ \ \ \ b:=1;}\\
%%&&\mathtt{\ \ \ \ score(k/(k+1));}\\
%%&&\mathtt{\ \ \}}\\
%%&&\mathtt{\}}\\
&&{ \color{red} 1}
\end{eqnarray*}

In the last step, we used the well-known fact that the series $\sum_{i=2}^{\infty}  \frac{1}{i^2}$ converges to 
$ \frac{\pi^2}{6} - 1$, to establish that $\mathtt{wp}|[C|](1) = \frac{\pi^2}{12} - \frac{1}{2}$.
%. Hence,  the desired postexpectation is $\mathtt{wp}|[C|](1) = \frac{\pi}{12} - \frac{1}{2}$.
\end{example}

\begin{example} %\marginpar{TBC}
%We need an example showing the applicability of wlp to an AST program with scoring and sampling, preferably (a fragment/simplified version of ) one from the introduction.
%
In order to illustrate the $\mathtt{wp}$ semantics on a more realistic program,
let us recall the tortoise and hare example with soft conditioning from the introduction (with the $\mathtt{time}$ variable removed
for simplicity and the Gaussian density in $\mathtt{score}$ replaced by $\mathtt{softeq}$ to ensure that scores are bounded).
After expanding the syntactic sugar,
this program has the following form:
\begin{verbatim}
u1 := U;
t := Gaussian_inv_cdf(5,2,u1);   
h := 0.0;
while (h < t)
{
  t := t + 1 + e;
  u2 := U;
  if (u2 < 0.25)
  {
    u3 := U;
    h := h + Gaussian_inv_cdf(4,2,u3);
  }
}
score(softeq(t, 60.0));
\end{verbatim}

Let us suppose we want to calculate the expected distance travelled by the tortoise before it gets caught.
To this end, we need to calculate $\mathtt{wp}|[C|](t)$ for the above program $C$ and apply it to the empty
initial state (or in fact any initial state, as the program contains no free variables). 

Like in the previous example, we begin by calculating the characteristic function of the while-loop. 
%By plugging the predicate $\mathtt{h} < \mathtt{t}$
%into the formula for the characteristic function, we get 
%${}^{\mathtt{wp}}_{\langle h < t , C' \rangle} \Phi_{f}(X) = [\mathtt{h} \geq \mathtt{t}]{\cdot}f +[\mathtt{h} < \mathtt{t}]{\cdot}\mathtt{wp}|[C'|](X)$,
%where $C'$ is the body of the while loop. By Lemma~\ref{lemma:while-kleene}, we can then conclude that
%$\mathtt{wp}|[\mathtt{while}(h<t)\{C'\}|](f)
%= \sup_n\, {}^{\mathtt{wp}}_{\langle h < t , C' \rangle} \Phi^{n}_{f}(0)$.
%
We first need to compute $\mathtt{wp}|[C'|](X)$ for the loop body $C'$.
We assume that $X$ does not depend directly on variables $\mathtt{u2}$ and $\mathtt{u3}$---this is safe
for $X$ of the form $X = {}^{\mathtt{wp}}_{\langle h < t , C' \rangle} \Phi^{n}_{f}(0)$, as long as
$f$ does not depend directly on the aforementioned variables.
%
%To simplify presentation, we assume that $X$ does not
%depend directly on variables $\mathtt{u2}$ and $\mathtt{u3}$---we can show by a simple induction that this
%holds for $X = {}^{\mathtt{wp}}_{\langle h < t , C' \rangle} \Phi^{n}_{f}(0)$ for any $n$ (if $f$ is independent 
%of $\mathtt{u2}$ and $\mathtt{u3}$), and we only need to  apply
%$\mathtt{wp}|[C'|](\cdot)$ to functions $X$ of this form.

%\text{\tiny ($X$ independent of $\mathtt{u2}$ and $\mathtt{u3}$)} 
\begin{eqnarray*}
{\color{red} //} && {\color{red} \lambda \sigma .\, 0.25 \cdot
 \int_{(0,1)} X(\sigma[\mathtt{t} \mapsto \sigma(\mathtt{t}) {+} 1 {+} e]
[\mathtt{h} \mapsto \sigma(\mathtt{h})}\\
&& \qquad {\color{red} + \mathtt{G}(4,2, v_3)])\, \mu_L(dv_3) + 0.75 \cdot X(\sigma
[\mathtt{t} \mapsto \sigma(\mathtt{t}) {+} 1 {+} e])}\\
&&\mathtt{t := t + 1 + e;}\\
{\color{red} //} && {\color{red} \lambda \sigma .\, 0.25 \cdot \int_{(0,1)} X(\sigma
[\mathtt{h}  \mapsto \sigma(\mathtt{h} ) + \mathtt{G}(4,2, v_3)])\, \mu_L(dv_3) + 0.75 \cdot
X(\sigma)}\\
 && {\color{red} =}\\
{\color{red} //} && {\color{red} \lambda \sigma . \int_{(0,1)}[v_2 < 0.25]
 \int_{(0,1)} X(\sigma[\mathtt{u2} \mapsto v_2][\mathtt{u3} \mapsto v_3] 
[\mathtt{h}  \mapsto \sigma(\mathtt{h} )}\\
&& \qquad {\color{red} + \mathtt{G}(4,2, v_3)])\, \mu_L(dv_3) + [v_2 \geq 0.25] 
X(\sigma[\mathtt{u2} \mapsto v_2])
\, \mu_L(dv_2)}\\
%&& \qquad  { \color{red} + [\sigma(\mathtt{u2}) \leq 0.25] X(\sigma)}\\
&&\mathtt{u2 := U;}\\
{\color{red} //} && {\color{red} \lambda \sigma . [\sigma(\mathtt{u2}) < 0.25] \int_{(0,1)} X(\sigma[\mathtt{u3} \mapsto v_3] 
[\mathtt{h}  \mapsto \sigma(\mathtt{h} ) {+} \mathtt{G}(4,2, v_3)]\, \mu_L(dv_3)}\\
&& \qquad  { \color{red} + [\sigma(\mathtt{u2}) \geq 0.25] \cdot  X(\sigma)}\\
&&\mathtt{ if (u2 < 0.25)}\\
&&\mathtt{\{}\\
{\color{red} //} && {\color{red} \lambda \sigma . \int_{(0,1)} X(\sigma[\mathtt{u3} \mapsto v_3] 
[\mathtt{h}  \mapsto \sigma(\mathtt{h} ) {+} \mathtt{G}(4,2, v_3)]\, \mu_L(dv_3)}\\
&&\mathtt{\ \ u3 := U;}\\
{\color{red} //} && {\color{red} \lambda \sigma . X(\sigma[\mathtt{h}  \mapsto \sigma(\mathtt{h} ) {+} \mathtt{G}(4,2, \sigma(\mathtt{u3}))]) }\\
&&\mathtt{\ \ h := h + Gaussian\_inv\_cdf(4,2,u3); }\\
{\color{red} //} && {\color{red} X}\\
&&\mathtt{\}}\\
{\color{red} //} && {\color{red} X}
\end{eqnarray*}

%Now, if we write $\sigma'_x$ for 

Thus, we have 
$  {}^{\mathtt{wp}}_{\langle h < t , C' \rangle} \Phi_{f}(X) =
\lambda \sigma .\, 
[\sigma(\mathtt{h}) \geq \sigma(\mathtt{t})]{\cdot}f(\sigma) +[\sigma(\mathtt{h}) < \sigma(\mathtt{t})]{\cdot}
(0.25 \cdot
 \int_{[0,1]} X(\sigma'_{\sigma, v_3})\, \mu_L(dv_3) + 0.75 \cdot X(\sigma''_\sigma) ) )$,
where $\sigma'_{\sigma, v_3} = \sigma[\mathtt{t} \mapsto \sigma(\mathtt{t}) + 1 + e]
[h \mapsto \sigma(\mathtt{h}) + \mathtt{G}(4,2, v_3)]$ is the state $\sigma$
updated after a step where both the tortoise and the hare moved (the latter by $\mathtt{G}(4,2, v_3)$)
and $\sigma''_{\sigma} = \sigma
[\mathtt{t} \mapsto \sigma(\mathtt{t}) + 1 + e]$ is state $\sigma$ updated after a step where
the hare stood still.

By the inductive definition of the $\mathtt{wp}$ operator, we have $\mathtt{wp}|[\mathtt{while}(\mathtt{h} < \mathtt{t}) \{C'\}|] = 
\sup_n\, {}^{\mathtt{wp}}_{\langle h < t , C' \rangle} \Phi^{n}_{f}(0)$.
Then $\sup_n\, {}^{\mathtt{wp}}_{\langle h < t , C' \rangle} \Phi^{n}_{f}(0)$ is guaranteed to exist,
but unlike in the previous example, it does not have a nice closed form. This is indeed the case for most
real-world programs.

%Thus,  $\mathtt{wp}|[C'|](X) = \lambda \sigma .\, 0.25 \cdot
% \int_{[0,1]} X(\sigma[\mathtt{t} \mapsto \sigma(\mathtt{t}) + 1 + e] [\mathtt{u3} \mapsto v_3] 
%[h \mapsto \sigma(\mathtt{h}) + \mathtt{Gaussian\_inv\_cdf}(4,2, v_3)]\, \lambda(dv_3) + 0.75 \cdot X(\sigma
%[\mathtt{t} \mapsto \sigma(\mathtt{t}) + 1 + e][\mathtt{u2} \mapsto v_2])$. 

%and so the
%characteristic function is 
%$\lambda \sigma .  {}^{\mathtt{wp}}_{\langle h < t , C' \rangle} \Phi_{f}(X) = \lambda \sigma .\, 
%[\sigma(\mathtt{h}) \geq \sigma(\mathtt{t})]{\cdot}f +[\sigma(\mathtt{h}) < \sigma(\mathtt{t})]{\cdot}
%(0.25 \cdot
% \int_{[0,1]} X(\sigma[\mathtt{t} \mapsto \sigma(\mathtt{t}) + 1 + e] [\mathtt{u3} \mapsto v_3] 
%[h \mapsto \sigma(\mathtt{h}) + \mathtt{Gaussian\_inv\_cdf}(4,2, v_3)]\, \lambda(dv_3) + 0.75 \cdot X(\sigma
%[\mathtt{t} \mapsto \sigma(\mathtt{t}) + 1 + e][\mathtt{u2} \mapsto v_2]) )$. 
%

%In the step marked with $*$, we made the assumption that $X$ does not depend directly on variables $\mathtt{u2}$
%and $\mathtt{u3}$. 

We can now derive the formula for the expected final value of $\mathtt{t}$:

\begin{eqnarray*}
{\color{red} //} && {\color{red} \lambda \sigma . \int_{(0,1)}  \sup_n\, {}^{\mathtt{wp}}_{\langle h < t , C' \rangle} 
  \Phi^{n}_{e^{-(\mathtt{t}- 60.0)^2} \mathtt{t}}(0)}\\
{\color{red} //} && {\color{red} \qquad (\sigma[t \mapsto \mathtt{G}(5,2,\sigma(v_1))][\mathtt{h} \mapsto 0])
\, \mu_L(dv_1) }\\
&& {\color{red} =}\\
{\color{red} //} && {\color{red} \lambda \sigma . \int_{(0,1)}  \sup_n\, {}^{\mathtt{wp}}_{\langle h < t , C' \rangle} 
  \Phi^{n}_{ e^{-(\mathtt{t}- 60.0)^2} \mathtt{t}}(0)}\\
&&{ \color{red} \qquad (\sigma[\mathtt{u1} \mapsto v_1] [t \mapsto \mathtt{G}(5,2,\sigma(\mathtt{u1}))][\mathtt{h} \mapsto 0])
\, \mu_L(dv_1) }\\
&&\mathtt{u1 := U;}\\
{\color{red} //} && {\color{red} \lambda \sigma . \sup_n\, {}^{\mathtt{wp}}_{\langle h < t , C' \rangle} 
  \Phi^{n}_{ e^{-(\mathtt{t}- 60.0)^2} \mathtt{t}}(0)
(\sigma[t \mapsto \mathtt{G}(5,2,\sigma(\mathtt{u1}))][\mathtt{h} \mapsto 0]) }\\
&&\mathtt{t := Gaussian\_inv\_cdf(5,2,u1);}\\
%{\color{red} //} && {\color{red} TODO}\\
%{\color{red} //} && {\color{red} \lambda \sigma .  \sup_n (\lambda  X . [\mathtt{h} \geq \mathtt{t}]{\cdot}
% \mathtt{softeq}(\mathtt{t}, 60.0){\cdot} \mathtt{t}
% +[\mathtt{h} < \mathtt{t}]{\cdot}\mathtt{wp}|[C'|](X) )^n(0) (\sigma[\mathtt{h} \mapsto 0]) }\\
{\color{red} //} && {\color{red} \lambda \sigma . \sup_n\, {}^{\mathtt{wp}}_{\langle h < t , C' \rangle} 
  \Phi^{n}_{e^{-(\mathtt{t}- 60.0)^2} \mathtt{t}}(0)(\sigma[\mathtt{h} \mapsto 0]) }\\
&&\mathtt{h := 0.0;}\\
{\color{red} //} && {\color{red} \sup_n\, {}^{\mathtt{wp}}_{\langle h < t , C' \rangle} 
  \Phi^{n}_{ e^{-(\mathtt{t}- 60.0)^2} \mathtt{t}}(0) }\\
&&\mathtt{while (h < t) \{...\}}\\
{\color{red} //} && {\color{red} \ e^{-(\mathtt{t}- 60.0)^2} \mathtt{t}}\\
&&\mathtt{score(softeq(t, 60.0));}\\
{\color{red} //} && {\color{red} \mathtt{t}}
\end{eqnarray*}

In the last step, we used the fact that
$ {}^{\mathtt{wp}}_{\langle h < t , C' \rangle} 
  \Phi^{n}_{\mathtt{softeq}(\mathtt{t}, 60.0) \mathtt{t}}(0)$
does not depend directly on $\mathtt{u1}$. We have now derived the expression for the weakest preexpectation semantics
of the program $C$:
\begin{eqnarray*}
\mathtt{wp}|[\mathtt{C}|](t) &= &
 \lambda \sigma . \int_{(0,1)}  \sup_n\, 
  \Phi^{n}(0) (\sigma[t \mapsto \mathtt{Gaussian\_inv\_cdf}(5,2,\sigma(v_1))][\mathtt{h} \mapsto 0])\\
\text{where} &&\\
%\Phi^0 (X) &=&\lambda \sigma . \sigma \\
\Phi(X) &=&
\lambda \sigma .\, 
[\sigma(\mathtt{h}) \geq \sigma(\mathtt{t})]{\cdot}\mathtt{softeq}(\sigma(\mathtt{t}), 60.0) \sigma(\mathtt{t}) \\
&& \qquad + [\sigma(\mathtt{h}) < \sigma(\mathtt{t})]{\cdot}
(0.25 \cdot
 \int_{(0,1)} X(\sigma'_{\sigma, v_3})\, \mu_L(dv_3) \\
&&\qquad + 0.75 \cdot X(\sigma''_\sigma) ) )\\
\sigma'_{\sigma, v_3} &=& \sigma[\mathtt{t} \mapsto \sigma(\mathtt{t}) + 1 + e]
[h \mapsto \sigma(\mathtt{h}) + \mathtt{Gaussian\_inv\_cdf}(4,2, v_3)]\\
\sigma''_{\sigma} &=& \sigma [\mathtt{t} \mapsto \sigma(\mathtt{t}) + 1 + e].
\end{eqnarray*}
\end{example}

\subsection{Weakest liberal preexpectation semantics}

We now define a different variant of the above semantics, called the 
\emph{weakest liberal preexpectation} semantics ($\mathtt{wlp}$).  
In standard, discrete $\mathtt{pGCL}$ without scores \cite{MM05:AbstractionRefinementProofForProbabilisticSystems,Olmedo18},
the weakest liberal preexpectation defines the expected value of a function bounded by $1$ (as per $\mathtt{wp}$) plus the probability of
divergence---in other words, in contrast to $\mathtt{wp}$, the $\mathtt{wlp}$ operator considers the value of the input function to be
$1$, rather than $0$, for diverging program runs. 
If the input function is a binary predicate $\phi$, $\mathtt{wlp}$ defines the probability
of this predicate being satisfied in the final state \emph{or} the program never terminating.

In \ccpgclns, the concept of weakest liberal preexpectation is similar, except that probabilities of all outcomes again have to be multiplied by scores encountered during the program's execution. 
Formally, $\mathtt{wp}|[C |](\cdot)$, takes a measurable function $f$ mapping $\statespace$ to 
$[0,1]$ and returns another measurable function from $\statespace$ to $[0,1]$. 
Note that, in contrast to $\mathtt{wp}$, the domain of the input function is restricted to the unit interval.

The $\mathtt{wlp}$ operator is defined in Fig.~\ref{figure:wlp-ccpgcl}, with changes from $\mathtt{wp}$ marked in blue. 

The semantics of a $\mathtt{while}$ loop is now computed with the greatest fixpoint rather than the least fixpoint---this
has the effect that the ``default'' outcome for diverging loops is $1$ instead of $0$. Similarly, $\mathtt{diverge}$ converts every 
function into a constant $1$ function. 
The remaining changes are just that the recursive invocations to $\mathtt{wp}$ are replaced with calls to $\mathtt{wlp}$.

\begin{figure}[t]
\centering
%\begin{minipage}{2.3cm}
\begin{eqnarray*}
\mathtt{wlp}|[ \mathtt{skip} |](f) & \ = \ & f \\[1ex]
\mathtt{wlp}|[ \mathtt{diverge} |](f) &=& \textcolor{blue}{1} \\[1ex]
\mathtt{wlp}|[ x := E |](f) &=& \lambda \sigma . f(\sigma[x \mapsto \sigma(E)])\\[1ex]
\mathtt{wlp} |[  x :\approx U  |] (f) &=& \lambda \sigma . \int_{[0,1]} 
f(\sigma[x \mapsto v ])  \, \mu_L(dv)\\[1ex]
\mathtt{wlp}|[ \mathtt{observe}(\phi)|](f)  &=&\lambda \sigma . [\sigma(\phi)]{\cdot}f(\sigma)\\[1ex]
\mathtt{wlp} |[ \mathtt{score}(E) |] (f) &=& \lambda \sigma . [\sigma(E) \in (0,1]]{\cdot}\sigma(E) {\cdot} f(\sigma)\\[1ex]
\mathtt{wlp}|[ C_1;C_2 |](f) &=& \textcolor{blue}{\mathtt{wlp}}|[ C_1|](\textcolor{blue}{\mathtt{wlp}}|[C_2 |](f))\\[1ex]
\mathtt{wlp}|[\mathtt{if}(\phi)\{ C\} |](f) &=& [\phi]{\cdot}\textcolor{blue}{\mathtt{wlp}}|[C|](f)+ [\neg \phi]{\cdot}f \\[1ex]
\mathtt{wlp}|[ \mathtt{while}(\phi)\{C\} |](f) &=& \textcolor{blue}{\mathtt{gfp}} X . [\neg \phi]{\cdot}f + [\phi]{\cdot}\textcolor{blue}{\mathtt{wlp}}|[C|](X)
\end{eqnarray*}
%\end{minipage}
\caption{Weakest liberal preexpectation semantics of \ccpgcl}
 \label{figure:wlp-ccpgcl}
\end{figure}

\paragraph{Well-definedness of $\mathtt{wlp}$.}
To show that the liberal semantics is well-defined, we use a similar argument as for $\mathtt{wp}$. First, note that
if we restrict the set of measurable functions $f \colon \statespace -> \extposreals$ to functions
$f \colon \statespace -> [0,1]$ with values in $ [0,1]$, the constant function $\lambda \sigma . 1$ (denoted $1$ in short) 
is its top element. Hence, we can invert the complete partial order to get an $\omega$-cpo with
inverse pointwise ordering and a ``bottom'' element $1$. %$f: \statespace -> [0,1]$
The supremum of functions in this inverted cpo corresponds to the infimum in the original cpo,
so continuity of $\mathtt{wlp}$ can again be proven by induction using the same domain-theoretic results.
In the proof of measurability
of $\mathtt{wlp}|[C|](f)$, we use the fact that the infimum of a sequence of measurable functions is measurable,
just like with supremum. 
By Kleene's Fixpoint Theorem, we again know that
$\mathtt{gfp}\ X . [\neg \phi]{\cdot}f + [\phi]{\cdot}\mathtt{wlp}|[C|](X)$ exists and 
equals $\inf_n {}^{\mathtt{wlp}}_{\langle \phi, C \rangle} \Phi_f^n(1)$,
where ${}^{\mathtt{wlp}}_{\langle \phi, C \rangle} \Phi_f(X) = [\neg \phi]{\cdot}f + [\phi]{\cdot} \mathtt{wlp}|[C|](X)$.

Note that the weakest liberal preexpectation is only defined for bounded postexpectations $f$, as we need some upper bound to set the preexpectation to in case of divergence. 
If we chose this bound to be $\lambda \sigma. \infty$, $\mathtt{wlp}$ would effectively always be set to $\infty$ for all non almost-surely terminating programs, rendering the semantics useless.

\begin{example} %\marginpar{TBC}

To show how $\mathtt{wlp}$ differs from $\mathtt{wp}$, let us consider Example 2 again. This time, we want to compute $\mathtt{wlp}|[C|](1)$,
where $C$ is again the full program. Like before, we begin by computing the semantics of the loop. As the body $C'$ of the loop is
itself loop (and $\mathtt{diverge}$)-free, we have $\mathtt{wlp}|[C'|](1) = \mathtt{wp}|[C'|](1) = 
 \lambda \sigma .
 \frac{\sigma(\mathtt{k})+1}{(\sigma(\mathtt{k})+2)^3} X(\sigma[\mathtt{k} \mapsto \sigma(\mathtt{k})+1][\mathtt{b} \mapsto 1])
+  \frac{(\sigma(\mathtt{k})+2)^2 - 1}{(\sigma(\mathtt{k})+2)^2}  X(\sigma[k \mapsto \sigma(k)+1]))$,
which implies ${}^{\mathtt{wlp}}_{\langle \mathtt{b}=0, C' \rangle} \Phi_1(X) =
{}^{\mathtt{wp}}_{\langle \mathtt{b}=0, C' \rangle} \Phi_1(X) =
\lambda \sigma . 
[\sigma(\mathtt{b}) \neq 0] + [\sigma(\mathtt{b})=0] (
 \frac{\sigma(\mathtt{k})+1}{(\sigma(\mathtt{k})+2)^3} X(\sigma[\mathtt{k} \mapsto \sigma(\mathtt{k})+1][\mathtt{b} \mapsto 1])
+  \frac{(\sigma(\mathtt{k})+2)^2 - 1}{(\sigma(\mathtt{k})+2)^2}  X(\sigma[k \mapsto \sigma(k)+1])) )$.
%Thus, $\mathtt{wlp}|[\mathtt{while}(\mathtt{b}=0)\{C'\}|](1) = \inf_n 
%{}^{\mathtt{wlp}}_{\langle \mathtt{b}=0, C' \rangle} \Phi_1^n(1)$
%%
The first terms of the sequence ${}^{\mathtt{wlp}}_{\langle \mathtt{b}=0, C' \rangle} \Phi_1(1)$ are as follows:
\begin{eqnarray*}
{}^{\mathtt{wlp}}_{\langle \mathtt{b} = 0, C' \rangle} \Phi^{0}_{1}(1) &=& [\mathtt{b} \neq 0] \\
{}^{\mathtt{wlp}}_{\langle \mathtt{b} = 0 , C' \rangle} \Phi^{1}_{1}(1) &=& [\mathtt{b} \neq 0] 
{+} [\mathtt{b} = 0] \left( \frac{\mathtt{k}{+}1}{\mathtt{k}{+}2}  \frac{1}{(\mathtt{k}{+}2)^2} + \frac{(\mathtt{k}{+}2)^2{-}1}{(\mathtt{k}{+}2)^2} \right)\\
{}^{\mathtt{wlp}}_{\langle \mathtt{b} = 0 , C' \rangle} \Phi^{2}_{1}(1) &=& [\mathtt{b} \neq 0] 
+ [\mathtt{b} = 0] \left( \frac{\mathtt{k}{+}1}{\mathtt{k}{+}2} \left( \frac{1}{(\mathtt{k}{+}2)^2} + \frac{1}{(\mathtt{k}{+}3)^2}\right ) \right.\\ 
&& \phantom{xxxxxxxxxxxxx} +  \left. \frac{(\mathtt{k}{+}2)^2{-}1}{(\mathtt{k}{+}2)^2} {\cdot} \frac{(\mathtt{k}{+}3)^2 -1}{(\mathtt{k}{+}3)^2}  \right)\\
{}^{\mathtt{wlp}}_{\langle \mathtt{b} = 0 , C' \rangle} \Phi^{3}_{1}(1) &=& [\mathtt{b} \neq 0] 
+ [\mathtt{b} = 0] \left( \frac{\mathtt{k}{+}1}{\mathtt{k}{+}2} \left( \frac{1}{(\mathtt{k}{+}2)^2} + \frac{1}{(\mathtt{k}{+}3)^2}
+ \frac{1}{(\mathtt{k}{+}4)^2} \right ) \right.\\ 
&& \phantom{xxxxxxxxxxxxx} +  \left. \frac{(\mathtt{k}{+}2)^2{-}1}{(\mathtt{k}{+}2)^2} {\cdot} \frac{(\mathtt{k}{+}3)^2{-}1}{(\mathtt{k}{+}3)^2}  
{\cdot} \frac{(\mathtt{k}{+}4)^2{-}1}{(\mathtt{k}{+}4)^2} \right) \\
&\dots&
\end{eqnarray*}
We can now see what the pattern is:
\[
{}^{\mathtt{wlp}}_{\langle b = 0 , C' \rangle} \Phi^{n}_{1}(0) = [b \neq 0] 
+ [b=0] {\cdot} \left ( \frac{k{+}1}{k{+}2} {\cdot} \sum_{i=2}^{n{+}1}  \frac{1}{(k{+}i)^2} + \prod_{i=2}^{n{+}1} 
\frac{(\mathtt{k} {+} i)^2{-}1}{(\mathtt{k} {+} i)^2} \right )
\]
Moreover, we can quickly check that 
$\prod_{i=2}^{n+1} 
\frac{(\mathtt{k}{+}i)^2-1}{(\mathtt{k}{+} i)^2} = \frac{\mathtt{k}{+}1}{\mathtt{k}{+}2} {\cdot}\frac{\mathtt{k}{+}n{+}2}{\mathtt{k}{+}n{+}1}$.
We obtain $\mathtt{wlp}|[C|]$ by computing
the $\mathtt{wlp}$ of $\inf_n {}^{\mathtt{wlp}}_{\langle b = 0 , C' \rangle} \Phi^{n}_{1}(0)$
with respect to the two initial statements, $\mathtt{k := 0}$ and $\mathtt{b := 0}$.
Thus, 
$$
\begin{array}{lcl}
\mathtt{wlp}|[C|](1) & = & \displaystyle \frac{1}{2} {\cdot} \inf_n   \sum_{i=2}^{n+1}  \frac{1}{i^2}
+ \frac{1}{2} \frac{n{+}2}{n{+}1} \\
& = & \displaystyle \frac{1}{2} {\cdot} \lim_{n ->\infty} \sum_{i=2}^{n+1}  \frac{1}{i^2}
+ \frac{1}{2} \frac{n{+}2}{n{+}1} 
\ = \ \frac{1}{2} {\cdot} \left( \frac{\pi^2}{6} {-} 1 \right) + \frac{1}{2} 
\ = \ \frac{\pi^2}{12}.
\end{array}
$$
\end{example}

\subsection{Redundancy of $\mathtt{score}$}

With respect to the weakest preexpectations semantics, the $\mathtt{score}$ operator admitting only arguments bounded by one is redundant, because scoring by a number in the unit interval can be simulated by rejection sampling without affecting the expected value of the given function. 
We show this result in this section.

To this end, we first need some additional concepts. 
Let $\mathtt{dom}(\sigma)$ be the set of variables which are assigned values in state $\sigma$. 
A function $f$ is said to be \emph{independent} of a variable $x$ if %$f(\sigma[x \mapsto V_1]) = f(\sigma[x \mapsto V_2])$
$f(\sigma) = f(\sigma[x \mapsto V])$
for all $\sigma \in \statespace$ and $V \in \mathbb{R}$.
%\marginpar{why is $\mathbb{Z}$ explicitly mentioned here?}
%$V_1, V_2 \in \mathbb{R} \cup \mathbb{Z}$.
Let $\mathtt{fv}(E)$ be the set of free variables of an expression $E$, and $\mathtt{vars}(E)$ and $\mathtt{vars}(C)$ be the sets of all variables (free or bound) appearing in, respectively, the expression $E$ and the program $C$.

\begin{lemma} \label{lemma:score-redundancy-step}
For every expectation $f$, expression $E$ and variable $u$  %$\sigma$ and $E$ 
such that %$u \notin \mathtt{dom}(\sigma)$ and
$u \notin \mathtt{vars}(E)$ and
$f$ is independent of $u$, it holds:
\begin{eqnarray*}
\mathtt{wp}|[\mathtt{score}(E)|](f) \ &=& \ 
\mathtt{wp}|[u :\approx U; \mathtt{observe}(E \in (0,1] \wedge u \leq E)|](f) \\[1ex] 
\mathtt{wlp}|[\mathtt{score}(E)|](f) \ &=& \
\mathtt{wlp}|[u :\approx U; \mathtt{observe}(E \in (0,1] \wedge u \leq E)|](f).
\end{eqnarray*}
\end{lemma}
\begin{proof}

For $\mathtt{wp}$ we have:
\begin{eqnarray*}
&   & \mathtt{wp}|[u :\approx U|](\mathtt{wp}|[\mathtt{observe}(E \in (0,1] \wedge u \leq E)|](f)) \\
&= & \lambda \sigma . \int_{[0,1]} \mathtt{wp}|[\mathtt{observe}(E \in (0,1] \wedge u \leq E)|](f)(\sigma[u \mapsto v]) \, \mu_L(dv)\\
&=& \lambda \sigma .  \int_{[0,1]} [\sigma[u \mapsto v](E) \in (0,1] ] [v \leq \sigma[u \mapsto v](E)] f(\sigma[u \mapsto v])\, \mu_L(dv) \\
{\tiny\text{($*$)}}&=& \lambda \sigma .  \int_{[0,1]} [\sigma(E) \in (0,1] ] [v \leq \sigma(E)] f(\sigma)\, \mu_L(dv) \\
&=& \lambda \sigma .  f(\sigma) \int_{[0,1]} [\sigma(E) \in (0,1] ] [v \leq \sigma(E)] \, \mu_L(dv) \\
{\tiny\text{(Lebesgue)}}&=& \lambda \sigma . f(\sigma) \cdot \sigma(E) \\
&=& \mathtt{wp}|[\mathtt{score}(E)|](f).
\end{eqnarray*}

Proof step $(*)$ follows from the fact that $u \notin \mathtt{fv}(E)$ and that $f$ is independent of $u$.
%%
%{\tiny\text{($u$ not in $\mathtt{fv}(E)$, $f$ indep. of $u$)}}
%%
The above result also proves the second item of the lemma, as $\mathtt{wp}$ and $\mathtt{wlp}$ coincide for programs without loops and $\mathtt{diverge}$ statements.
\qed \end{proof}

%is a fresh variable.%u can be the same for each draw, as the result of the draw is used immediately and ignored later

%Define $\mathtt{noscore}(C)$ to be a function which transforms the program
%$C$ into a $\mathtt{score}$-free program by replacing each expression of the
%form $\mathtt{score}(E)$ with $u :\approx U; \mathtt{observe}(u \leq E)$,
%where $u \notin \mathtt{fv}(E)$.

Let $\mathtt{noscore}(C)$ denote the program obtained from program $C$ by replacing each expression of the form $\mathtt{score}(E)$ by $u :\approx U; \mathtt{observe}(u \leq E)$
for sufficiently fresh variable $u \notin \mathtt{vars}(E)$. 
By ``sufficiently fresh'' we mean that $u$ does not appear in the program $C$ and that no function $f$ whose expected value we are interested in depends on $u$.
(We do not formalise this notion for the sake of brevity.)

\begin{lemma} \label{lemma:score-redundancy}
For every expectation $f$ we have:
\[ \mathtt{wp}|[ \mathtt{noscore}(C)|](f) = \mathtt{wp}|[C|](f) \quad
\mbox{and} \quad 
\mathtt{wlp}|[ \mathtt{noscore}(C)|](f) = \mathtt{wlp}|[C|](f).
\]
\end{lemma}
\begin{proof}
By induction on the structure of $C$, with appeal to Lemma~\ref{lemma:score-redundancy-step}.
\qed \end{proof}

\section{Operational semantics}

In addition to the denotational semantics, we also present an operational semantics of \ccpgclns. 
Apart from serving as a sanity check for the wp-semantics, this semantics is of interest on its own: an operational semantics is typically closer to a sample-based semantics that provides the basis for simulation-based evaluation of probabilistic programs (such as MCMC and Metropolis Hasting), is closer to models that are amenable to automated verification techniques such as probabilistic model checking~\cite{DBLP:conf/lics/Katoen16}, and sometimes simplifies for the reasoning about probabilistic programs, such as proving some sort of program equivalence~\cite{DBLP:journals/pacmpl/WandCGC18}.

\subsection{Entropy space} \label{subsection:entropy-space}

A \emph{small-step} operational semantics of a deterministic imperative language typically takes a program $C$ and state $\sigma$ and performs a single step of program evaluation, returning a new program $C'$ and an updated state $\sigma'$.
For probabilistic languages, this is not possible, as a probabilistic program has multiple updated  states---in fact, infinitely and uncountably many of them for programs with continuous distributions---depending on the outcomes of random draws. 
A possible way around this is to define the operational semantics of probabilistic languages with respect to a fixed sequence of values sampled from subsequent distributions, called a \emph{trace}. By fixing a trace, a probabilistic program can be evaluated deterministically.
%
%Traces often have the form of finite \cite{DBLP:conf/icfp/BorgstromLGS16} or infinite \citep{park08sampling} lists of values.
%Such linear traces are easy to understand and it might be tempting to use them, but this would make the \emph{compositionality}
%of the semantics harder to achieve. Suppose, for instance, that we have 
%

Traces often have the form of finite \cite{DBLP:conf/icfp/BorgstromLGS16} or infinite \citep{park08sampling} lists of values.
To obtain a compositional semantics, we will instead use an abstract, infinite structure called \emph{entropy}, as defined by~\cite{Culpepper17} and~\cite{DBLP:journals/pacmpl/WandCGC18}.
%%
%% \marginpar{what are $\mathbb{S}, \mathcal{S}$? Formally and intuitively?}

\begin{definition}[\cite{DBLP:journals/pacmpl/WandCGC18}]
An \emph{entropy space} is a measurable space $(\mathbb{S}, \mathcal{S})$ equipped with
a measure $\mu_{\mathbb{S}}$ with $\mu_\mathbb{S}(\mathbb{S}) = 1$, and measurable functions $\pi_{U}\colon \mathbb{S} -> [0,1]$, $(\mathrel{::}) \colon \mathbb{S} \times \mathbb{S} -> \mathbb{S}$,
$\pi_L, \pi_R \colon \mathbb{S} -> \mathbb{S}$ such that:
\begin{itemize}
\item For all measurable functions $f \colon [0,1] -> \extposreals$ and Lebesgue measure $\lambda$,
% \mathbb{R}_{+}$:
\[
\int f(\pi_U(\theta))\, \mu_\mathbb{S}(d\theta) = \int_{[0,1]} f(x)\, \mu_L(dx)
\]
\item $(\mathrel{::})$ is a surjective pairing function defined by:
$\pi_L(\theta_L \mathrel{::} \theta_R) = \theta_L$ and $\pi_R(\theta_L \mathrel{::} \theta_R) = \theta_R$
%$\pi_L(\theta) \mathrel{::} \pi_R(\theta) = \theta$
\item For all measurable functions $g \colon \mathbb{S} \times \mathbb{S} -> \extposreals$:  %\mathbb{R}_{+}$:
\[
\int g(\pi_L(\theta), \pi_R(\theta))\, \mu_{\mathbb{S}}(d\theta) =
\int \int g(\theta_L, \theta_R)\, \mu_{\mathbb{S}}(d\theta_L ) \mu_{\mathbb{S}}(d\theta_R ).
\]
\end{itemize}
An element $\theta \in \mathbb{S}$ of the entropy space is called an \emph{entropy}.
\end{definition}

In the above definition, $\mathbb{S}$ is the set of all possible entropies and $\mathcal{S}$ a $\sigma$-algebra on it.
The entropy space is abstract, so we do  not specify what $\mathbb{S}$ and $\mathcal{S}$ are and what
they look like, we only assume that they satisfy the above properties.

\paragraph{Example of an entropy space} A simple concrete realisation of the entropy space, for which the properties are satisfied, is the following:
%is the space $[0,1]^\omega$ of infinite-dimensional Hilbert cubes. In this entropy space:

\begin{itemize}
\item The set $\mathbb{S}$ is the set $[0,1]^\omega$ of infinite sequences of numbers in $[0,1]$ (the so-called Hilbert cube). Thus,
each entropy $S \in \mathbb{S}$ is an infinite sequence $S = (s_1, s_2, s_3 \dots)$ such that  $s_i \in [0,1]$ for all $i$.
\item The $\sigma$-algebra $\mathcal{S}$ is, intuitively, the product of infinitely many copies of the Borel $\sigma$-algebra on $[0,1$].
More formally,  $\mathcal{S}$ is the $\sigma$-algebra generated by cylinder sets of the form
$A_1 \times A_2 \times \dots \times A_k \times [0,1] \times [0,1] \times [0,1] \dots$, where $A_1$, $A_2$, \dots, $A_k$ are Borel
subsets of $[0,1]$.
\item The measure $\mu_{\mathbb{S}}$ on $(\mathbb{S}, \mathcal{S})$ is the extension of the Lebesgue measure to
the infinite product space $(\mathbb{S}, \mathcal{S})$. Formally, it
is the unique (by Kolmogorov's extension theorem) measure such that for all finite sequences of 
Borel subets $A_1$, $A_2$, \dots, $A_k$ of $[0,1]$, we have
$\mu_{\mathbb{S}}(A_1 \times A_2 \times \dots \times A_k 
\times [0,1]^\omega) = \mu_L(A_1) \times \mu_L(A_2) \times \dots \times \mu_L (A_k) \times \mu_L([0,1]) \times \mu_L([0,1]) \dots
= \mu_L(A_1) \times \mu_L(A_2) \times \dots \times \mu_L (A_k) $
%\item The measure $\mu_{\mathbb{S}}$ on $(\mathbb{S}, \mathcal{S})$ is TODO
\item The function $\pi_U$ returns the first element of the given sequence---that is, $\pi_U((s_1, s_2, s_3, \dots)) = s_3$.
\item The functions $\pi_L$ and $\pi_R$ return the subsequences consisting of odd and even elements of the input sequence, respectively.
Thus, $\pi_L((s_1, s_2, s_3, s_4, \dots)) = (s_1, s_3, \dots)$ and $\pi_R((s_1, s_2, s_3, s_4, \dots)) = (s_2, s_4, \dots)$.
\item The function $\mathrel{::}$ interleaves the two input sequences, so that 
$(s_1, s_2, s_3, \dots) \mathrel{::} (t_1, t_2, t_3, \dots) = (s_1, t_1, s_2, t_2, s_3, t_3, \dots)$.
\end{itemize}

Observe that the functions $\pi_L$ and $\pi_R$ return two \emph{disjoint} infinite subsequences of the input sequence.
This means that if we want to perform two random computations, but only have a single entropy $s = (s_1, s_2, s_3, s_4 \dots)$,
we can perform the first computation with the sequence $\pi_L((s_1, s_2, s_3, s_4 \dots)) = (s_1, s_3, \dots)$ and
the second one with $\pi_R((s_1, s_2, s_3, s_4 \dots)) = (s_2, s_4, \dots)$ and no entropy component from the first 
computation will be reused in the second one. In other words, for each new random sample we will have a ``fresh'' value in the entropy.

%(infinite sequences of numbers in $[0,1]$) with the function $\pi_U$ returning the first element of the sequence
%and $\pi_L$ and $\pi_R$ returning the subsequences consisting of odd and even elements of the original sequence, respectively.

Results presented in this paper will, however, only depend on the abstract definition of entropy space.

\subsection{Extended state space}

In order to define the operational semantics and the distributions induced by it, we need to extend the set of states  $\statespace$ with two exception states: $\failure$, denoting a \emph{failed hard constraint or an evaluation error}, and $\diverge$, denoting \emph{divergence}. 
We denote this extended space by $\fullstatespace$. 
A metric space on $\fullstatespace$ is defined by extending the metric $d_\sigma$ on $\statespace$ to $\fullstatespace$ as follows:
\[
\hat{d}_\sigma(\sigma_1, \sigma_2) =
\begin{cases}
0 & \text{if } \sigma_1 = \sigma_2 \in \{ \failure, \diverge \} \\
%% 0 & \text{if } \sigma_1 = \sigma_2 = \diverge \\
d_\sigma(\sigma_1, \sigma_2) & \text{if } \sigma_1, \sigma_2 \in \statespace \\
\infty & \text{otherwise.}
\end{cases}
\]
It is easy to check that the extended metric space $(\fullstatespace, \hat{d}_\sigma)$ is separable.
The $\sigma$-algebra $\fullstatesa$ on $\fullstatespace$ is then induced by the metric $\hat{d}_\sigma$ and $(\fullstatespace, \fullstatesa)$ is the measurable space of all program states.

In the remainder of this section, we will use two operators to extend real-valued functions on $\statespace$ to the state space $\fullstatespace$: for each function $f \colon \statespace -> \extposreals$, the extended functions $\hat{f}, \check{f} \colon \fullstatespace -> \extposreals$ are defined as follows:
\[
\hat{f}(\tau) =
\begin{cases}
f(\tau) & \text{if}\ \tau \in \statespace
\\
0 & \text{otherwise}
\end{cases}
\quad
\mbox{and} 
\quad
\check{f}(\tau) = \begin{cases} f(\tau)\ & \text{if}\ \tau \in \statespace \\ 
1\  & \text{if}\  \tau = \, \diverge \\
0\  & \text{otherwise}. \end{cases}
\]

%A $\sigma$-algebra
%$\fullstatesa$ on $\fullstatespace$ can be easily defined by taking a disjoint union of $\statesa$ and
%the powerset $\sigma$-algebra on $\{ \failure,  \diverge \}$. Obviously, 

\subsection{Reduction relation}
\newcommand{\config}[7]{\langle #1, \allowbreak #2, \allowbreak #3, \allowbreak #4, \allowbreak #5, \allowbreak #6, 
		 \allowbreak #7 \rangle}
%These commands allow splitting configurations into several lines, if they do not fit on one. 
\newcommand{\configLa}[4]{\langle #1, #2 ,#3, #4,}
\newcommand{\configRa}[3]{#1, #2, #3 \rangle}
\newcommand{\configLb}[3]{\langle #1, #2 ,#3,}
\newcommand{\configRb}[4]{#1, #2, #3, #4 \rangle}

To ensure that the entropy is split correctly between partial computations, we use continuations, similarly to~\cite{DBLP:journals/pacmpl/WandCGC18}. 
A continuation is represented by a list of expressions that are to be evaluated after completing the current evaluation. 
We keep track of two distinct entropies: one for the current computation and one to be used when evaluating the continuation.

The reduction relation is defined as a binary relation on \emph{configurations}, i.e., tuples of the form
$$\config{\theta}{C}{K}{\sigma}{\theta_K}{n}{w}$$ 
where $C$ is the current program statement to be evaluated, $\sigma$ is the current program state, $K$ is the continuation, $\theta$ and $\theta_K$ are, respectively, the entropies to be used when evaluating $C$ and the continuation $K$; finally, $n \in \mathbb{N}$ is the number of %evaluation steps taken so far 
reduction rules applied so far and $w \in \mathbb{R} \cap [0,1]$ is the weight of the current program run so far. 
In order to access elements of a configuration $\kappa$, we use functions, e.g., for $\kappa$ as given above $\mbox{\sf weight}(\kappa) = w$ and $\mbox{\sf state}(\kappa) = \sigma$.
The \emph{reduction relation} $\vdash$ is a binary relation on configurations where  
$$
\underbrace{\config{\theta}{C}{K}{\sigma}{\theta_K}{n}{w}}_{\mbox{\footnotesize configuration } \kappa}
\, \vdash \, 
\underbrace{\config{\theta'}{C'}{K'}{\sigma'}{\theta'_K}{n'}{w'}}_{\mbox{\footnotesize configuration } \kappa'}
$$ 
means that the configuration $\kappa$ reduces to configuration $\kappa'$ in one step. 
Let $\vdash^{*}$ denote the reflexive and transitive closure of the reduction relation $\vdash$, i.e., $\kappa \vdash^{*} \kappa'$ means that $\kappa$ reduces to $\kappa'$ in zero or more reduction steps.
%TODO: Format rule names differently?

We present the reduction rules one-by-one for each syntactic construct of \ccpgclns.
The symbol $\modownarrow$ means successful termination and $\failure$ is a special state reached after a failed observation or an execution error.
\begin{itemize}
\item[Skip.]
As the \texttt{skip} statement cannot do anything, it has no reduction rule.
\item[Divergence.]
The (diverge) rule states that the $\mathtt{diverge}$ statement reduces to itself indefinitely in any non-failure state $\sigma$:
$$
\inference[diverge]
{\sigma \neq \failure}
{\config{\theta}{\mathtt{diverge}}{K}{\sigma}{\theta_K}{n}{w} \ \vdash \ 
\config{\theta}{\mathtt{diverge}}{K}{\sigma}{\theta_K}{n{+}1}{w}
}
$$ 
\item[Assignment.]
The rule (assign) evaluates the expression $E$ in the current state $\sigma$ and sets the value of $x$ in the state to the outcome of this evaluation:
$$
\inference[assign]
{\sigma \neq \failure \qquad \sigma(E) = V}
{\config{\theta}{x := E}{ K}{\sigma}{\theta_K}{n}{w}
\ \vdash \
\config{\theta}{\modownarrow}{ K}{\sigma[x \mapsto V]}{\theta_K}{n{+}1}{w}
}
$$
\item[Random draw.]
The rule (draw) evaluates a random draw from a uniform distribution:
$$
\inference[draw]
%{ v \in [0,1] }
{\sigma \neq \failure}
{\config{\theta}{x :\approx U}{K}{\sigma}{\theta_K}{n}{w} 
\, \vdash \, 
\config{\pi_R(\theta)}{\modownarrow}{K}{\sigma[x \mapsto \pi_U(\pi_L(\theta))]}{\theta_K}{n{+}1}{w}
}
$$
The outcome of this random draw is determined by the entropy $\theta$ and is set to $\pi_U(\pi_L(\theta))$, intuitively the first element of the ``left" part of the entropy~$\theta$~\footnote{If we set this value to just $\pi_U(\theta)$, we would lose the  property that an already used ``element'' of the entropy cannot appear in the entropy in the subsequent configuration, because we do not know what parts of $\theta$ the value of  $\pi_U(\theta)$ depends on.  In the Hilbert cube implementation discussed before, $\pi_U(\theta)$ is equivalent to  $\pi_U(\pi_L(\theta))$ and ``disjoint'' from $\pi_R(\theta)$, but if we defined $\pi_U(\theta)$ to be, for instance, the second element of the sequence encoded by $\theta$, this would not be the case. Obviously, this does not matter in practice, as after the (draw) rule, the expression to be evaluated with entropy $\pi_R(\theta)$ is empty, but it is still elegant to keep this property.}.
The value of the sampled variable is assigned to variable $x$. 
The weight $w$ is unchanged, as the density of the uniform distribution on $[0,1]$ is constant and equal to $1$ for every point in the unit interval---as all outcomes are equally likely, there is no need to weigh the program runs.
\item[Hard.]
The rules (condition-true) and (condition-false) evaluate a hard condition in an $\mathtt{observe}$ statement. If the condition is satisfied, (condition-true) returns the current state and weight unchanged, otherwise the state is set to the error state $\failure$ by (condition-false):
$$
\inference[condition-true]
{\sigma \neq \failure \qquad \sigma(\phi) = \mathtt{true}}
{
\config{\theta}{\mathtt{observe}(\phi)}{K}{\sigma}{\theta_K}{n}{w}
\ \vdash \
\config{\theta}{\modownarrow}{K}{\sigma}{\theta_K}{n{+}1}{w}
}
$$
$$
\inference[condition-false]
{\sigma \neq \failure \qquad \sigma(\phi) = \mathtt{false}}
{\config{\theta}{\mathtt{observe}(\phi)}{K}{\sigma}{\theta_K}{n}{w}
\ \vdash \
\config{\theta}{\modownarrow}{[]}{\failure}{\theta_K}{n{+}1}{w}} 
$$
\item[Soft.]
The rule (score) evaluates its argument, a real number in the unit interval, and multiplies it by the weight of the current run so far:
$$
\inference[score]
{\sigma \neq \failure \qquad v = \sigma(E) \in (0,1] }
{\config{\theta}{\mathtt{score}(E)}{K}{\sigma}{\theta_K}{n}{w} \vdash 
\config{\theta}{\modownarrow}{K}{\sigma}{\theta_K}{n{+}1}{w{\cdot}v} 
}
$$
\item[Sequencing.]
The rule (seq) is used to move statements from the current statement $C$ to the continuation $K$:
$$
\inference[seq]
{\sigma \neq \failure \qquad C_1 \neq C_1'; C_1''}
{\config{\theta}{C_1;C_2}{K}{\sigma}{\theta_K}{n}{w}
\ \vdash \
\config{\pi_L(\theta)}{C_1}{C_2 \mathrel{::} K}{\sigma}{\pi_R(\theta) \mathrel{::} \theta_K}{n{+}1}{w}
}
$$
If the current statement is a sequence of statements, (seq) splits it into the first statement $C_1$ and the sequence of remaining statements $C_2$ in such a way that $C_1$ itself is a single concrete statement and not a sequence of statements---in other words, $C_1$ is as small as possible. 
The expression $C_1$ is then retained as the current expression to be evaluated, while $C_2$ is pushed onto the top of the expression stack in the continuation $K$. 
The expression $C_1$ is evaluated with only the ``left'' part of the entropy $\theta$, and the right part is appended to the entropy of the continuation $K$; it is stored to be used later when $C_2$ is popped from the stack and evaluated.
The reason that $C_1$ is required not to be a sequence is to ensure that there is a unique way to split the sequence of statements into $C_1$ and $C_2$. 
If the entropy could be split in different ways into sub-computations, this would make the 
semantics nondeterministic.
Note that $C_1$ may be, e.g., an $\mathtt{if}$-statement or a $\mathtt{while}$ loop which includes a sequence of statements as its sub-expression; we only require that it is not a sequence at the top level.
The rule
$$
\inference[pop]
{\sigma \neq \failure}
{\config{\theta}{\modownarrow}{C \mathrel{::} K}{\sigma}{\theta_K}{n}{w}
\ \vdash \
\config{\pi_L(\theta_K)}{C}{K}{\sigma}{\pi_R(\theta_K)}{n{+}1}{w}
}
$$
is the dual of (seq). 
After the current statement has been completely evaluated, (pop) fetches the top statement $C$ 
from the continuation $K$ and sets it as the next statement to be evaluated. 
The unused part $\theta$ of the entropy used in evaluating the last expression is discarded and replaced by the left part of the continuation entropy $\theta_K$. 
If the evaluation started with an empty continuation, $\pi_L(\theta_K)$ will be the entropy ``reserved'' for evaluating $C$ as it was pushed on the continuation by (seq). 
Obviously, the entropy reserved for $C$ has to be removed from the entropy saved for evaluating the rest of the continuation, hence the latter is set to $\pi_R(\theta_K)$.
\item[Conditional.]
The rules (if-true) and (if-false) are standard and self-explanatory. 
$$
\inference[if-true]
{\sigma \neq \failure \qquad \sigma(\phi) = \mathtt{true}}
{
\config{\theta}{\mathtt{if}(\phi)\{ C\}}{K}{\sigma}{\theta_K}{n}{w}
\vdash
\config{\theta}{C}{K}{\sigma}{\theta_K}{n{+}1}{w}
}
$$
$$
\inference[if-false]
{\sigma \neq \failure \qquad \sigma(\phi) = \mathtt{false}}
{
\config{\theta}{\mathtt{if}(\phi)\{ C\}}{K}{\sigma}{\theta_K}{n}{w}
\vdash
\config{\theta}{\modownarrow}{K}{\sigma}{\theta_K}{n{+}1}{w}
}
$$
\item[Loops.]
The (while-true) and (while-false) rules are standard too.
If the loop-guard $\phi$ is true, the loop body is to be executed possibly followed by the loop itself.
Otherwise the loop terminates. The (while-true) rule reads
$$
\inference
{\sigma \neq \failure \qquad \sigma(\phi) = \mathtt{true}}
{
\config{\theta}{\mathtt{while}(\phi)\{ C\}}{K}{\sigma}{\theta_K}{n}{w}
\ \vdash \
\config{\theta}{C;\mathtt{while}(\phi)\{ C\}}{K}{\sigma}{\theta_K}{n{+}1}{w}
}
$$
$$
\inference[while-false]
{\sigma \neq \failure \qquad \sigma(\phi) = \mathtt{false}}
{
\config{\theta}{\mathtt{while}(\phi)\{ C\}}{K}{\sigma}{\theta_K}{n}{w}
\ \vdash \
\config{\theta}{\modownarrow}{K}{\sigma}{\theta_K}{n{+}1}{w}
}
$$
\end{itemize}
%%
%% The reduction rules are shown in Figure~\ref{figure:operational}.  
The (final) rule is a dummy rule which applies to fully evaluated programs.
It does nothing, except for increasing the step counter. 
Its purpose is to allow for reasoning about infinite evaluations, as explained later. 
$$
\inference[final]
{\sigma \neq \failure}
{\config{\theta}{\modownarrow}{[]}{\sigma}{\theta_K}{n}{w} \vdash 
\config{\theta}{\modownarrow}{[]}{\sigma}{\theta_K}{n{+}1}{w}
} 
$$

The initial configuration for program $C$ is of the form $\config{\theta}{C}{[]}{\sigma}{\theta_K}{0}{1} $ with the initial statement $C$, the empty continuation,  initial state $\sigma$, a zero step count and initial weight one.
Note that if the initial continuation is $[]$, the initial continuation entropy $\theta_K$ is irrelevant, as it can never be copied to the entropy of the current expression.
A program is considered fully evaluated when the evaluation reaches a configuration where $C = \, \modownarrow$ and $K = []$.
In this case, only the dummy (final) rule can be applied. 
Thus, if
$$
\config{\theta}{C}{[]}{\sigma}{\theta_K}{0}{1} \ \vdash^{*} \ 
\config{\theta'}{\modownarrow}{[]}{\sigma'}{\theta'_K}{n}{w}
$$
we say that the program $C$ with initial state $\sigma$ under entropy $\theta$ terminates in $n$ steps in the state $\sigma'$ with weight $w$. 

%\subsection{Properties of the semantics}
%
%Having fully defined the operational semantics, which reduces a program deterministically
%given the fixed entropy and initial state, we now prove some key properties of the semantics.
%
%TODO

\paragraph{Examples} 
We demonstrate how the semantics works by revisiting two of the examples in Section~\ref{section:denotational-semantics}.
For clarity, we now add line numbers to programs and
write $C_i$ for line $i$ of the given program $C$ and $C_{i,j}$ for the part of the program between lines $i$ and $j$.
We also write $\pi_{d_1, \dots, d_n}(\theta)$ (where $d_1, \dots, d_n \in \{L,R\}$) for 
$\pi_{d_1} ( \pi_{d_2} \dots (\pi_{d_n}(\theta)\dots))$.

\begin{example} \label{example:blr-os}
%An example with a trace of one of the example programs used before in the paper.
%\marginpar{TBD}

We begin by revisiting the Bayesian linear regression example:

\begin{verbatim}
1        u1 := U;
2        a := Gaussian_inv_cdf(0,2,u1);
3        u2 := U;
4        b := Gaussian_inv_cdf(0,2,u2);
5        score(softeq(a*0 + b, 2));
6        score(softeq(a*1 + b, 3));
\end{verbatim}

Let us suppose we want to evaluate this program with an empty initial state and with an entropy $\theta$ such that
the two values sampled in the program (which are $\pi_{U}(\pi_{L,L}(\theta))$ and $\pi_{U}(\pi_{L,L,R,R}(\theta))$), are, 
respectively, 0.5 and $v$, where $v \in (0,1)$ is a value such that
$\mathtt{Gaussian\_inv\_cdf}(0,2,v) = 2$. 
For the particular instantiation of the entropy space shown in Section~\ref{subsection:entropy-space},
we have $\pi_{U}(\pi_{L,L}((s_1, s_2, s_3, \dots ))) = s_1$ and $\pi_{U}(\pi_{L,L,R,R}((s_1, s_2, s_3, \dots ))) = s_{13}$,
so we can assume that $\theta$ is any infinite sequence $ (s_1, s_2, s_3, \dots )$ whose first element is $0.5$ and thirteenth element is $v$.
Note that $\mathtt{Gaussian\_inv\_cdf}(0,2,0.5) = 0$, because the Gaussian
distribution is symmetric, so exactly half of the total probability mass is below the mean.

The evaluation chain is shown below. We use colour to highlight states and scores which have changed from the previous
configuration. Since the evaluation starts with a configuration with empty continuation, the initial continuation entropy
$\theta_{K}$ is a ``dummy'' entropy whose values are irrelevant and do not affect the computation.

%We highlight components of a configuration which have changed from the previous
%configuration (other than terms and continuations, which change in ba)
%colour to highlit components of configurations 

\allowdisplaybreaks
\begin{eqnarray*}
&&\config{\theta}{\mathtt{u1}:\approx U;C_{2,6} }{[]}{[]}{\theta_K}{0}{1} \ \vdash \\
\text{(seq)}&&\config{{ \pi_L(\theta)}}{\mathtt{u1}:\approx \mathtt{U} }{[C_{2,6}]}{[]}{{ \pi_R(\theta) \mathrel{::} \theta_K}}{1}{1} \ \vdash \\
\text{(draw)}&&\config{{ \pi_{R,L}(\theta)}}{\modownarrow }{[C_{2,6}]}{{\color{red} [\mathtt{u1} \mapsto 0.5]}}
    { \pi_R(\theta) \mathrel{::} \theta_K}{2}{1} \ \vdash \\
\text{(pop)}&&\config{{\pi_R(\theta)}}{\mathtt{a := G(0,2,u1) };C_{3,6} }{[]}{[\mathtt{u1} \mapsto 0.5]}{{ \theta_K}}{3}{1} \ \vdash \\
\text{(seq)}&&\config{{ \pi_{L,R}(\theta)}}{\mathtt{a := G(0,2,u1) }}{[C_{3,6}]}{[\mathtt{u1} \mapsto 0.5]}
    {{ \pi_{R,R}(\theta) \mathrel{::} \theta_K}}{4}{1} \ \vdash \\
\text{(assign)}&&\config{{\pi_{L,R}(\theta)}}{\modownarrow}{[C_{3,6}]}{
{\color{red}[\mathtt{u1}, \mathtt{a} \mapsto 0.5, 0]}}
    {{ \pi_{R,R}(\theta) \mathrel{::} \theta_K}}{5}{1} \ \vdash \\
\text{(pop)}&&\config{{ \pi_{R,R}(\theta)}}{\mathtt{u2}:\approx \mathtt{U};C_{4,6}}{[]}
  {[\mathtt{u1}, \mathtt{a} \mapsto 0.5, 0]}
  {{ \theta_K}}{6}{1} \ \vdash \\
\text{(seq)}&&\config{{ \pi_{L,R,R}(\theta)}}{\mathtt{u2}:\approx \mathtt{U}}{[C_{4,6}]}
  {[\mathtt{u1}, \mathtt{a} \mapsto 0.5, 0]}
  {{ \pi_{R,R,R}(\theta) \mathrel{::} \theta_K}}{7}{1} \ \vdash \\
\text{(draw)}&&\config{{ \pi_{R,L,R,R}(\theta)}}{\modownarrow}{[C_{4,6}]}
  {{\color{red} [\mathtt{u1}, \mathtt{a}, \mathtt{u2} \mapsto 0.5, 0, v]}}
  {{ \pi_{R,R,R}(\theta) \mathrel{::} \theta_K}}{8}{1} \ \vdash \\
\text{(pop)}&&\config{{ \pi_{R,R,R}(\theta)}}{\texttt{b := G(0,2,u2)};C_{5,6}}{[]}
  {{ [\mathtt{u1}, \mathtt{a}, \mathtt{u2} \mapsto 0.5, 0, v]}}
  {{ \theta_K}}{9}{1} \ \vdash \\
\text{(seq)}&&\configLa{{ \pi_{L,R,R,R}(\theta)}}{\texttt{b := G(0,2,u2)}}{[C_{5,6}]}
  {{ [\mathtt{u1}, \mathtt{a}, \mathtt{u2} \mapsto 0.5, 0, v]}}\\
&& \qquad \configRa{{ \pi_{R,R,R,R}(\theta) \mathrel{::} \theta_K}}{10}{1} \ \vdash \\
\text{(assign)}&&\configLa{{ \pi_{L,R,R,R}(\theta)}}{\modownarrow}{[C_{5,6}]}
  {{\color{red} [\mathtt{u1}, \mathtt{a}, \mathtt{u2}, \mathtt{b} \mapsto 0.5, 0, v, 2]}}\\
&&\qquad \configRa{{ \pi_{R,R,R,R}(\theta) \mathrel{::} \theta_K}}{11}{1} \ \vdash \\
\text{(pop)}&&\config{{ \pi_{R,R,R,R}(\theta)}}{C_{5,6}}{[]}
  {{ [\mathtt{u1}, \mathtt{a}, \mathtt{u2}, \mathtt{b} \mapsto 0.5, 0, v, 2]}}
  {{ \theta_K}}{12}{1} \ \vdash \\
\text{(seq)}&&\configLb{{ \pi_{L,R,R,R,R}(\theta)}}{\mathtt{ score(softeq(a*0 + b, 2))}}{[C_6]}  \\
&&\qquad \configRb{{ [\mathtt{u1}, \mathtt{a}, \mathtt{u2}, \mathtt{b} \mapsto 0.5, 0, v, 2]}}
  {{ \pi_{R,R,R,R,R}(\theta) \mathrel{::} \theta_K}}{13}{1} \ \vdash \\
\text{(score)}&&\configLa{{ \pi_{L,R,R,R,R}(\theta)}}{\modownarrow}{[C_6]}
  {{ [\mathtt{u1}, \mathtt{a}, \mathtt{u2}, \mathtt{b} \mapsto 0.5, 0, v, 2]}}\\
&& \qquad  \configRa{{ \pi_{R,R,R,R,R}(\theta) \mathrel{::} \theta_K}}{14}{1} \ \vdash \\
\text{(pop)}&&\configLb{{ \pi_{R,R,R,R,R}(\theta)}}{\mathtt{score(softeq(a*1 + b, 3))}}{[]}\\
&& \qquad \configRb{{ [\mathtt{u1}, \mathtt{a}, \mathtt{u2}, \mathtt{b} \mapsto 0.5, 0, v, 2]}}
  {{ \theta_K}}{15}{1} \ \vdash \\
\text{(score)}&&\config{{ \pi_{R,R,R,R,R}(\theta)}}{\modownarrow}{[]}
  {{ [\mathtt{u1}, \mathtt{a}, \mathtt{u2}, \mathtt{b} \mapsto 0.5, 0, v, 2]}}
  {{ \theta_K}}{16}{{ \color{red} e^{-1}}} .
\end{eqnarray*}
Hence, with the given entropy $\theta$, the program evaluates with score $e^{-1}$ to a state
where $\mathtt{a} = 0$ and $\mathtt{b} = 0$.

\end{example}

\begin{example} \label{example:closed-form-os}
Let us now consider the program in Example 2:
\begin{verbatim}
1        b := 0;
2        k := 0;
3        while (b=0)
        {
4          u := U;
5          k := k+1;
6          if(u < 1/(k+1)^2)
          {
7            b := 1;
8            score(k/(k+1))
          }
        }
\end{verbatim}

We want to compute the final state and weight for this program, assuming an empty initial state
and an entropy $\theta$ such that the first value drawn (that is, $\pi_U(\pi_{L,L,R,R}(\theta))$) is $0.1$,
which also means the loop terminates after the first iteration. The evaluation proceeds as follows:
\allowdisplaybreaks
\begin{eqnarray*}
&&\config{\theta}{\mathtt{b}:=0;C_{2,8} }{[]}{[]}{\theta_K}{0}{1} \ \vdash \\
\text{(seq)}&&\config{{ \pi_L(\theta)}}{\mathtt{b}:=0 }{[C_{2,8}]}{[]}{{\color{red} \pi_R(\theta) \mathrel{::} \theta_K}}{1}{1} \ \vdash \\
\text{(assign)}&&\config{{ \pi_L(\theta)}}{\modownarrow }{[C_{2,8}]}{{\color{red} [\mathtt{b} \mapsto 0]}}
    { \pi_R(\theta) \mathrel{::} \theta_K}{2}{1} \ \vdash \\
\text{(pop)}&&\config{{\pi_R(\theta)}}{C_{2,8} }{[]}{[\mathtt{b} \mapsto 0]}{{ \theta_K}}{3}{1} \ \vdash \\
\text{(seq)}&&\config{{ \pi_{L,R}(\theta)}}{\mathtt{k} :=0 }{[C_{3,8}]}{[\mathtt{b} \mapsto 0]}
    {{ \pi_{R,R}(\theta) \mathrel{::} \theta_K}}{4}{1} \ \vdash \\
\text{(assign)}&&\config{{ \pi_{L,R}(\theta)}}{\modownarrow }{[C_{3,8}]}
  {{\color{red} [\mathtt{b}, \mathtt{k} \mapsto 0, 0]}}{ \pi_{R,R}(\theta) \mathrel{::} \theta_K}{5}{1} \ \vdash \\
\text{(pop)}&&\config{{ \pi_{R,R}(\theta)}}{\mathtt{while}(\mathtt{b}=0)\{C_{4,8}\}}{[]}
  { [\mathtt{b}, \mathtt{k} \mapsto 0, 0]}{{ \theta_K}}{6}{1} \ \vdash \\
\text{(while-true)}&&\config{\pi_{R,R}(\theta)}{C_{4};
C_{5,8};\mathtt{while}(\mathtt{b}=0)\{C_{4,8}\}}{[]}
  { [\mathtt{b}, \mathtt{k} \mapsto 0, 0]}{\theta_K}{7}{1} \ \vdash \\
\text{(seq)}&&\configLa{{ \pi_{L,R,R}(\theta)}}{\mathtt{u} :\approx U }{[C_{5,8};\mathtt{while}(\mathtt{b}=0)\{C_{4,8}\}]}
  { [\mathtt{b}, \mathtt{k} \mapsto 0, 0]}\\
  && \qquad \configRa{{\color{red} \pi_{R,R,R}(\theta) \mathrel{::} \theta_K}}{8}{1} \ \vdash \\
\text{(draw)}&&\configLa{{\pi_{R,L,R,R}(\theta)}}{\modownarrow}{[C_{5,8};\mathtt{while}(\mathtt{b}=0)\{C_{4,8}\}]}
  {{\color{red} [\mathtt{b}, \mathtt{k}, \mathtt{u} \mapsto 0, 0, 0.1]}}\\
  && \qquad \configRa{ \pi_{R,R,R}(\theta) \mathrel{::} \theta_K}{9}{1} \ \vdash \\
\text{(pop)}&&\configLb{{\pi_{R,R,R}(\theta)}}{\mathtt{k} := \mathtt{k}+1; C_{6,8};\mathtt{while}(\mathtt{b}=0)\{C_{4,8}\}}{[]}\\
&& \qquad  \configRb{ [\mathtt{b}, \mathtt{k}, \mathtt{u} \mapsto 0, 0, 0.1]}{{ \theta_K}}{10}{1} \ \vdash \\
\text{(seq)}&&\configLb{{\pi_{L,R,R,R}(\theta)}}{\mathtt{k} := \mathtt{k}+1}{[C_{6,8};\mathtt{while}(\mathtt{b}=0)\{C_{4,8}\}]}\\
&& \qquad \configRb
  { [\mathtt{b}, \mathtt{k}, \mathtt{u} \mapsto 0, 0, 0.1]}{{\pi_{R,R,R,R}(\theta) \mathrel{::}  \theta_K}}{11}{1} \ \vdash \\
\text{(assign)}&&\configLa{{\pi_{L,R,R,R}(\theta)}}{\modownarrow} {[C_{6,8};\mathtt{while}(\mathtt{b}=0)\{C_{4,8}\}]}
{ [\mathtt{b}, \mathtt{k}, \mathtt{u} \mapsto 0, 1, 0.1]}\\
  && \qquad \configRa{{\pi_{R,R,R,R}(\theta) \mathrel{::}  \theta_K}}{12}{1} \ \vdash \\
\text{(pop)}&&\configLa{{\pi_{R,R,R,R}(\theta)}}{C_{6,8};\mathtt{while}(\mathtt{b}=0)\{C_{4,8}\}}{[]}
  { [\mathtt{b}, \mathtt{k}, \mathtt{u} \mapsto 0, 1, 0.1]}\\
  && \qquad\configRa{{\theta_K}}{13}{1} \ \vdash \\
\text{(seq)}&&\configLb{{\pi_{L,R,R,R,R}(\theta)}}{\mathtt{if(u < 1/(k+1)^2)}\{C_{7,8}\}}{[\mathtt{while}(\mathtt{b}=0)\{C_{4,8}\}]}\\
&& \qquad \configRb{ [\mathtt{b}, \mathtt{k}, \mathtt{u} \mapsto 0, 1, 0.1]}
     {{\pi_{R,R,R,R,R}(\theta) \mathrel{::}  \theta_K}}{14}{1} \ \vdash \\
\text{(if-true)}&&\configLb{{\pi_{L,R,R,R,R}(\theta)}}{\mathtt{b}:=1;C_{8}}{[\mathtt{while}(\mathtt{b}=0)\{C_{4,8}\}]}\\
&&  \qquad \configRb{ [\mathtt{b}, \mathtt{k}, \mathtt{u} \mapsto 0, 1, 0.1]}{{\pi_{R,R,R,R,R}(\theta) \mathrel{::}  \theta_K}}{15}{1} \ \vdash \\
\text{(seq)}&&\configLa{{\pi_{L,L,R,R,R,R}(\theta)}}{\mathtt{b}:=1}{[C_{8}, \mathtt{while}(\mathtt{b}=0)\{C_{4,8}\}]}
  { [\mathtt{b}, \mathtt{k}, \mathtt{u} \mapsto 0, 1, 0.1]}\\
&&\qquad \configRa{{
   \pi_{R,L,R,R,R,R}(\theta) \mathrel{::}(\pi_{R,R,R,R,R}(\theta) \mathrel{::}  \theta_K)}}{16}{1} \ \vdash \\
\text{(assign)}&&\configLa{{\pi_{L,L,R,R,R,R}(\theta)}}{\modownarrow}{[C_{8}, \mathtt{while}(\mathtt{b}=0)\{C_{4,8}\}]}
  {{\color{red} [\mathtt{b}, \mathtt{k}, \mathtt{u} \mapsto 1,1, 0.1]}}\\
  && \qquad \configRa{{ \pi_{R,L,R,R,R,R}(\theta) \mathrel{::}(\pi_{R,R,R,R,R}(\theta) \mathrel{::}  \theta_K)}}{17}{1} \ \vdash \\
\text{(pop)}&&\configLb{{\pi_{R,L,R,R,R,R}(\theta)}}{\mathtt{score(k/(k+1))}}{[\mathtt{while}(\mathtt{b}=0)\{C_{4,8}\}]} \\
&& \qquad  \configRb{ [\mathtt{b}, \mathtt{k}, \mathtt{u} \mapsto 1, 1, 0.1]}
  {{ \pi_{R,R,R,R,R}(\theta) \mathrel{::}  \theta_K}}{18}{1} \ \vdash \\
\text{(score)}&&\configLa{{\pi_{R,L,R,R,R,R}(\theta)}}{\modownarrow}{[\mathtt{while}(\mathtt{b}=0)\{C_{4,8}\}]}
  { [\mathtt{b}, \mathtt{k}, \mathtt{u} \mapsto 1, 1, 0.1]}\\
&& \qquad  \configRa{{ (\pi_{R,R,R,R,R}(\theta) \mathrel{::}  \theta_K)}}{19}{{\color{red} \nicefrac{1}{2}}} \ \vdash \\
\text{(pop)}&&\config{{\pi_{R,R,R,R,R}(\theta)}}{\mathtt{while}(\mathtt{b}=0)\{C_{4,8}\}}{[]}
  { [\mathtt{b}, \mathtt{k}, \mathtt{u} \mapsto 1, 1, 0.1]}{{
     \theta_K}}{20}{\nicefrac{1}{2}} \ \vdash \\
 \text{(while-false)}&&\config{{\pi_{R,R,R,R,R}(\theta)}}{\modownarrow}{[]}
  { [\mathtt{b}, \mathtt{k}, \mathtt{u} \mapsto 1, 1, 0.1]}{{
     \theta_K}}{21}{\nicefrac{1}{2}}.
\end{eqnarray*}
Hence, the program evaluates to a state where $\mathtt{k} = 1$ with score $\nicefrac{1}{2}$.

\end{example}

\subsection{Measure on final program states}

As mentioned before, the operational semantics so far only defines the final state and weight for a single program execution, for a fixed entropy. 
We now explain how a probability distribution on final states can be obtained by integrating the semantics over the entropy space.

%% \subsubsection{Operational semantics as total functions}

\paragraph{Two auxiliary functions.}
In order to define a probability distribution over the final program states, two auxiliary functions are technically convenient.
The function $\mathbf{O}_C^{\sigma}\colon \mathbb{S} -> \fullstatespace$ determines the \emph{final state} of program $C$ with initial state $\sigma$ for entropy $\theta \in \mathbb{S}$.
It is defined by:
\begin{eqnarray*}
\mathbf{O}_C^{\diverge}(\theta) &\ = \ &\diverge \\
%% \mathbf{O}_C^\failure(\theta) &=& \failure\\
\mathbf{O}_C^{\sigma}(\theta) &=&
\begin{cases}
\tau & \text{if}\  \config{\theta}{C}{[]}{\sigma}{\theta_K}{0}{1} \vdash^{*}
\config{\theta'}{\modownarrow}{[]}{\tau}{\theta_K}{n }{w}\ \text{and}\ \tau \neq \failure \\
\failure & \text{if}\  \config{\theta}{C}{[]}{\sigma}{\theta_K}{0}{1} \vdash^{*}
\config{\theta'}{C'}{K}{\tau}{\theta_K'}{n }{w} \nvdash \\ % \ \text{and}\ (C',K) \neq (\modownarrow,[]) \\
\diverge & \text{otherwise}.
\end{cases} \\
\end{eqnarray*}
The final state of program $C$ and proper initial state $\sigma$ (i.e., $\sigma \neq \, \diverge$ and $\sigma \neq \failure$) with entropy $\theta$ equals state $\tau$ provided its execution ends in configuration $\langle{\cdot}, [], \tau, {\cdots} \rangle$.
For instance, for the program $C$ and entropy $\theta$ from Example \ref{example:closed-form-os}, we have
$\mathbf{O}_C^{[]}(\theta) = [\mathtt{b} \mapsto 1, \mathtt{k} \mapsto 1, \mathtt{u} \mapsto 0.1 ]$.
If the evaluation reaches a configuration which cannot be reduced any further (e.g., due to a failed hard constraint), the final state equals the error state $\failure$.
Note that this is also applicable to the initial state $\sigma = \failure$. 
Finally, if the evaluation can neither be completed nor reach an irreducible configuration, this means that the evaluation of program $C$ with state $\sigma$ diverges with entropy $\theta$.
This results in the extended state $\diverge$.
For an example of a computation leading to such a state, consider Example~\ref{example:closed-form-os} again.
If we take an entropy $\theta'$ such that the value sampled in line 4 in each iteration is $0.5$, the evaluation will
never terminate and so we have $\mathbf{O}_C^{[]}(\theta') = \, \diverge$.
The final state of running $C$ from initial state $\modownarrow$ is the initial state.
This way of handling exceptions ensures compositionality. 
We omit $C$ and $\sigma$ as sub- and superscript if they are clear from the context.

Besides the final state obtained from a program's run, we need also the \emph{run's score}.
The function $\mathbf{SC}_C^{\sigma}\colon \mathbb{S} -> \mathbb{R}_{+}$ yields the score of executing program $C$ from initial state $\sigma$ for a given entropy. 
This definition is a bit more complicated, due to the handling of diverging runs.
A naive solution would be to define $\mathbf{SC}$ similarly to $\mathbf{O}$ and return $0$ for diverging runs.
This would, however, mean that the semantics would quietly ignore diverging runs, while a key motivation for this work is to handle divergence in the presence of soft conditioning in a meaningful way.
Our proposal is to let the score for diverging runs be the \emph{limit} of the weight $w$ as the number of steps $n$ goes to infinity.
Formally, this is done as follows:
Let us define an approximation function $\mathbf{SC}_C^{\sigma}\colon \mathbb{S} \times \mathbb{N} -> \mathbb{R}_{+}$, such that $\mathbf{SC}_C^{\sigma}(\theta, n)$ returns the score 
for program $C$ with entropy $\theta$ and initial state $\sigma$
after $n$ evaluation steps:
\[
\mathbf{SC}_C^{\sigma}(\theta,n) =
\begin{cases}
w & \text{if}\ \config{\theta}{C}{[]}{\sigma}{\theta_K}{0}{1} \vdash^{*}
\config{\theta'}{C'}{K}{\tau}{\theta_K'}{n }{w} \ \text{and}\ \tau \neq \failure \\
0 & \text{otherwise}.
\end{cases} 
\]
The function  $\mathbf{SC}_C^{\sigma}\colon \mathbb{S} -> \mathbb{R}_{+}$ is now defined for proper state $\sigma$ as the limit, or equivalently infimum, of its $n$-the approximation:
\[
\mathbf{SC}_C^{\sigma}(\theta) \ = \ 
\lim_{n \rightarrow \infty} \mathbf{SC}_C^{\sigma}(\theta,n) 
\ = \ \inf_n\ \mathbf{SC}_C^{\sigma}(\theta,n).
\]
For the special cases $\sigma = \failure$ and $\sigma = \, \diverge$ we define:
$$
\mathbf{SC}_C^{\diverge}(\theta) \ = \ 1 \quad \mbox{and} \quad
\mathbf{SC}_C^\failure(\theta) \ =\ 0.
$$

\begin{example}
Let us revisit the program $C$ from Example~\ref{example:closed-form-os} to show how the $\mathbf{SC}$ function works.
The program terminates with score $\frac{1}{2}$ after $16$ steps with the original entropy $\theta$ used in the example.
Thus, $\mathbf{SC}_C^{[]}(\theta, 16) = \frac{1}{2}$. The final configuration reduces to itself by (diverge) infinitely many times,
so we have $\mathbf{SC}_C^{[]}(\theta, n) = \frac{1}{2}$ for  all $n \geq 16$. Thus, 
$\mathbf{SC}_C^{[]}(\theta) = \frac{1}{2} = \lim_{n ->\infty} \mathbf{SC}_C^{[]}(\theta, n) = \frac{1}{2}$.
This is the same result which would be returned by a naive definition of $\mathbf{SC}$, similar to $\mathbf{O}$.
However, if we use the entropy $\theta'$ described above, where all values sampled are $0.5$, this is no longer the case.
In this case, we have $\mathbf{SC}_C^{[]}(\theta, n) = 1$ for all $n$, because the only $\mathtt{score}$ statement
is never reached. Hence,
$\mathbf{SC}_C^{[]}(\theta') = \lim_{n ->\infty} \mathbf{SC}_C^{[]}(\theta', n) = 1$, while the naive definition
would return $0$.
\end{example}
\begin{example}
For a more illustrative example, we consider the trivial program $C'$:
\begin{verbatim}
i := 1;
while(true)
{
  i := i+1;
  score((i^2 - 1) / i^2)
}
\end{verbatim}

Now, for any $\theta$, we have $\mathbf{SC}_C^{[]}(\theta, n) = \frac{1}{2} {\cdot} \frac{i(n)+2}{i(n)+1}$,
where $i(n)$ is the number of loop iterations completed after $n$ steps. As the second factor converges to $1$
as $n$ (and so $i(n)$) goes to infinity, it follows that $\mathbf{SC}_C^{[]}(\theta) = \frac{1}{2}$, so the limit score
is $\frac{1}{2}$ even though the program never terminates.
\end{example}
%where $n(i)$ is the step counter after $i$-th loop iteration. As the second factor converges to $1$ as with It follows immediately 

Note that by the monotone convergence theorem, the limit of approximations always exists.
Thus, $\mathbf{SC}$ is well-defined.
This can be seen as follows.
As scores are bounded by $1$, and no rule other than (score) affects the weight of a program run, a reduction step cannot increase the total score:

\begin{lemma} \label{lemma:w-decreasing}
$\kappa \, \vdash \, \kappa'$ implies $\mbox{\sf weight}(\kappa) \geq \mbox{\sf weight}(\kappa')$.
%%
%%If $\config{\theta}{C}{K}{\sigma}{\theta_K}{m}{w} \vdash^{*} \config{\theta'}{C'}{K'}{\sigma'}{\theta'_K}{m+n'}{w'}$,
%%then $w' \leq w$
\end{lemma}
Thus, scores are antitone:
%%
%%\begin{corollary}
$k \geq n$ implies $\mathbf{SC}_C^{\sigma}(\theta,k) \, \leq \, \mathbf{SC}_C^{\sigma}(\theta,n)$.
%%\end{corollary}
%%
The monotone convergence theorem now yields:
\begin{lemma} \label{lemma:sc-defined}
For each \ccpgcl program $C$, state $\sigma$ and entropy $\theta$, $\lim_{n \rightarrow \infty} \mathbf{SC}_C^{\sigma}(\theta,n)$ exists and is finite.
\end{lemma}

To define the probability distribution of states as a Lebesgue integral involving functions $\mathbf{O}$ and $\mathbf{SC}$, these functions need to be shown to be measurable. 
Although this is a property satisfied by almost all functions used in practice and it is known to be hard to construct non-measurable functions---to construct non-Lebesgue-measurable sets of reals, and hence non-measurable functions on reals, requires the Axiom of Choice \cite{Solovay70}---measurability proofs tend to be lengthy and tedious. 
A detailed proof of measurability of functions similar to $\mathbf{O}_C$ and $\mathbf{SC}_C$, can be found in~\cite{DBLP:conf/icfp/BorgstromLGS16} when providing a semantics to a probabilistic functional programming language.
More details are in Appendix~\ref{app:proofs-meas}; we summarise here the main things:
%%\marginpar{do we really need to provide these details?}
%%
\begin{lemma} \label{lemma:o-sc-measurable}
For all $C$ and $\sigma \in \fullstatespace$: 
$$
(1) \ \mathbf{O}_C^{\sigma}(\cdot) \mbox{ is } \mathcal{S} / \fullstatespace \mbox{ measurable}
\quad \mbox{and} \quad (2) \ 
%% \item $\mathbf{SC}_C^{\sigma}(\cdot,n)$ is $\mathcal{S} / \mathcal{R}$ measurable, for all $n$.
\mathbf{SC}_C^{\sigma}(\cdot) \mbox{ is } \mathcal{S} / \mathcal{R} \mbox{ measurable}.
$$
\end{lemma}
\begin{proof} $\phantom{xx}$
(1) 
Analogous to the proof of Lemma~92 in \cite{DBLP:conf/icfp/BorgstromLGS16}, more details in Appendix~\ref{app:proofs-meas}.
(2)
Analogous to the proof of Lemma~93 in \cite{DBLP:conf/icfp/BorgstromLGS16}, it follows that $\mathbf{SC}_C^{\sigma}(\cdot,n)$ is $\mathcal{S} / \mathcal{R}$ measurable, for all $n$; more details  in Appendix~\ref{app:proofs-meas}.
The result now follows by the fact that point-wise limits of measurable real-valued functions are measurable.
 \qed \end{proof}

%%\subsubsection{Properties of the semantic functions}

\paragraph{Distribution over final program states.}
We are now in a position to define the distribution on final states in terms of the operational semantics.
We first define the distribution on entropies $\langle C \rangle_{\sigma} \colon  \mathcal{S} \rightarrow \mathbb{R}_{+}$, as an integral of score $\mathbf{SC}_C^{\sigma}$ with respect to the standard measure on entropy space:
\[
\langle C \rangle_{\sigma}(B) \ =  \int_{B}  \mathbf{SC}_C^{\sigma}(\theta) \, \mu_{\mathbb{S}} (d\theta).
\]
For each measurable subset $B \in \mathcal{S}$ of the entropy space, $\langle C \rangle_{\sigma}(B)$ is the probability that if we run program $C$ with initial state $\sigma$, the random values sampled during execution will match some element of the set $B$ of entropies.

The probability distribution $|[C|]_{\sigma}\colon \fullstatesa \rightarrow \mathbb{R}_{+}$ on extended states can now be defined as the push-forward measure of $\langle C \rangle_{\sigma}$ with respect to $\mathbf{O}_C^{\sigma}(\theta)$:
\[
\begin{split}
|[C|]_{\sigma}(A) \ = \ 
\langle C \rangle_{\sigma}( \{ \theta\ |\ \mathbf{O}_C^{\sigma}(\theta) \in A \})
\ = \
\langle C \rangle_{\sigma}({\mathbf{O}_C^{\sigma}}^{-1}(A)) \\
\ =\ \int [\mathbf{O}_C^{\sigma}(\theta)(A)] {\cdot} \mathbf{SC}_C^{\sigma}(\theta) \,\mu_{\mathbb{S}}(d\theta).
\end{split}
\]
For program $C$ with $\langle C \rangle_{\sigma}(\mathbb{S}) > 0 $, this distribution can be normalised as follows:
\[
\hat{|[C|]}_{\sigma}(A) \ = \ \frac{|[C|]_{\sigma}(A)}{\langle C \rangle_{\sigma}(\mathbb{S}) }.
\]
\newcommand{\restr}[2]{{#1}\rvert_{#2}}
The measure $|[C|]_{\sigma}(A)$ is a measure on $(\fullstatespace, \fullstatesa)$. 
Let $\restr{|[C|]_{\sigma}}{\statespace}$ be this measure restricted to  $(\statespace, \statesa)$, i.e., the space of proper states without $\failure$ and $\diverge$, such that 
$\restr{|[C|]_{\sigma}}{\statespace}(A) = |[C|]_{\sigma}(A)$ for $A \subseteq \statesa$.

%$\restr{|[C|]_{\sigma}}{\statespace}$ is a measure on $(\statespace, \statesa)$.

\begin{example} \label{example:blr-measure}
We go back once again to the Bayesian linear regression program from Example~\ref{example:blr-os}. We will
%first compute the measure $\langle C \rangle_{\sigma}$ on traces for this program.
first compute the measure $|[C|]_{\sigma}$ on program outcomes. 

%Suppose we want to evaluate the program with entropy $\theta$.
If the value sampled in line 1 is $v_1$ and the value sampled in line 3 is $v_2$, it follows from
the operational semantics that the final state is 
$[\mathtt{u1} \mapsto v_1, \mathtt{a} \mapsto \mathtt{G}(0,2,v_1), 
\mathtt{u2} \mapsto v_2, \mathtt{b} \mapsto \mathtt{G}(0,2,v_2)]$
and the final score is 
$$e^{-(\mathtt{G}(0,2,v_2){-}2)^2} e^{-(\mathtt{G}(0,2,v_1) +\mathtt{G}(0,2,v_2){-}3)^2}
= e^{-(\mathtt{G}(0,2,v_2){-}2)^2 -(\mathtt{G}(0,2,v_1) +\mathtt{G}(0,2,v_2){-}3)^2}.
$$
%for the values of $\mathtt{a}$ and $\mathtt{b}$ in the final state.
Hence, for entropy $\theta$ with 
$\pi_u(\pi_{L,L}(\theta)) = v_1$ and 
$\pi_u(\pi_{L,L,R,R}(\theta)) = v_2$:
%the first two values sampled are $v_1$ and $v_2$, 
\begin{eqnarray*}
\mathbf{O}_C^{[]}(\theta) & = & [\mathtt{u1} \mapsto v_1, \mathtt{a} \mapsto \mathtt{G}(0,2,v_1), 
\mathtt{u2} \mapsto v_2, \mathtt{b} \mapsto \mathtt{G}(0,2,v_2)] \\
\mathbf{SC}_C^{[]}(\theta) & = & e^{-(\mathtt{G}(0,2,v_2)-2)^2 -(\mathtt{G}(0,2,v_1) +\mathtt{G}(0,2,v_2) - 3)^2}.
\end{eqnarray*}
This means that the integral
$\int [\mathbf{O}_C^{[]}(\theta)(A)] {\cdot} \mathbf{SC}_C^{[]}(\theta) \,\mu_{\mathbb{S}}(d\theta)$
can be written as 
$\int f(\pi_u(\pi_{L,L}(\theta)), \pi_u(\pi_{L,L,R,R}(\theta))) \,\mu_{\mathbb{S}}(d\theta)$,
where 
\begin{eqnarray*}
f(v_1,v_2) & = & [[\mathtt{u1}, \mathtt{a}, \mathtt{u2}, \mathtt{b} \mapsto v_1, \mathtt{G}(0,2,v_1), v_2, \mathtt{G}(0,2,v_2)] \in A] \\
&& \phantom{xxx} {\cdot} e^{-(\mathtt{G}(0,2,v_2){-}2)^2 -(\mathtt{G}(0,2,v_1) +\mathtt{G}(0,2,v_2){-}3)^2}.
\end{eqnarray*}
By the definition of entropy, we have:
\begin{eqnarray*}
&& \int f(\pi_u(\pi_{L,L}(\theta)), \pi_u(\pi_{L,L,R,R}(\theta))) \,\mu_{\mathbb{S}}(d\theta)\\
&=& \int \int f(\pi_u(\pi_{L}(\theta_L)), \pi_u(\pi_{L,L,R}(\theta_R))) \,\mu_{\mathbb{S}}(d\theta_L) \mu_{\mathbb{S}}(d\theta_R)\\
&=& \int \int \int \int f(\pi_u(\theta_{L,L})), \pi_u(\pi_{L,L}(\theta_{R,R}))) 
\,\mu_{\mathbb{S}}(d\theta_{L,L}) \mu_{\mathbb{S}}(d\theta_{R,L})
\,\mu_{\mathbb{S}}(d\theta_{L,R}) \mu_{\mathbb{S}}(d\theta_{R,R})\\
&=& \int \int f(\pi_u(\theta_{L,L})), \pi_u(\pi_{L,L}(\theta_{R,R}))) 
\,\mu_{\mathbb{S}}(d\theta_{L,L}) \mu_{\mathbb{S}}(d\theta_{R,R}).
\end{eqnarray*}
By repeatedly applying the definition of entropy like above, we get:
\begin{eqnarray*}
&&\int \int f(\pi_u(\theta_{L,L})), \pi_u(\pi_{L,L}(\theta_{R,R}))) 
\,\mu_{\mathbb{S}}(d\theta_{L,L}) \mu_{\mathbb{S}}(d\theta_{R,R})\\
&& \vdots\\
&=& \int \int f(\pi_u(\theta_{L,L})), \pi_u(\theta_{L,L,R,R})) 
\,\mu_{\mathbb{S}}(d\theta_{L,L}) \mu_{\mathbb{S}}(d\theta_{L,L,R,R})\\
&=& \int_{[0,1]} \int_{[0,1]} f(v_1,v_2) \, \mu_L(dv_1) \mu_L(dv_2).
\end{eqnarray*}
Thus, 
\[
\begin{split}
|[C|]_{\sigma}(A)  =
\int_{(0,1)} \int_{(0,1)}
[[\mathtt{u1} \mapsto v_1, \mathtt{a} \mapsto \mathtt{G}(0,2,v_1), 
\mathtt{u2} \mapsto v_2, \mathtt{b} \mapsto \mathtt{G}(0,2,v_2)] \in A] \\
{\cdot} e^{-(\mathtt{G}(0,2,v_2)-2)^2 -(\mathtt{G}(0,2,v_1) +\mathtt{G}(0,2,v_2) - 3)^2}
 \, \mu_L(dv_1) \mu_L(dv_2).
\end{split}
\]
Now, suppose that $A$ is a set of states such that $\mathtt{a} < 0$ and $\mathtt{b} < 0$. Then:
\[
\begin{split}
|[C|]_{\sigma}(A)  =
\int_{(0,1)} \int_{(0,1)}
[\mathtt{G}(0,2,v_1) < 0] 
[\mathtt{G}(0,2,v_2)]< 0] \\
{\cdot} e^{-(\mathtt{G}(0,2,v_2)-2)^2 -(\mathtt{G}(0,2,v_1) +\mathtt{G}(0,2,v_2) - 3)^2}
 \, \mu_L(dv_1) \mu_L(dv_2).
\end{split}
\]
Like in Example~\ref{example:blr}, this expression can be rewritten as a double integral of Gaussian densities over the real line:
\begin{eqnarray*}
|[C|]_{\sigma}(A)  &=&
\int_{} \int_{}\ 
[x_1 < 0][x_2 < 0] 
{\cdot}e^{-p(\vec{x})}
{\cdot} \mathtt{G}_{pdf}(0,2,x_1) {\cdot} \mathtt{G}_{pdf}(0,2,x_2) 
 \, \mu_L(dx_1) \mu_L(dx_2)\\
&=& \int_{(- \infty, 0)} \int_{(-\infty,0)}
{\cdot} e^{-p(\vec{x})}
{\cdot} \mathtt{G}_{pdf}(0,2,x_1) {\cdot} \mathtt{G}_{pdf}(0,2,x_2) 
 \, \mu_L(dx_1) \mu_L(dx_2)
\end{eqnarray*}
\noindent where $e^{-p(\vec{x})} =  e^{-(x_2-2)^2 -(x_1 +x_2 - 3)^2}$.
Let us now compute the normalising constant $\langle C \rangle_{\sigma}(\mathbb{S})$.
By a similar reasoning as above, we get:
\begin{eqnarray*}
\langle C \rangle_{\sigma}(\mathbb{S}) &=&  \int \mathbf{SC}_C^{\sigma}(\theta) \, \mu_{\mathbb{S}} (d\theta)\\
&=&
\int_{(0,1)} \int_{(0,1)}
e^{-(\mathtt{G}(0,2,v_2)-2)^2 -(\mathtt{G}(0,2,v_1) +\mathtt{G}(0,2,v_2) - 3)^2}
 \, \mu_L(dv_1) \mu_L(dv_2) \\
&=& \int \int
e^{-p(\vec{x})}
{\cdot} \mathtt{G}_{pdf}(0,2,x_1) {\cdot} \mathtt{G}_{pdf}(0,2,x_2) 
 \, \mu_L(dx_1) \mu_L(dx_2)
\end{eqnarray*}
Hence, the normalised semantics $\hat{|[C|]}_{\sigma}(A)$ applied to the above set $A$ is:
\[
\hat{|[C|]}_{\sigma}(A) =
\frac{ \int_{(- \infty, 0)} \int_{(-\infty,0)}
 e^{-p(\vec{x})}
{\cdot} \mathtt{G}_{pdf}(0,2,x_1) {\cdot} \mathtt{G}_{pdf}(0,2,x_2) 
 \, \mu_L(dx_1) \mu_L(dx_2)}
{\int \int
e^{-p(\vec{x})}
{\cdot} \mathtt{G}_{pdf}(0,2,x_1) {\cdot} \mathtt{G}_{pdf}(0,2,x_2) 
 \, \mu_L(dx_1) \mu_L(dx_2)}.
\]
\end{example}

\subsection{Expectations}

The weakest preexpectation semantics determines the expected value of an arbitrary measurable function $f$ on states with respect to a program. 
We can also obtain such expected value by integrating $f$ with respect to the measure $|[C|]_{\sigma}(A)$ defined just above.
%%---that is, computing $\int f(\tau) \restr{|[C|]_{\sigma}}{\statespace}(d\tau)$. 
By change of variable, this integral can be easily transformed into an integral with respect to the default measure on entropies.
\begin{lemma} \label{lemma:unfold-exp-wrt-sem}
For all measurable $f$, 
$$
\int f(\tau) \restr{|[C|]_{\sigma}}{\statespace}(d\tau) 
\ = \
\int \hat{f}(\mathbf{O}_{C}^\sigma(\theta)) \cdot \mathbf{SC}_C^\sigma(\theta)\, \mu_{\mathbb{S}}(d\theta).
$$
%\end{cases}$
\end{lemma}
%\hphantom{}
\begin{proof}
\allowdisplaybreaks
\begin{eqnarray*}
\int f(\tau) \restr{|[C|]_{\sigma}}{\statespace}(d\tau) &=& 
\int \hat{f}(\tau) |[C|]_{\sigma}(d\tau)  \\
%Above, we used Lemma 7.9 from the University of Crete notes
\text{(by property of the pushforward)} &=& 
\int  \hat{f}(\mathbf{O}_C^{\sigma}(\theta)) \,  \langle C \rangle_{\sigma}  (d\theta)  \\
\text{(by Radon-Nikod{\'y}m theorem)} &=&  
\int \hat{f}(\mathbf{O}_C^{\sigma}(\theta)) \mathbf{SC}_C^{\sigma}(\theta) \, \mu_{\mathbb{S}}(d\theta). \\
%\intf(\mathbf{O}_{C}^\sigma(\theta)) 
\end{eqnarray*}
%\noindent as required.
\qed  \end{proof}

%As \cite{DBLP:conf/icfp/BorgstromLGS16} have already presented a detailed
%proof of measurability of similar functions

\begin{example}
Let us compute the expected value of the variable $\mathtt{a}$ in the Bayesian linear
regression example. To this end, we take a function $f$ such that
$f(\sigma) = \sigma(\mathtt{a})$ if $\mathtt{a} \in \mathtt{dom}(\sigma)$
and $f(\sigma)  = 0$ otherwise.
By a similar reasoning as in Example~\ref{example:blr-measure},
we get:
\begin{eqnarray*}
&& \int f(\tau) \restr{|[C|]_{[]}}{\statespace}(d\tau)  \\
\ &=& \
\int \hat{f}(\mathbf{O}_{C}^{[]}(\theta)) \cdot \mathbf{SC}_C^{[]}(\theta)\, \mu_{\mathbb{S}}(d\theta)\\
&=&\ 
\int_{(0,1)} \int_{(0,1)}
\hat{f}([\mathtt{u1} \mapsto v_1, \mathtt{a} \mapsto \mathtt{G}(0,2,v_1), 
\mathtt{u2} \mapsto v_2, \mathtt{b} \mapsto \mathtt{G}(0,2,v_2)]) \\
&& \qquad {\cdot} e^{-(\mathtt{G}(0,2,v_2)-2)^2 -(\mathtt{G}(0,2,v_1) +\mathtt{G}(0,2,v_2) - 3)^2}
 \, \mu_L(dv_1) \mu_L(dv_2)\\
&=&\ 
\int_{(0,1)} \int_{(0,1)}  \mathtt{G}(0,2,v_1) 
{\cdot} e^{-(\mathtt{G}(0,2,v_2)-2)^2 -(\mathtt{G}(0,2,v_1) +\mathtt{G}(0,2,v_2) - 3)^2}
 \, \mu_L(dv_1) \mu_L(dv_2)
\end{eqnarray*}
This is the same result as the one we obtain using $\mathtt{wp}$ in Example~\ref{example:blr}.

\begin{example}
Let us now revisit the program from Example~\ref{example:diverge-wp} and calculate the expected value of
the constant function $f(\sigma) = 1$ with respect to the program using the operational semantics. To this end,
we need to calculate 
$\int \hat{f}(\mathbf{O}_{C}^{[]}(\theta)) \cdot \mathbf{SC}_C^{[]}(\theta)\, \mu_{\mathbb{S}}(d\theta)$
for the given program $C$. 
By evaluating the first two statements in the program, like in Example~\ref{example:blr-os}, we can check 
that 
$\mathbf{O}_{C}^{[]}(\theta)
= \mathbf{O}_{C'}^{\sigma}(\pi_{R,R}(\theta))$
and $\mathbf{SC}_C^{[]}(\theta) = \mathbf{SC}_{C'}^{\sigma}( \pi_{R,R}(\theta))$,
where $C' = \mathtt{while}(\mathtt{b}=0)\{C''\}$  ($C''$ being the loop body)
and $\sigma = [\mathtt{b} \mapsto 0 , \mathtt{k} \mapsto 0]$.
It follows from the properties of entropy that 
$\int \hat{f}(\mathbf{O}_{C'}^{\sigma}(\pi_{R,R}(\theta))) \cdot \mathbf{SC}_{C'}^{\sigma}( \pi_{R,R}(\theta))\, \mu_{\mathbb{S}}(d\theta)
= \int \hat{f}(\mathbf{O}_{C'}^{\sigma}(\theta)) \cdot \mathbf{SC}_{C'}^{\sigma}(\theta)\, \mu_{\mathbb{S}}(d\theta)$.

Now, let $C'_n = \mathtt{while}^n(\mathtt{b}=0)\{C''\}$. %By Lemma~\ref{lemma:unfold-exp-wrt-sem},
We can show (using Proposition~\ref{lemma:sup-while-o-sc} from Appendix~\ref{app:proofs-op-sem} and the Beppo Levi's theorem) that 
$$ \int \hat{f}(\mathbf{O}_{C'}^{\sigma}(\theta)) \cdot \mathbf{SC}_{C'}^{\sigma}(\theta)\, \mu_{\mathbb{S}}(d\theta)
=  \sup_n \int \hat{f}(\mathbf{O}_{C'_n}^{\sigma}(\theta)) \cdot \mathbf{SC}_{C'_n}^{\sigma}(\theta)\, \mu_{\mathbb{S}}(d\theta).$$
Since $\hat{f}$ has value $1$ on all proper states and is $0$ on state $\diverge$, $\hat{f}(\mathbf{O}_{C'_n}^{\sigma}(\theta)) = 1$ if 
$\mathbf{O}_{C'}^{\sigma}(\theta)$ is a proper state (that is, if $C'$ terminates with initial state $\sigma$ and entropy $\theta$) and
$\hat{f}(\mathbf{O}_{C'_n}^{\sigma}(\theta)) = 0$ if $C'$ does not terminate with $\theta$.
Thus, $\hat{f}(\mathbf{O}_{C'_n}^{\sigma}(\theta)) = [\theta \in S_1] + [\theta \in S_2] + \dots + [\theta \in S_{n-1}]$, where
$S_i$ is the set of entropies resulting in termination after  exactly $i$ iterations\footnote{
The reason the last set is $S_{n-1}$
and not $S_{n}$ is that $\mathtt{while}^1(\phi)\{C''\} = C'';\mathtt{diverge}$
if $\phi$ is true, so 
$\mathtt{while}^n(\theta)\{C''\}$ only terminates if the loop body is executed at most $n-1$ times.}.

%In every terminating run, 
The score is only multiplied by $\frac{\mathtt{k}}{\mathtt{k}+1}$ in the last iteration, after which the guard is satisfied.
As long as the guard of the while-loop is false, the score stays at $1$. Thus, we have
$\mathbf{SC}_C^{[]}(\theta)  = [\theta \in S_1] {\cdot}\frac{1}{2}+ [\theta \in S_2]  {\cdot}\frac{2}{3}+ 
\dots + [\theta \in S_{n}] {\cdot}\frac{n}{n+1}  + [\theta \notin S_1 \cup \dots \cup S_{n}]$\footnote{
This time, the last set is  $S_{n}$, because the $\mathtt{score}$ statement will be executed
even if the loop body is followed by $\mathtt{diverge}$.}.

Therefore, for each $n$, 
\begin{eqnarray*}
\int \hat{f}(\mathbf{O}_{C'_n}^{\sigma}(\theta)) \cdot \mathbf{SC}_{C'_n}^{\sigma}(\theta)\, \mu_{\mathbb{S}}(d\theta)
&=&\int \sum_{k=1}^{n-1} \frac{k}{k+1} [\theta \in S_k]\,  \mu_{\mathbb{S}}(d\theta) \\
&=& \sum_{k=1}^{n-1} \frac{k}{k+1}  \int [\theta \in S_k]\,  \mu_{\mathbb{S}}(d\theta) \\
\end{eqnarray*}

Now we need to calculate $ \int [\theta \in S_k]\,  \mu_{\mathbb{S}}(d\theta)$ for each $k$.
Observe that whether $\theta \in S_k$, depends only on parts of $\theta$ which are sampled from
(that is, on sub-entropies to which $\pi_U$ is applied). 
The value of $[\theta \in S_k]$ depends
only on the sub-entropies $\pi_{p_1}(\theta)$, \dots, $\pi_{p_k}(\theta)$, where $p_1, \dots, p_k$ are the paths
leading to values sampled in subsequent iterations. 
An entropy $\theta$ leads to termination in the $k$-th step if $\pi_U(\pi_{p_1}(\theta)) \geq \frac{1}{4}$, \dots,
$\pi_U(\pi_{p_{k-1}}(\theta)) \geq \frac{1}{k^2}$ and
$\pi_U(\pi_{p_{k}}(\theta)) \leq \frac{1}{(k+1)^2}$. 
Thus, by the definition of entropy, we have
\allowdisplaybreaks
\begin{eqnarray*}
\int [\theta \in S_k]\,  \mu_{\mathbb{S}}(d\theta) &=&
\int \left[\pi_U(\pi_{p_1}(\theta)) \geq \frac{1}{4}\right] {\cdot} \dots {\cdot}
\left[\pi_U(\pi_{p_{k-1}}(\theta)) \geq \frac{1}{k^2}\right] \\
&& \qquad {\cdot} \left[\pi_U(\pi_{p_{k}}(\theta)) \leq \frac{1}{(k{+}1)^2}\right]
\, \mu_{\mathbb{S}}(\theta)
\\
&=&\int \dots \int \left[\pi_U(\pi_{p_1}(\theta)) \geq \frac{1}{4}\right] {\cdot} \dots {\cdot}
\left[\pi_U(\pi_{p_{k-1}}(\theta)) \geq \frac{1}{k^2}\right] \\
&& \qquad {\cdot} \left[\pi_U(\pi_{p_{k}}(\theta)) \leq \frac{1}{(k{+}1)^2}\right]
\, \mu_{\mathbb{S}}(\theta_{p_1}) \dots \mu_{\mathbb{S}}(\theta_{p_k})\\
&=&\int \dots \int \left[v_1 \geq \frac{1}{4}\right] {\cdot} \dots {\cdot}
\left[v_{k-1} \geq \frac{1}{k^2}\right] \\
&& \qquad {\cdot} \left[v_k \leq \frac{1}{(k{+}1)^2}\right]
\, \mu_L(d v_1) \dots \mu_L(d v_k) \\
&=& \left (\prod_{i=1}^{k-1} \frac{(i{+}1)^2{-}1}{(i{+}1)^2} \right ) {\cdot} \frac{1}{(k{+}1)^2}\\
&=& \frac{1}{2} {\cdot} \frac{1}{k \cdot (k{+}1)}.
\end{eqnarray*}
Hence,
$ \sum_{k=1}^{n-1} \frac{k}{k+1}  \int [\theta \in S_k]\,  \mu_{\mathbb{S}}(d\theta) 
= \frac{1}{2} {\cdot} \sum_{k=1}^{n-1} \frac{1}{(k+1)^2}$, so we have
\begin{eqnarray*}
\int \hat{f}(\mathbf{O}_{C}^{[]}(\theta)) \cdot \mathbf{SC}_C^{[]}(\theta)\, \mu_{\mathbb{S}}(d\theta)
&=&  \int \hat{f}(\mathbf{O}_{C'}^{\sigma}(\theta)) \cdot \mathbf{SC}_{C'}^{\sigma}(\theta)\, \mu_{\mathbb{S}}(d\theta)\\
&=& \sup_n \int \hat{f}(\mathbf{O}_{C'_n}^{\sigma}(\theta)) \cdot \mathbf{SC}_{C'_n}^{\sigma}(\theta)\, \mu_{\mathbb{S}}(d\theta)\\
&=&   \frac{1}{2}   \sum_{k=1}^{\infty } \frac{1}{k \cdot (k+1)}\\
&=& \frac{\pi^2}{12} - \frac{1}{2}
\end{eqnarray*}
\end{example}
This is exactly the result we obtained with the weakest preexpectation semantics in Example~\ref{example:diverge-wp}. The correspondence between the weakest preexpectation and operational semantics is the topic of the next section.

\end{example}

\section{Equivalence of $\mathtt{wp}$ and operational semantics}

The aim of this section is to show that the weakest preexpectation semantics of \ccpgcl is equivalent to its operational semantics. 
This property is formalised by two theorems which relate the wp and wlp semantics to the operational semantics. 
The first result asserts that the expected value of an arbitrary function $f$ defined by the weakest preexpectation operator equals the expected value of $f$ computed as an integral of $f$ with respect to the distribution induced by the operational semantics.
\begin{theorem} \label{thm:wp-op-equiv}
For all measurable functions $f \colon \statespace -> \extposreals$, \ccpgcl programs $C$ and initial states $\sigma \in \statespace$:
$$
\mathtt{wp} |[ C |](f)(\sigma) \ = \ \int f(\tau) |[C|]_{\sigma}(d\tau). 
$$
\end{theorem}

\begin{proof}
By Lemma~\ref{lemma:unfold-exp-wrt-sem}, it suffices to prove that for all $f$:
$$
\int \hat{f}(\mathbf{O}_{C}^\sigma(\theta)) \cdot \mathbf{SC}_C^\sigma(\theta)\, \mu_{\mathbb{S}}(d\theta) \ = \mathtt{wp}  |[ C |](f)(\sigma).
$$
This can be proven by induction on the structure of $C$. The detailed proof can be found in Appendix~\ref{section:thms-1-2}.
%What is interesting?
%
The proof makes use of several compositionality properties of the operational semantics and properties of finite
approximations of $\mathtt{while}$-loops, which are also proven in the appendix. A key insight used in the proof
is that Beppo Levi's theorem can be used to express the expectation of $f$ with respect to a $\mathtt{while}$-loop
as the limit of expectations of $f$ with respect to finite approximations of the loop.
 \qed \end{proof}

The second main theorem of this paper states that the weakest liberal preexpectation of a non-negative function $f$ bounded by $1$ is equivalent to the expected value of $f$ with respect to the distribution defined by the operational semantics plus the probability of divergence weighted by scores.

\begin{theorem} \label{thm:wlp-op-equiv}
For every measurable non-negative function $f \colon \statespace -> \extposreals$ with $f(\sigma) \leq 1$ for all states $\sigma$, \ccpgcl program $C$ and initial state $\sigma \in \statespace$:
$$
\mathtt{wlp} |[ C |](f)(\sigma) \ = \ 
\int f(\tau) \cdot \restr{|[C|]_{\sigma}}{\statespace}(d\tau) + 
\underbrace{\int [\mathbf{O}_C^{\sigma}(\theta) =\ \diverge] \cdot \mathbf{SC}_C^{\sigma}(\theta)  \, \mu_{\mathbb{S}} (d\theta)}_{\mbox{
%\parbox{3.5cm}
\begin{tabular}{c}
\footnotesize probability of divergence \\  \footnotesize multiplied by the score
\end{tabular}
}}.
$$
\end{theorem}
\begin{proof}
By induction on the structure of $C$. Details in Appendix~\ref{section:thms-1-2}.
 \qed \end{proof}

\begin{corollary}
%For all $\Gamma$ and $\sigma \models \Gamma$,  
For every \ccpgcl program $C$ and state $\sigma$:
$$
\mathtt{wlp} |[ C |](1)(\sigma) \ = \ 
\int \mathbf{SC}_C^{\sigma}(\theta)  \, \mu_{\mathbb{S}} (d\theta).
$$
\end{corollary}

\section{Related work}

\paragraph{Semantics of languages for Bayesian inference.}
Research on the semantics of probabilistic programs dates back to the pioneering work by Saheb-Djahromi \cite{DBLP:conf/mfcs/Saheb-Djahromi78} and Kozen \cite{DBLP:journals/jcss/Kozen81}, among others. However, this early work is mostly motivated by applications such as the analysis of randomised algorithms, so the languages involved mostly only supported discrete distributions and did not allow conditioning.

A recent explosion of popularity of machine learning, and the rise of probabilistic programming as a tool for Bayesian inference, have sparked a new line of work on semantics of languages with continuous random draws and conditioning. 
An early example of such work is the paper by Park et al. \cite{park08sampling},
% \cite{DBLP:conf/popl/ParkPT05} (later extended to a journal article \cite{park08sampling}), 
 who present an operational semantics for a higher-order language with conditioning, parametrised by an infinite trace of random values.
Borgstr{\"o}m et al. \cite{DBLP:journals/corr/BorgstromGGMG13} define a denotational semantics of a first-order language with both discrete and continuous distributions, which also supports conditioning, including zero-probability observations. 
Nori et al. \cite{R2} define a denotational semantics of an imperative language with (hard) conditioning, similar to the weakest preexpectation semantics; they however do not consider possible program divergence.
Toronto et al. \cite{DBLP:conf/esop/TorontoMH15} present a denotational semantics for a first-order functional language which interprets programs as deterministic functions on the source of randomness. 
Huang and Morrisett \cite{DBLP:conf/esop/HuangM16} define a semantics for a first-order language, restricted to computable operations. 
Heunen et at. \cite{DBLP:conf/lics/HeunenKSY17} present a denotational semantics of a higher-order functional language with continuous random draws and conditioning. 
They manage to overcome the well-known problem with measurability of higher-order function application \cite{aumann61} by replacing standard Borel spaces with so-called quasi-Borel spaces. 
This idea, simplifying the authors' previous work \cite{Staton16}, has since gained a lot of attraction in the community: 
{\'S}cibior et al. \cite{DBLP:journals/pacmpl/ScibiorKG18} use quasi-Borel spaces to prove correctness of 
sampling-based inference algorithms, while V{\'a}k{\'a}r et al. \cite{DBLP:journals/pacmpl/VakarKS19} define a domain 
theory for higher-order functional probabilistic programs, which extends the quasi-Borel space approach to programs with 
higher-order recursion and recursive types. 
A different approach is followed in a recent paper by Dahlqvist and Kozen \cite{DBLP:journals/pacmpl/DahlqvistK20}, who define a semantics of a probabilistic language with conditioning in terms of Banach spaces.

The operational semantics presented in this chapter is strongly inspired by the semantics of Borgstr{\"o}m et al. \cite{DBLP:conf/icfp/BorgstromLGS16} and Wand et al. \cite{DBLP:journals/pacmpl/WandCGC18}, both defined for functional programs.
The former define a measure on program outcomes by integrating functions similar to our $\mathbf{O}$ and $\mathbf{SC}$, defined in terms of an operational semantics, with respect to a stock measure on traces of random values. 
The latter use a similar approach, but define their operational semantics in terms of infinite entropies instead of finite traces, and use continuations to fix evaluation order and split entropies between continuations consistently.

\paragraph{Program divergence.} 
Another line of research on probabilistic programs, coming mostly from the algorithms and program verification community and inspired by earlier papers by Kozen \cite{DBLP:journals/jcss/Kozen81} and McIver et al. \cite{Morgan96}, has focused on extending Dijkstra's weakest precondition calculus to probabilistic programs.
In this line of work, correct handling of diverging programs has been a key issue from the start. 
Recent developments
\cite{DBLP:journals/pe/GretzKM14, DBLP:journals/jacm/KaminskiKMO18, DBLP:conf/popl/ChatterjeeNZ17} focus on problems such as analysing runtimes, almost-sure termination (and variants thereof) and outcomes of algorithms.
A weakest preexpectation semantics for recursive imperative probabilistic programs is given by Olmedo et al. \cite{DBLP:conf/lics/OlmedoKKM16}.
Olmedo et al. \cite{Olmedo18} also extend the weakest preexpectation calculus to programs with hard conditioning and possible divergence, but their semantics only supports discrete distributions.

\paragraph{Combining continuous distributions, conditioning and divergence.} 
The issue of program divergence has so far mostly been disregarded when defining semantics of Bayesian probabilistic programs, with most authors assuming that their semantics is only applicable to almost-surely terminating programs. 
Conversely, semantics designed to handle diverging programs usually did not support conditioning, and when they did, they were not applicable to programs with continuous distributions.

To our best knowledge, the only existing semantics supporting the combination of divergence, continuous random draws, and conditioning is the recent work by Bichsel et al. \cite{BichselGV18}. 
The authors define a semantics of an imperative probabilistic language with continuous and discrete distributions and hard conditioning, in which the probability of failing a hard constraint, the probability of an execution error and the probability of divergence are defined explicitly. 
The semantics calculates probability measures on final program states and the above exceptions are treated as special states, like $\failure$ and $\diverge$ in our semantics.

Technically, the semantics in \cite{BichselGV18} is a superset of our semantics. 
A normalised expectation of the form
$\frac{\mathtt{wp}|[C|](f)(\sigma_0) }{\mathtt{wlp}|[C|](1)(\sigma_0) }$ can be defined in their semantics as
\[
\frac{\int_{\Omega_{\sigma}} f(\tau) |[C|](\sigma_0)(d\tau)}
{|[C|](\sigma_0)(\Omega_{\sigma}) + |[ C |](\sigma_0)(\diverge)}
\]
\noindent where $|[C|](\sigma_0)$ is the measure on final states of program $C$ with initial state $\sigma_0$, as defined by the semantics, and $\Omega_{\sigma}$ is the set of proper states (excluding errors and divergence).
However, we believe that extending the well-studied framework of weakest preexpectations to the continuous case is still a significant contribution, as it allows using established techniques, not applicable to the semantics in \cite{BichselGV18}, to analyse programs with continuous distributions, conditioning, in the presence of possible program divergence.

\section{Epilogue}

In this paper, we have considered a probabilistic while-language that contains three important ingredients: (a) sampling from continuous probability distributions, (b) soft and hard conditioning, and (c) program divergence. 
We have provided a weakest (liberal) preexpectation semantics for our language and showed that soft conditioning can be encoded by hard conditioning. 
The wp-semantics is complemented by an operational semantics using the concept of entropies.
The main results of this paper are the correspondence theorem between the wp-semantics (and wlp-semantics) and the operational semantics.
The paper has been written in a tutorial-like manner with various illustrative examples.

Let us conclude with a short discussion.
The interplay between divergence and conditioning is intricate.
For the discrete probabilistic setting, this has been extensively treated in~\cite{Olmedo18}.
Intuitively speaking, the problem is how conditioning is taken into account by program runs that diverge and never reach the \texttt{score} statement.
Consider the program:
\begin{verbatim}
t := 1; 
x := U;
if (x > 0.5) {  
   while(true) { t := t+1; } 
}
score(softeq(t, 1));
return t;
\end{verbatim}
This program terminates with probability $\nicefrac{1}{2}$ with $t=1$, and with the same probability diverges increasing $t$ ad infinitum.
One would perhaps expect the expected value of $t$ to be 1, as the possibility of $t$ going to infinity should be discarded by the \texttt{score} statement. 
However, for any function $f$, we have $\mathtt{wp}(\mathtt{while}(\mathtt{true})\{t=t{+}1\})(f) = 0$ and $\mathtt{wlp}(\mathtt{while}(\mathtt{true})\{t=t{+}1\})(f) = 1$.
Hence, the expected value of $t$ (for the empty initial state) will be:
\[
\frac{\mathtt{wp}|[C|](\lambda \sigma\ .t)([])}{\mathtt{wlp}|[C|](\lambda \sigma\ .t)([])}
\ = \ \frac{\nicefrac{1}{2}}{1} \ = \ \frac{1}{2}.
\]

\appendix

\section{Basics of measure theory}
\label{app:basics-measure-theory}

This section presents the basic definitions of measure theory used throughout this
of the paper. 
For a more thorough introduction to measure theory, please consult one of the standard textbooks such as \cite{billingsley95}.

%\subsubsection{Measurable, topological and metric spaces}
\subsubsection{Measurable spaces}

\begin{definition}
A $\sigma$-algebra $\Sigma$ on a set $\Omega$ is a set consisting of subsets of $\Omega$ which satisfies the
following properties:
\begin{itemize}
\item $\emptyset \in \Sigma$
\item If $A \in \Sigma$, then $\Omega \setminus A \in \Sigma$ (closure under complements)
\item If $A_i \in \Sigma$ for all $i \in \mathbb{N}$, then $\bigcup_{i \in \mathbb{N}} A_i \in \Sigma$ (closure
under countable unions)
\end{itemize}
The tuple $(\Omega, \Sigma)$ of a set $\Omega$ and its $\sigma$-algebra $\Sigma$ is called
a \emph{measurable space}. A set $A \in \Sigma$ is called a \emph{measurable set}.
\end{definition}

\begin{definition}
A $\sigma$-algebra on a set $\Omega$ \emph{generated} by a set $S$ of subsets of $\Omega$
is the smallest $\sigma$-algebra containing $S$.
\end{definition}

\begin{definition}
A \emph{countably generated} $\sigma$-algebra on $\Omega$ is a $\sigma$-algebra generated by a
countable set of subsets of $\Omega$
\end{definition}

\begin{definition}
If $(\Omega_1, \Sigma_1)$ and $(\Omega_2, \Sigma_2)$ are measurable spaces, the
$\emph{product}$ of the $\sigma$-algebras $\Sigma_1$ and $\Sigma_2$ is the $\sigma$-algebra $\Sigma_1 \otimes \Sigma_2$ on $\Omega_1 \times \Omega_2$ defined as $\Sigma_1 \otimes \Sigma_2 = \sigma(\{(A_1 \times A_2\ |\ A_1 \in \Sigma_1, A_2 \in \Sigma_2 \})$.
This definition extends naturally to arbitrary finite products of measures.
\end{definition}

\begin{definition}
A \emph{Borel $\sigma$-algebra} $\mathcal{R}$ on $\mathbb{R}$ is the $\sigma$-algebra generated by the
set of open intervals $(a, \infty)$ for $a \in \mathbb{R}$. A Borel $\sigma$-algebra $\mathcal{R}_n$ on  
$\mathbb{R}^n$ is the $n$-fold product of $\mathcal{R}$.
\end{definition}

\subsubsection{Measures}

\begin{definition}
A $\emph{measure}$ on the measurable space $(\Omega, \Sigma)$ is a function
$\mu : \Sigma \mapsto \extposreals$ such that $\mu(\emptyset) = 0$ and 
for any collection of pairwise disjoint sets $A_1, A_2, \dots$,
$\mu(\bigcup_{i \in \mathbb{N}} A_i) 
= \sum_{i \in \mathbb{N}} \mu(A_i)$ (i.e. $\mu$ is countably additive). 
\end{definition}

\begin{definition}
A \emph{product} $\mu_1 \otimes \mu_2$ of measures $\mu_1$ and $\mu_2$ on $(\Omega_1, \Sigma_1)$ and $(\Omega_2, \Sigma_2)$,
respectively, is the unique measure on $(\Omega_1 \times \Omega_2, \Sigma_1 \times \Sigma_2)$ which satisfies
$(\mu_1 \otimes \mu_2)(A_1 \times A_2) = \mu_1(A_1) \mu_2(A_2)$ for all $A_1 \in \Sigma_1$, $A_2 \in \Sigma_2$.
This definition extends naturally to finite products of higher dimensions.
\end{definition}

\begin{definition}
The \emph{Lebesgue measure} on $(\mathbb{R}, \mathcal{R})$ is the unique  measure $\mu_L$ which
satisfies $\mu_L([a,b]) = b - a$ for all $a, b \in \mathbb{R}$ such that $b \geq a$. The Lebesgue measure
on $(\mathbb{R}^n, \mathcal{R}_n)$ is the $n$-fold product of $\mu_L$.
\end{definition}

\begin{definition}
A \emph{probability measure} on  $(\Omega, \Sigma)$ is a measure $\mu$ such that $\mu(\Omega) = 1$.
A \emph{subprobability measure} on  $(\Omega, \Sigma)$ is a measure $\mu$ with $\mu(\Omega) \leq 1$.
\end{definition}

\begin{definition}
A measure $\mu$ on $(\Omega, \Sigma)$ is \emph{$\sigma$-finite} if there exists a sequence of sets
$A_i \in \Sigma$ such that $A_i \subseteq A_{i+1}$ for all $i$ and $\mu(A_i) < \infty$ and
$\Omega = \bigcup_{i \in \mathbb{N}} A_i$.
\end{definition}

%TODO: do we need to define measure restriction?

\subsubsection{Measurable functions and integrals}

\begin{definition}
A function $f$ between measurable spaces $(\Omega_1, \Sigma_1)$ and $(\Omega_2, \Sigma_2)$ is 
\emph{measurable $\Sigma_1 / \Sigma_2$} if for all $B \in \Sigma_2$, $f^{-1}(B) \in \Sigma_1$. 
If the $\sigma$-algebras $\Sigma_1$ and $\Sigma_2$ are clear from the context, we will
simply call $f$ measurable.
\end{definition}

\begin{definition}
For a measurable space $(\Omega, \Sigma)$, a \emph{simple function} $g \colon \Omega -> \mathbb{R}_{+}$ is 
a measurable $\Sigma / \mathcal{R}$ function with a finite image set, which can be expressed as 
$g(x) = \Sigma_{i=1}^{n} \alpha_i [x \in A_i]$, where $A_i = f^{-1}(\alpha_1)$.
The \emph{Lebesgue integral} of a simple function $g(x) =  \Sigma_{i=1}^{n} \alpha_i [x \in A_i]$ with
respect to a measure $\mu$ on $(\Omega, \Sigma)$ is defined as:
\[
\int g(x)\, \mu(dx) = \sum_{i=1}^n \alpha_i \mu(A_i)
\]
The Lebesgue integral of any measurable function $f$ is then defined as the limit of integrals of simple functions
pointwise smaller than $f$:
\[
\int f(x)\, \mu(dx) = \sup \left \{\int g(x)\, \mu(dx) \ |\ g\ \text{simple},  g \leq f \right \}
\]
\end{definition}

\begin{theorem}[Beppo Levi]
Let $f_i \colon X -> \extposreals$ be a (pointwise) non-decreasing sequence of positive measurable functions
and let $f = \lim_{n -> \infty} \int f_i$ be the pointwise limit of the sequence.Then $f$ is measurable
and
\[
\int f\, d\mu = \lim_{n -> \infty} f_n\, d\mu
\]

The same holds for non-increasing sequences, provided that  $\int f_0\, d\mu < \infty$.
\end{theorem}

Note that the limit and supremum of a non-decreasing sequence coincide. limit 
and infimum of a non-increasing sequence also coincide.

\subsubsection{Metric and topological spaces}

\begin{definition}
A \emph{metric} on a set $\Omega$ is a function $d \colon \Omega \times \Omega -> \extposreals$ such that
$d(x,x) = 0$  and $d(x,y) + d(y,z) \geq d(x,z)$ for all $x, y, z \in \Omega$. The pair $(\Omega, d)$ is called
a \emph{metric space}.
\end{definition}

\begin{definition}
If $(\Omega, d)$ is a metric space, $A \subseteq \Omega$ is \emph{open} if every element $x \in A$ has
a neighbourhood which is completely enclosed in $A$, i.e. there exists $\epsilon > 0$ such that
$\{y \in \Omega\ |\ d(x,y) < \epsilon \} \subseteq A$.
\end{definition}

\begin{definition}
If $(\Omega_1, d_1)$ and $(\Omega_2, d_2)$ are metric spaces, then a \emph{product} of
$(\Omega_1, d_1)$ and $(\Omega_2, d_2)$ is the metric space
$(\Omega_1 \times \Omega_2, d_{12})$, where $d_{12}$ is the \emph{Manhattan product} of metrics
$d_1$ and $d_2$, defined as 
$$d_{12}((x_1, y_1), (x_2, y_2)) = d_1(x_1, y_1) + d_2(x_2, y_2).$$
This definition naturally extends to finite products of higher dimensions.
\end{definition}

A product of topological spaces can also be defined using the standard Euclidean product metric
 $d_{12}((x_1, y_1), (x_2, y_2)) =\sqrt{ d_1(x_1, y_1)^2 + d_2(x_2, y_2)^2}$, both metrics induce 
the same topologies. We use Manhattan products as they are easier to work with.

\begin{definition}
A \emph{topology} on a set $\Omega$ is a set $\mathcal{O}$ of subsets of $\Omega$ such that
\begin{itemize}
\item $\emptyset \in \mathcal{O}$
\item $\Omega \in \mathcal{O}$
\item For all $O_1, \dots, O_n \in \mathcal{O}$, $O_1 \cap O_2 \cap \dots \cap O_n \in \mathcal{O}$
\item If $O_i \in \mathcal{O}$ for all $i \in \mathbb{N}$, then $\bigcup_{n \in \mathbb{N}} O_i \in \mathcal{O}$.
\end{itemize}
The pair $(\Omega, \mathcal{O})$ is called a \emph{topological space} and
the elements of the topology $\mathcal{O}$ are called \emph{open sets}.
\end{definition}

\begin{definition}
If $(\Omega_1, \mathcal{O}_1)$ and $(\Omega_2, \mathcal{O}_2)$ are topological spaces, then a \emph{product} of
$(\Omega_1, d_1)$ and $(\Omega_2, d_2)$ is the metric space
$(\Omega_1 \times \Omega_2, \mathcal{O}_1 \times \mathcal{O}_2)$, where 
 the \emph{product} of topologies
$\mathcal{O}_1 \times \mathcal{O}_2$ is the smallest topology on $\Omega_1 \times \Omega_2$
which makes both left and right projections continuous. This definition naturally extends to final products of 
higher dimensions.
\end{definition}

\begin{definition}
A function $f$ between metric spaces $(\Omega_1,d_1)$ and $(\Omega_2,d_2)$ is \emph{continuous}
if for every $x \in \Omega_1$ and $\epsilon > 0$, there exists $\delta$ such that for all $y \in \Omega_1$,
if $d_1(x,y) < \epsilon$, then $d_2(f(x), f(y)) < \delta$.
\end{definition}

\begin{definition}
A function $f$ between topological spaces $(\Omega_1,\mathcal{O}_1)$ and $(\Omega_2, \mathcal{O}_2)$ is \emph{continuous}
if for every open set $O \in \mathcal{O}_2$, $f^{-1}(O) \in \mathcal{O}_1$.
\end{definition}

\subsubsection{From metric to measurable spaces}

\begin{definition}
A topology on $\Omega$ \emph{induced} by a metric $d$ is the smallest topology
which contains all open sets of the metric space $(\Omega, d)$.
\end{definition}

\begin{definition}
The \emph{Borel $\sigma$-algebra } $\mathcal{B}(\Omega, \mathcal{O})$ is the $\sigma$-algebra generated
by a topology  $\mathcal{O}$ on $\Omega$.
\end{definition}

\begin{definition}
We call the Borel $\sigma$-algebra on $\Omega$ generated by the topology induced by the metric $d$
the \emph{$\sigma$-algebra  induced by $d$}. We denote such a $\sigma$ algebra by $\mathcal{B}(\Omega, d)$.
\end{definition}

The following lemmas are well-established results:

\begin{lemma}
If $\mathcal{O}_1$ and $\mathcal{O}_2$ are, respectively, topologies on $\Omega_1$ and $\Omega_2$ induced
by metrics $r_1$ and $r_2$, and a function $f$ between the metric spaces $(\Omega_1,d_1)$ and $(\Omega_2,d_2)$
is continuous, then $f$ is also continuous as a function between topological spaces 
$(\Omega_1, \mathcal{O}_1)$ and $(\Omega_2, \mathcal{O}_2)$.
\end{lemma}

\begin{lemma}
If $f$ is a continuous function between topological spaces $(\Omega_1, \mathcal{O}_1)$ and $(\Omega_2, \mathcal{O}_2)$
and $\Sigma_1$ and $\Sigma_2$ are the Borel $\sigma$-algebras on, respectively, $\Omega_1$ and $\Omega_2$ generated
by topologies $\mathcal{O}_1$ and $\mathcal{O}_2$, then the function $f$ is measurable.
\end{lemma}

\begin{corollary} \label{corr:cont-fun-meas}
If $(\Omega_1, d_1)$ and $(\Omega_2, d_2)$ are metric spaces and $f$ is a continuous function
from $\Omega_1$ to $\Omega_2$, then $f$ is measurable $\mathcal{B}(\Omega_1, d_1) / \mathcal{B}(\Omega_2, d_2)$
\end{corollary}

\begin{lemma}
If $(\Omega_1, d_1)$ and $(\Omega_2, d_2)$ are separable metric spaces, then for the Manhattan product $d_{12}$ of
metrics $d_1$ and $d_2$
\[
\mathcal{B}(\Omega_1 \times \Omega_2, d_{12}) = \mathcal{B}(\Omega_1, d_1) \times \mathcal{B}(\Omega_2, d_2)
\]
\end{lemma}

\begin{corollary} \label{corr:cont-fun-prod-meas}
\sloppy If $(\Omega_1, d_1)$,  $(\Omega_2, d_2)$, $(\Omega_3, d_3)$ and $(\Omega_4, d_4)$ are separable 
metric spaces and $f$ is a continuous function from $\Omega_1 \times \Omega_2$ to 
$\Omega_3 \times \Omega_4$ (with respect to corresponding product metrics) then $f$
is measurable $\mathcal{B}(\Omega_1, d_1) \times \mathcal{B}(\Omega_2, d_2) / 
\mathcal{B}(\Omega_3, d_3) \times \mathcal{B}(\Omega_4, d_4)$.
\end{corollary}

All the above results extend naturally to arbitrary finite products.

\section{Basics of domain theory}
\label{appendix:domain-theory}

This section includes some basic definitions from domain theory which are required to understand the paper.
For readers wanting a more complete, tutorial-style introduction, there are many resources
available, including \citep{HuttonDomThy} and \citep{AbramskyJungDT}.
%A more complete introduction to domain theory can be found in many standard 

Please note that we use the notions of $\omega$-complete partial order and $\omega$-continuity, 
defined in terms of countable sequences of increasing values ($\omega$-chains), rather than the more general
notions of complete partial order (requiring existence of suprema of directed sets) and continuity 
(requiring the given function to preserve suprema of all subsets of the domain). While $\omega$-completeness
and $\omega$-continuity are technically weaker than completeness and continuity, respectively,
they are sufficient for our purposes, as they allow applying the Kleene Fixpoint Theorem.

\begin{definition} [Partially-ordered set]
A \emph{partially-ordered set} is a pair $(D, \sqsubseteq)$ of set $D$ and relation $\sqsubseteq$ such that:
\begin{itemize}
\item For each $a \in D$, $a \sqsubseteq a$ (reflexiveness)
\item For each $a, b, c \in D$, if $a \sqsubseteq b$ and $b \sqsubseteq c$, then $a \sqsubseteq c$ (transitivity)
\item For each $a, b \in D$, if $a \sqsubseteq b$ and $b \sqsubseteq a$, then $a = b$ (antisymmetry)
\end{itemize}
\end{definition}

\begin{definition}[$\omega$-chain and its supremum]
A $\omega$-chain in a partially-ordered set $(D, \sqsubseteq)$ is an infinite
sequence $d_0, d_1, d_2, \dots$ such that for all $i$, $d_i \in D$ and $d_i \sqsubseteq d_{i+1}$.
The \emph{supremum} $\sup_i d_i$ of a chain $d_0, d_1, d_2, \dots$ is the supremum
of the set $\{ d_0, d_1, d_2, \dots \}$ of elements of the chain.
%is the set of all elements
%of an infinite non-decreasing sequence of elements of $D$.
%infinite sequence $d_0$, $d_1$, $d_2$, \dots, such that for each $i$, $d_i \in D$ and
%$d_i \sqsubseteq d_{i+1}$.
\end{definition}

\begin{definition}[$\omega$-complete partial order]
A \emph{$\omega$-complete partial order} ($\omega$-cpo) is a partial order $(D, \sqsubseteq)$
such that for each $\omega$-chain $d_0, d_1, d_2, \dots$ in $(D, \sqsubseteq)$,
the supremum $\sup_i d_i$ exists in $D$.
%the supremum of $C$ exists in $D$.
\end{definition}

\begin{definition}[Monotone function]
A function $f \colon D -> D'$ between $\omega$-cpos $(D, \sqsubseteq)$ and $(D', \sqsubseteq')$
is \emph{monotone} if $f(d) \sqsubseteq' f(d')$ for each $d, d' \in D$ such that $d \sqsubseteq d'$.
\end{definition}

\begin{definition}[$\omega$-continuous function]
A function $f \colon D -> D'$ between $\omega$-cpos $(D, \sqsubseteq)$ and $(D', \sqsubseteq')$
is \emph{$\omega$-continuous} if it is monotone and for each  $\omega$-chain $d_0, d_1, d_2, \dots$ in $(D, \sqsubseteq)$,
$f(\sup_i d_i) = \sup_i f(d_i)$.
\end{definition}

\noindent Note that in the definition above, the requirement that $f$ is monotone ensures 
that $f(d_0)$, $f(d_1)$, $f(d_2)$, \dots is a $\omega$-chain.

\begin{definition}[Least fixpoint]
Let $(D, \sqsubseteq)$ be a $\omega$-cpo and $f \colon D-> D$ a function on $(D, \sqsubseteq)$. 
A \emph{fixpoint} of $f$ is an element $d \in D$ such that $f(d) = d$. A \emph{least fixpoint}
of $f$ is a fixpoint $d_0$ of $f$ such that for all other fixpoints $d$ of $f$, $d_0 \sqsubseteq d$.
\end{definition}

\begin{theorem}[Kleene Fixpoint Theorem]
Let $(D, \sqsubseteq)$ be a $\omega$-cpo and $f \colon D -> D$ a $\omega$-continuous function.
Then $f$ has a least fixpoint, which is the supremum of the
chain $\bot$, $f(\bot)$, $f(f(\bot))$, \dots, that is, 
$\sup_i f^i(\bot)$.
\end{theorem}

\section{Proofs for the $\mathtt{wp}$ and $\mathtt{wlp}$ semantics}
\label{app:proofs-wp-wlp}

In order to prove that $\mathtt{wp}|[C|](f)$ is measurable for all $f$, we first need to prove
that the state update $\lambda (x, \sigma, E) . \sigma[x \mapsto \sigma(E)]$ is measurable.
%We prove that state update is measurable---this key property is needed in the measurability proofs.
Since states are a new structure, not discussed in the proofs of measurability in \citep{SzymczakPhD},
we present the proof in more detail than other measurability proofs in this paper.

We define a metric $d_{\mathcal{N}}$ on variables as $d_{\mathcal{N}}(x,x)=0$ and $d_{\mathcal{N}}(x,y) = \infty$
for $x \neq y$. The metric space $(\mathcal{N}, d_{\mathcal{N}})$ induces the usual discrete $\sigma$-algebra on $\mathcal{N}$.

%\begin{lemma} \label{lemma:subst-measurable}
%For all $x$, the update function $h_x : \statespace \times (\mathbb{R} \uplus \mathbb{Z}) -> \statespace$ defined
%by $h_x(\sigma, v) = \sigma[x \mapsto v]$, is measurable.
%\end{lemma}
%\begin{proof}
%We prove that this function is continuous, which implies measurability. Take $\sigma_1$, $\sigma_2$ and
%$V_1, V_2 \in \mathbb{R} \uplus \mathbb{Z}$. If $\mathtt{dom}(\sigma_1) \neq \mathtt{dom}(\sigma_2)$
%then $d_\sigma(\sigma_1, \sigma_2) = \infty$, so trivially $d_\sigma(h_x(\sigma_1,V_1), h_x(\sigma_2,V_2))
%\leq d_\sigma(\sigma_1, \sigma_2) + d_T(V_1, V_2)$. The inequality also immediately holds if
%$V_1 \in \mathbb{R}$ and $V_2 \in \mathbb{Z}$ (or vice versa), because then $d_T(V_1, V_2) = \infty$.
%
%Now, suppose that  $\mathtt{dom}(\sigma_1) = \mathtt{dom}(\sigma_2)
%= \{x_1, \dots, x_n \}$ and either
%$V_1, V_2 \in \mathbb{R}$ or $V_1, V_2 \in \mathbb{Z}$.
%Now, if $x = x_k$  for some $k$, then 
%\begin{eqnarray*}
%d_\sigma(h_x(\sigma_1,V_1), h_x(\sigma_2,V_2))
%&=& \Sigma_{i \in 1..n, i \neq k } d_T(\sigma_1(x_i), \sigma_2(x_i)) + d_T(V_1, V_2) \\
%&\leq& \Sigma_{i \in 1..n} d_T(\sigma_1(x_i), \sigma_2(x_i)) + d_T(V_1, V_2) \\
%&=& d_\sigma(\sigma_1, \sigma_2) + d_T(V_1, V_2)
%\end{eqnarray*}
%If $x \neq x_k$  for any $k$, we simply have:
%\begin{eqnarray*}
%d_\sigma(h_x(\sigma_1,V_1), h_x(\sigma_2,V_2))
%&=& \Sigma_{i \in 1..n} d_T(\sigma_1(x_i), \sigma_2(x_i)) + d_T(V_1, V_2) \\
%&=& d_\sigma(\sigma_1, \sigma_2) + d_T(V_1, V_2)
%\end{eqnarray*}
%Thus, $h_x$ is continuous, and so measurable.
%\qed \end{proof}
%
%UPDATED VERSION:

\begin{lemma} \label{lemma:subst-measurable}
The update function $h \colon \mathcal{N} \times \statespace \times (\mathbb{R} \uplus \mathbb{Z}) -> \statespace$ defined
by $h(x, \sigma, v) = \sigma[x \mapsto v]$, is measurable.
\end{lemma}
\begin{proof}
%TODO: metric on variables?
We prove that this function is continuous, which implies measurability. Take $x_1, x_2 \in \mathcal{N}$, 
$\sigma_1, \sigma_2 \in \statespace$ and
$V_1, V_2 \in \mathbb{R} \uplus \mathbb{Z}$. If $\mathtt{dom}(\sigma_1) \neq \mathtt{dom}(\sigma_2)$
then $d_\sigma(\sigma_1, \sigma_2) = \infty$, so trivially $d_\sigma(h(x_1, \sigma_1,V_1), h(x_2, \sigma_2,V_2))
\leq d_{\mathcal{N}}(x_1, x_2) + d_\sigma(\sigma_1, \sigma_2) + d_T(V_1, V_2) = \infty$. 
The same holds when $x_1 \neq x_2$ (which implies $d_{\mathcal{N}}(x_1, x_2) = \infty$).
The inequality also immediately holds if
$V_1 \in \mathbb{R}$ and $V_2 \in \mathbb{Z}$ (or vice versa), because then $d_T(V_1, V_2) = \infty$.

Now, suppose that  $x_1 = x_2 = x$, $\mathtt{dom}(\sigma_1) = \mathtt{dom}(\sigma_2)
= \{y_1, \dots, y_n \}$ and either
$V_1, V_2 \in \mathbb{R}$ or $V_1, V_2 \in \mathbb{Z}$.
Now, if $x = y_k$  for some $k$, then 
\begin{eqnarray*}
d_\sigma(h(x, \sigma_1,V_1), h(x, \sigma_2,V_2))
&=& \sum_{i \in 1..n, i \neq k } d_T(\sigma_1(y_i), \sigma_2(y_i)) + d_T(V_1, V_2) \\
&\leq& \sum_{i \in 1..n} d_T(\sigma_1(y_i), \sigma_2(y_i)) + d_T(V_1, V_2) \\
&=& d_\sigma(\sigma_1, \sigma_2) + d_T(V_1, V_2) + d_{\mathcal{N}}(x, x)
\end{eqnarray*}
If $x \neq x_k$  for any $k$, we simply have:
\begin{eqnarray*}
d_\sigma(h(x, \sigma_1,V_1), h(x, \sigma_2,V_2))
&=& \sum_{i \in 1..n} d_T(\sigma_1(y_i), \sigma_2(y_i)) + d_T(V_1, V_2) \\
&=& d_\sigma(\sigma_1, \sigma_2) + d_T(V_1, V_2) + d_{\mathcal{N}}(x, x)
\end{eqnarray*}
Thus, $h_x$ is continuous, and so measurable.
\qed \end{proof}
%\begin{corollary}
%For all $x$ and $V \in \mathbb{R}$, the update function $h_{x,V} : \statespace -> \statespace$ defined
%by $h_{x,V}(\sigma) = \sigma[x \mapsto V]$, is measurable.
%\end{corollary}

\begin{restate}{Lemma~\ref{lemma:wp-continuous-measurable}}
For every program $C$, the function $\mathtt{wp}|[C|](\cdot)$ is $\omega$-continuous. Moreover,
for every measurable $f \colon \statespace -> \extposreals$, $\mathtt{wp}|[C|](f)(\cdot)$ is measurable.
\end{restate}
\begin{proof}[of Lemma~\ref{lemma:wp-continuous-measurable}]
By induction on the structure of $C$. The continuity part of the proof is largely similar to the proof of 
the analogous property in \cite{DBLP:journals/pe/GretzKM14},
with additional care needed because of the use of Lebesgue integration.
We need to show that for any $C$ and any $\omega$-chain $f_1 \leq f_2 \leq f_3 \dots$,
$\mathtt{wp}|[C|](\sup_i f_i) = \sup_i\ \mathtt{wp}|[C|](f_i)$ and that
$\mathtt{wp}|[C|](f)$ is measurable for any measurable $f$.

\begin{itemize}
\item Case $C =  x :\approx U $:

\begin{itemize}
\item \textbf{Continuity:}
\begin{eqnarray*}
\mathtt{wp}|[C|](\sup_i f_i) &=& \lambda \sigma . \int_{[0,1]} (\sup_i f_i)(\sigma[x \mapsto v]) \,\mu_L(dv)\\
{\tiny \text{(by Beppo Levi's theorem)}}&=& \lambda \sigma .\ \sup_i \int_{[0,1]} f_i(\sigma[x \mapsto v]) \,\mu_L(dv)\\
{\tiny \text{(sup taken wrt pointwise ordering)}}&=& \sup_i\ \lambda \sigma . \int_{[0,1]} f_i(\sigma[x \mapsto v]) \,\mu_L(dv)\\
&=& \sup_i\ \mathtt{wp}|[C|](f_i)
\end{eqnarray*}

\item \textbf{Measurability:}

We have
\[
\mathtt{wp}|[ C |](f) = \lambda \sigma . \int_{[0,1]}  g(x, \sigma, v)  \, \mu_L(dv)
\]
\noindent where $g(x, \sigma, v) = f(\sigma[x \mapsto v ])$. Now, take 
$h(x, \sigma,v) = \sigma[x \mapsto v]$. Then $g = f \circ h$. We know that substitutions are measurable
(Lemnma~\ref{lemma:subst-measurable}), so $h$ is measurable. This means that $g$ is measurable, as it is a composition of measurable
functions. Thus, by the Fubini-Tonelli theorem, $\lambda \sigma . \int_{[0,1]}  g(x, \sigma, v)  \, \mu_L(dv)$
is measurable, so $\mathtt{wp}|[ C |](f)$ is measurable.

\end{itemize}

\item Case $C = \mathtt{score}(E)$: 
\begin{itemize}
\item \textbf{Continuity:}
\begin{eqnarray*}
\mathtt{wp}|[C|](\sup_i f_i) &=& \lambda \sigma .\ [\sigma(E) \in (0,1]] \sigma(E) \cdot (\sup_i f_i)(\sigma)\\
{\tiny \text{(multiplying by a constant preserves sup)}}&=& \lambda \sigma .  \sup_i  (  [\sigma(E) \in (0,1]] \sigma(E) \cdot f_i(\sigma))\\
{\tiny \text{(sup taken wrt pointwise ordering)}}&=&  \sup_i \ \lambda \sigma .\  [\sigma(E) \in (0,1]]\sigma(E) \cdot f_i(\sigma)\\
&=& \sup_i\ \mathtt{wp}|[C|](f_i)
\end{eqnarray*}

\item \textbf{Measurability:}

We have $\mathtt{wp}|[ C |](f) = \lambda \sigma .\ [\sigma(E) \in (0,1]]\sigma(E) \cdot f(\sigma)$.
The substitution $\sigma(E)$ is measurable by assumption (as a function of $\sigma$).
Meanwhile, $[\sigma(E) \in (0,1]]$ is a composition of the measurable function $\sigma(E)$
and the indicator function of the measurable set $(0,1]$, which is obviously measurable.
% and
Finally, $f$ is measurable by assumption, so the pointwise product of these three functions is measurable.
\end{itemize}

\item Case $C = \mathtt{observe}(\phi)$:
\begin{itemize}
\item \textbf{Continuity:}
\begin{eqnarray*}
\mathtt{wp}|[C|](\sup_i f_i) &=& \lambda \sigma . [\sigma(\phi)] ( \sup_i\ f_i)(\sigma) \\
{\tiny \text{(multiplying by a constant preserves sup)}}&=& \lambda \sigma .  \sup_i  ([\sigma(\phi)]  f_i(\sigma))\\
{\tiny \text{(sup taken wrt pointwise ordering)}}&=&  \sup_i \ \lambda \sigma .\ [\sigma(\phi)] f_i(\sigma)\\
&=& \sup_i\ \mathtt{wp}|[C|](f_i)
\end{eqnarray*}

\item \textbf{Measurability:}

We have $\mathtt{wp}|[ C |](f) = \lambda \sigma . [\sigma(\phi)] f(\sigma) $.
The function $\sigma . [\sigma(\phi)]$ is measurable by assumption (we only allow measurable
predicates in the language), and $f$ is measurable by assumption of the lemma, hence their pointwise
product is measurable.
\end{itemize}
%
%\end{itemize}
%
% The remaining cases are the same as in the proof of Lemma 34 in \cite{DBLP:journals/pe/GretzKM14}
%(modulo the fact that we define continuity in terms of chains rather than directed sets). However, they are shown below
%(in more detail) for completeness.
%
%\begin{itemize}
\item Case $C = (x := E )$:
\begin{itemize}
\item \textbf{Continuity:}
\begin{eqnarray*}
\mathtt{wp}|[C|](\sup_i f_i) &=& \lambda \sigma . (\sup_i f_i)(\sigma[x \mapsto \sigma(E)]) \\
{\tiny \text{(sup taken wrt pointwise ordering)}}&=& \lambda \sigma . \sup_i f_i (\sigma[x \mapsto \sigma(E)]) \\
{\tiny \text{(sup taken wrt pointwise ordering)}}&=&  \sup_i \ \lambda \sigma .f_i (\sigma[x \mapsto \sigma(E)])\\
&=& \sup_i\ \mathtt{wp}|[C|](f_i)
\end{eqnarray*}

\item \textbf{Measurability:}

We have $\mathtt{wp}|[ C |](f) =\lambda \sigma . f(\sigma[x \mapsto \sigma(E)])$.
This can be represented as a composition of functions $\lambda \sigma . f \circ  F_2 \circ F_1 (\sigma)$,
where %$F_1(\sigma) = (\sigma, \sigma)$, 
$F_1(\sigma) = (\sigma, \sigma(E))$ and $F_2(\sigma, V) = \sigma[x \mapsto V]$.
The function $F_1$ is measurable, because the identity function $\lambda \sigma . \sigma$ is trivially measurable,
and $\lambda \sigma . \sigma(E)$ is measurable by assumption, so both components of $F_1$ are measurable.
The function $F_2$ is measurable by Lemma~\ref{lemma:subst-measurable}. Hence, $\mathtt{wp}|[ C |](f)$ is measurable as a composition
of measurable functions.

%The resulting function is a composition of measurable functions, which can be shown in the usual way.
%The only nontrivial part  is the measurability of $\lambda \sigma . \sigma(E)$ for any $E$. TODO: prove it, or at least sketch a proof.

\end{itemize}

\item Case $C = \mathtt{while}(\phi)\{C'\}$:
\begin{itemize}
\item \textbf{Continuity:}
We have: 
\begin{eqnarray*}
\mathtt{wp}|[C|](\sup_i f_i) &=& \mathtt{wp}|[\mathtt{while}(\phi)\{C'\}|](\sup_i f_i)\\
&=& \mathtt{lfp}\ X . [\neg \phi](\sup_i f_i) + [\phi] \mathtt{wp}|[C'|](X) \\
\end{eqnarray*}

Take $\Phi_{f}(X) =  [\neg \phi]f + [\phi] \mathtt{wp}|[C'|](X)$. By induction hypothesis,
$\mathtt{wp}|[C'|](\cdot)$ is continuous, so $\Phi_{f}(\cdot)$ is continuous for all $f \colon \statespace -> \extposreals$.
Moreover, it can be easily checked that for any $X$, $f \mapsto \Phi_{f}(X)$ is continuous as a function of $f$
(which means that $f \mapsto \Phi_{f}$ is continuous).
Thus,
\[
\mathtt{wp}|[C|](\sup_i f_i) = \sup_n \Phi_{\sup_i f_i}^n (0)
= \sup_n (\sup_i\ \Phi_{f_i})^n (0)
\]

By Theorem 2.1.19.2 from \cite{AbramskyJungDT}, the function $\Phi \mapsto \sup_n\ \Phi^n(0)$ is continuous.
If $f_1, f_2, \dots$ is an increasing chain, then $\Phi_{f_1}, \Phi_{f_2}, \dots$ is also an increasing chain
(because $\Phi_f$ is monotone in $f$). Thus,
 $\sup_n (\sup_i\ \Phi_{f_i})^n(0) = \sup_i (\sup_n\ \Phi_{f_i}^n(0))
=\sup_i\ \mathtt{wp}|[C|](f_i)$, as required.

\item \textbf{Measurability:}

%$\mathtt{wp}|[ C |](f) =\mathtt{lfp}\ X . [\neg \phi](f) + [\phi] \mathtt{wp}|[C'|](X)$
The function $\Phi_f(X) =  [\neg \phi](f) + [\phi] \mathtt{wp}|[C'|](X)$ is continuous for all
measurable $f$ by the induction hypothesis, so by the fixpoint theorem $\mathtt{lfp}\ X .\Phi_f(X)$
exists in the domain of measurable functions.
\end{itemize}

\item Case $C = C_1;C_2$:
\begin{itemize}
\item \textbf{Continuity:}

We have :
\[
\mathtt{wp}|[C|](\sup_i f_i) = \mathtt{wp}|[ C_1|](\mathtt{wp}|[C_2 |](\sup_i f_i))
\]
%$\mathtt{wp}|[ C |](f) = \mathtt{wp}|[ C_1|](\mathtt{wp}|[C_2 |](f))$:

By induction hypothesis, $\mathtt{wp}|[C_2 |](\sup_i f_i) = \sup_i \mathtt{wp}|[C_2 |](f_i) $.
The induction hypothesis also states that $\mathtt{wp}|[C_2 |](f_i) $ is measurable for all measurable $f_i$, 
which also means that $\sup_i \mathtt{wp}|[C_2 |](f_i) $ is measurable. Hence,
$\mathtt{wp}|[ C_1|](\sup_i \mathtt{wp}|[C_2 |](f_i))$ is well-defined. By applying the induction 
hypothesis again, we get $\mathtt{wp}|[ C_1|](\sup_i \mathtt{wp}|[C_2 |](f_i))
= \sup_i  \mathtt{wp}|[ C_1|](\mathtt{wp}|[C_2 |](f_i))$, as required.

\item \textbf{Measurability:}

By induction hypothesis, $\mathtt{wp}|[C_2 |](f)$ is measurable,
and so $\mathtt{wp}|[ C_1|](\mathtt{wp}|[C_2 |](f))$ is also measurable
by induction hypothesis.

\end{itemize}

\item The other cases are straightforward.

\end{itemize}
 \qed \end{proof}

\section{Proofs for the operational semantics}
\label{app:proofs-op-sem}

%TODO: restructure

\subsection{Properties of the operational semantics}

This section consists of proofs of properties of the operational semantics which are needed 
to prove Proposition~\ref{lemma:o-sc-seq}. 

\subsubsection{Basic properites} We begin by stating two basic properties: that reduction 
is deterministic and that the weight always stays positive.
%and  that the reduction relation (and its closure) is preserved by  changing the initial step 
%count and multiplying  the weights by a positive number.

\begin{lemma}[Evaluation is deterministic] \label{lemma:eval-det}
For any configuration $\kappa$, if $\kappa \vdash \kappa'$ and $\kappa \vdash \kappa''$, then $\kappa' = \kappa''$.
%%
%%
%%$\config{\theta}{C}{K}{\sigma}{\theta_K}{n}{w}$,
%%if 
%%$\config{\theta}{C}{K}{\sigma}{\theta_K}{n}{w}
%%\vdash^{*}
%%\config{\theta'}{\modownarrow}{K'}{\sigma'}{\theta'_K}{n'}{w'}$
%%and
%%$\config{\theta}{C}{K}{\sigma}{\theta_K}{n}{w}
%%\vdash^{*}
%%\config{\theta''}{\modownarrow}{K''}{\sigma''}{\theta''_K}{n''}{w''}$,
%%then $\config{\theta'}{\modownarrow}{K'}{\sigma'}{\theta'_K}{n'}{w'} = 
%%\config{\theta''}{\modownarrow}{K''}{\sigma''}{\theta''_K}{n''}{w''}$.
\end{lemma}
%%\begin{proof}
%%%%By induction on the derivation of $\config{\theta}{C}{K}{\sigma}{\theta_K}{n}{w}
%%%%\vdash^{*}
%%%%\config{\theta'}{\modownarrow}{K'}{\sigma'}{\theta'_K}{n'}{w'}$.
%%By inspection of the derivation rules.
%%\qed \end{proof}

\begin{lemma} \label{lemma:w-greater-0}
If $\kappa \vdash \kappa'$ and $\mbox{\sf weight}(\kappa) > 0$, then $\mbox{\sf weight}(\kappa') > 0$.
\end{lemma}
%%\begin{proof}
%%Straightforward induction on $n'$.
%%\qed \end{proof}

\subsubsection{Invariance of reduction relation} 
%When reasoning about reduction, we often
%need to construct a single reduction chain from two separate multi-step reductions. This requires
%that 
%
The functions $\mathbf{O}_C^{\sigma}$ and $\mathbf{SC}_C^{\sigma}$ are defined in terms of
reduction chains which start at configurations with $K=[]$, $n=0$ and $w=1$. However, in order to
reason about evaluation of compositions of terms, we need to deal with reduction sequences starting at 
intermediate configurations, where this property does not hold. The following lemmas show that the 
reduction relation is preserved by modifying the initial and final step count, weight and continuation. 
%We can easily observe that this holds for single-step reduction.

Proving invariance of the semantics under step count and weight change is straightforward:
\begin{lemma} \label{lemma:change-n-w}
If $\config{\theta}{C}{K}{\sigma}{\theta_K}{n}{w} \vdash^{*}
\config{\theta'}{C'}{K'}{\sigma'}{\theta'_K}{n+n'}{w'}$,
then for all $w'' > 0$ and integer $n'' \geq -n$,
$\config{\theta}{C}{K}{\sigma}{\theta_K}{n + n''}{w'' w} \vdash^{*}
\config{\theta'}{C'}{K'}{\sigma'}{\theta'_K}{n + n'' +n'}{w'' w'}$.
\end{lemma}
\begin{proof}
Simple induction on $n'$.
\qed \end{proof}

The rest of this section shows that the semantics is also preserved by extending the initial continuation.
In the following lemmas, we write $K @ K'$ for the concatenation of two continuations
$K$ and $K'$ (recall that a continuation is a list of expressions).

\begin{lemma} \label{lemma:add-k-step}
\begin{itemize}
\item If $\config{\theta}{C}{K'}{\sigma}{\theta_K}{n}{w}
\vdash
\config{\theta'}{C'}{K''}{\sigma'}{\theta_K'}{n + 1}{w'}$
and $\sigma' \neq \failure$
and $(C,K') \neq (\modownarrow, [])$,
then
$\config{\theta}{C}{K'@K}{\sigma}{\theta_K}{n}{w}
\vdash
\config{\theta'}{C'}{K''@K}{\sigma'}{\theta_K'}{n + 1}{w'}$.
\item If $\config{\theta}{C}{K'}{\sigma}{\theta_K}{n}{w}
\vdash
\config{\theta'}{C'}{K''}{\failure}{\theta_K'}{n + 1}{w'}$
then
$\config{\theta}{C}{K'@K}{\sigma}{\theta_K}{n}{w}
\vdash
\config{\theta'}{C'}{[]}{\failure}{\theta_K'}{n + 1}{w'}$.
\end{itemize}
\end{lemma}
\begin{proof}
By inspection of the reduction rules.
\qed \end{proof}

\begin{lemma} \label{lemma:min-full-red}
If $\config{\theta}{C}{K}{\sigma}{\theta_K}{n}{w}
\vdash^{*}
\config{\theta'}{\modownarrow}{[]}{\sigma'}{\theta'_K}{n + n'}{w'}$, then
there exists a unique $\hat{n} \leq n'$ such that
$\config{\theta}{C}{K}{\sigma}{\theta_K}{n}{w}
\vdash^{*}_{\mathtt{min}}
\config{\theta'}{\modownarrow}{[]}{\sigma'}{\theta'_K}{n + \hat{n}}{w'}$
\end{lemma}
\begin{proof}
Obvious.
\qed \end{proof}

\begin{lemma} \label{lemma:add-k}
If $\config{\theta}{C}{K}{\sigma}{\theta_K}{n}{w} \vdash^{*}
\config{\theta'}{C'}{K'}{\sigma'}{\theta'_K}{n+n'}{w'}$
and $(C', K') \neq (\modownarrow, [])$
and $\sigma' \neq \failure$,
then for all $K''$,
$\config{\theta}{C}{K @ K''}{\sigma}{\theta_K}{n}{w} \vdash^{*}
\config{\theta'}{C'}{K' @ K''}{\sigma'}{\theta'_K}{n+n'}{w'}$.
\end{lemma}
\begin{proof}
By induction on $n'$:
\begin{itemize}
\item Base case: $n' = 0$: trivial

\item Induction step: Let $n' > 0$. Then we have
$\config{\theta}{C}{K}{\sigma}{\theta_K}{n}{w} \vdash
\config{\hat{\theta}}{\hat{C}}{\hat{K}}{\hat{\sigma}}{\hat{\theta_K}}{n+1}{\hat{w}} \vdash^{*}
\config{\theta'}{C'}{K'}{\sigma'}{\theta'_K}{n+n'}{w'}$. We now need to split on the derivation of 
$\config{\theta}{C}{K}{\sigma}{\theta_K}{n}{w} \vdash
\config{\hat{\theta}}{\hat{C}}{\hat{K}}{\hat{\sigma}}{\hat{\theta_K}}{n+1}{w}$.

\begin{itemize}
\item
If $\config{\theta}{C}{K}{\sigma}{\theta_K}{n}{w} \vdash
\config{\hat{\theta}}{\hat{C}}{\hat{K}}{\hat{\sigma}}{\hat{\theta_K}}{n+1}{\hat{w}}$
was derived with (seq), then  $C = C_1; C_2$, $\hat{K}  = C_2 \mathrel{::} K$ and we have
$\config{\theta}{C_1; C_2}{K}{\sigma}{\theta_K}{n}{w} \vdash
\config{\pi_L(\theta)}{C_1}{C_2 \mathrel{::} K}{\sigma}{\pi_L(\theta) \mathrel{::} \theta_K}{n+1}{w} \vdash^{*}
\config{\theta'}{C'}{K'}{\sigma'}{\theta'_K}{n+n'}{w'}$.

By (seq), 
$\config{\theta}{C_1; C_2}{K@K''}{\sigma}{\theta_K}{n}{w} \vdash
\config{\pi_L(\theta)}{C_1}{C_2 \mathrel{::} K@K''}{\sigma}{\pi_L(\theta) \mathrel{::} \theta_K}{n+1}{w}$,
and by the induction hypothesis, 
$\config{\pi_L(\theta)}{C_1}{C_2 \mathrel{::} K @ K''}{\sigma}{\pi_L(\theta) \mathrel{::} \theta_K}{n+1}{\hat{w}} \vdash^{*}
\config{\theta'}{C'}{K' @ K''}{\sigma'}{\theta'_K}{n+n'}{w'}$.

\item
If $\config{\theta}{C}{K}{\sigma}{\theta_K}{n}{w} \vdash
\config{\hat{\theta}}{\hat{C}}{\hat{K}}{\hat{\sigma}}{\hat{\theta_K}}{n+1}{\hat{w}}$
was derived with (pop), then $C = \modownarrow$ and $K = C' \mathrel{::} K'''$ and we have
$\config{\theta}{\modownarrow}{C' \mathrel{::} K'''}{\sigma}{\theta_K}{n}{w} \vdash
\config{\pi_L(\theta_K)}{C'}{K'''}{\sigma}{\pi_R(\theta_K) }{n+1}{w} \vdash^{*}
\config{\theta'}{C'}{K'}{\sigma'}{\theta'_K}{n+n'}{w'}$.

By (pop),
$\config{\theta}{\modownarrow}{C' \mathrel{::} K''' @ K''}{\sigma}{\theta_K}{n}{w} \vdash
\config{\pi_L(\theta_K)}{C'}{K''' @ K''}{\sigma}{\pi_R(\theta_K) }{n+1}{w}$,
and by induction hypothesis, 
$\config{\pi_L(\theta_K)}{C'}{K''' @ K''}{\sigma}{\pi_R(\theta_K) }{n+1}{w} \vdash^{*}
\config{\theta'}{C'}{K' @ K''}{\sigma'}{\theta'_K}{n+n'}{w'}$.

\item
Otherwise, we have $\hat{K} = K$ and by inspection of the reduction rules,
$\config{\theta}{C}{K@K''}{\sigma}{\theta_K}{n}{w} \vdash
\config{\hat{\theta}}{\hat{C}}{K@K''}{\hat{\sigma}}{\hat{\theta_K}}{n+1}{\hat{w}}$, so the result
follows immediately by applying the induction hypothesis 
(note that $(C', K') \neq (\modownarrow, [])$ implies that 
$\config{\theta}{C}{K}{\sigma}{\theta_K}{n}{w} \vdash
\config{\hat{\theta}}{\hat{C}}{\hat{K}}{\hat{\sigma}}{\hat{\theta_K}}{n+1}{\hat{w}}$
is not derived with (final)).
\end{itemize}
\end{itemize}
\qed \end{proof}

\begin{corollary} \label{corr:change-n-w-k}
If $\config{\theta}{C}{K}{\sigma}{\theta_K}{n}{w} \vdash^{*}
\config{\theta'}{C'}{K'}{\sigma'}{\theta'_K}{n+n'}{w'}$
and $\sigma' \neq \failure$
and $(C', K') \neq (\modownarrow, [])$,
then for all $w'' > 0$, integer $n'' \geq -n$ and $K''$,
$\config{\theta}{C}{K @ K''}{\sigma}{\theta_K}{n + n''}{w'' w} \vdash^{*}
\config{\theta'}{C'}{K' @ K''}{\sigma'}{\theta'_K}{n + n'' +n'}{w'' w'}$.
\end{corollary}

The reason we added the condition $(C', K') \neq (\modownarrow, [])$ to the premise of
Lemma~\ref{lemma:add-k} is that in our semantics, a ``final'' configuration with 
statement $\modownarrow$ and empty continuation reduces to itself (by the (final) rule) infinitely. If we 
replaced $[]$ with some non-empty continuation $K$, the rule (pop) would be applied instead of (final) and
the reduction would be completely different. The statement 
$\config{\theta}{C}{K}{\sigma}{\theta_K}{n}{w} \vdash^{*}
\config{\theta'}{\modownarrow}{[]}{\sigma'}{\theta'_K}{n+n'}{w'}$ says nothing about how
many times the rule (final) was applied at the end, so we do not know what the final configuration
after $n'$ steps would be if we appended some continuation $K'$ to $K$.

Because of that, we need to treat the case $(C', K') = (\modownarrow, [])$ separately. We first introduce
some new notation: we write
$\config{\theta}{C}{K}{\sigma}{\theta_K}{n}{w}
\vdash^{*}_{\mathtt{min}}
\config{\theta'}{\modownarrow}{[]}{\sigma'}{\theta'_K}{n + n'}{w'}$
if $\config{\theta}{C}{K}{\sigma}{\theta_K}{n}{w}
\vdash^{*}
\config{\theta'}{\modownarrow}{[]}{\sigma'}{\theta'_K}{n + n'}{w'}$
and there is no $n'' < n'$ such that
$\config{\theta}{C}{K}{\sigma}{\theta_K}{n}{w}
\vdash^{*}
\config{\theta''}{\modownarrow}{[]}{\sigma''}{\theta''_K}{n + n''}{w''}$
(or, equivalently,
$\config{\theta}{C}{K}{\sigma}{\theta_K}{n}{w}
\vdash^{*}
\config{\theta'}{\modownarrow}{[]}{\sigma'}{\theta'_K}{n + n'}{w'}$
was derived without (final)).

%technical
\begin{lemma}[Evaluation with continuation] \label{lemma:add-k-final}
If $\config{\theta}{C}{[]}{\sigma}{\theta_K}{n}{w}
\vdash^{*}_{\mathtt{min}}
\config{\theta'}{\modownarrow}{[]}{\sigma'}{\theta'_K}{n + n'}{w'}$
and $\sigma' \neq \failure$, 
then
$\config{\theta}{C}{K}{\sigma}{\theta_K}{n}{w}
\vdash^{*}
\config{\theta'}{\modownarrow}{K}{\sigma'}{\theta'_K}{n + n'}{w'}$.
\end{lemma}
\begin{proof}
We will prove a more general statement:

If $\config{\theta}{C}{K'}{\sigma}{\theta_K}{n}{w}
\vdash^{*}_{\mathtt{min}}
\config{\theta'}{\modownarrow}{[]}{\sigma'}{\theta_K'}{n + n'}{w'}$, 
then
$\config{\theta}{C}{K'@K}{\sigma}{\theta_K}{n}{w}
\vdash^{*}
\config{\theta'}{\modownarrow}{K}{\sigma'}{\theta_K'}{n + n'}{w'}$, 

by induction on $n'$:
\begin{itemize}
\item Base case: $n' = 0$: This implies that $C = \modownarrow$ and $w'  = w$ and $K' = []$ and $\theta_K' = \theta_K$, so the result follows trivially.
\item Induction step: for $n' > 0$, we have
$\config{\theta}{C}{K'}{\sigma}{\theta_K}{n}{w} \vdash
\config{\hat{\theta}}{\hat{C}}{\hat{K'}}{\hat{\sigma}}{\hat{\theta_K}}{n+1}{\hat{w}} \vdash^{*}_{\mathtt{min}}
\config{\theta'}{\modownarrow}{[]}{\sigma'}{\theta_K'}{n + n'}{w'}$,
where $(C, K') \neq (\modownarrow, [])$, as otherwise the configuration would reduce in $0$ steps.

By Lemma~\ref{lemma:add-k-step},
$\config{\theta}{C}{K'@K}{\sigma}{\theta_K}{n}{w} \vdash
\config{\hat{\theta}}{\hat{C}}{\hat{K'}@K}{\hat{\sigma}}{\hat{\theta_K}}{n+1}{\hat{w}}$ 
and by induction hypothesis,
$\config{\hat{\theta}}{\hat{C}}{\hat{K'} @ K}{\hat{\sigma}}{\hat{\theta_K}}{n+1}{\hat{w}} \vdash^{*}
\config{\theta'}{\modownarrow}{K}{\sigma'}{\theta_K'}{n + 1 + (n'-1)}{w'}$,
which ends the proof.
\end{itemize}
\qed \end{proof}
\begin{corollary}
If $\config{\theta}{C}{[]}{\sigma}{\theta_K}{n}{w}
\vdash^{*}_{\mathtt{min}}
\config{\theta'}{\modownarrow}{[]}{\sigma'}{\theta_K}{n + n'}{w'}$
and $\sigma' \neq \failure$,
then
$\config{\theta}{C}{K}{\sigma}{\theta_K}{n}{w}
\vdash^{*}
\config{\theta'}{\modownarrow}{K}{\sigma'}{\theta_K}{n + n'}{w'}$.
\end{corollary}

We also need to show that reductions leading to a failed observation are also preserved
when appending a continuation.

\begin{lemma} \label{lemma:add-k-fail}
If $\config{\theta}{C}{K}{\sigma}{\theta_K}{n}{w} \vdash^{*}
\config{\theta'}{\modownarrow}{[]}{\failure}{\theta'_K}{n+n'}{w'}$
then for all $K''$,
$\config{\theta}{C}{K @ K''}{\sigma}{\theta_K}{n}{w} \vdash^{*}
\config{\theta'}{C'}{[]}{\failure}{\theta'_K}{n+n'}{w'}$.
\end{lemma}
\begin{proof}
If $n' = 0$, the result follows trivially.

If $n' > 0$, then we have $\sigma \neq \failure$ (otherwise the initial configuration would not reduce),
and so the last rule in the derivation of 
$\config{\theta}{C}{K}{\sigma}{\theta_K}{n}{w} \vdash^{*}
\config{\theta'}{C'}{K'}{\failure}{\theta'_K}{n+n'}{w'}$
must have been (condition-false). 

Hence,
$\config{\theta}{C}{K}{\sigma}{\theta_K}{n}{w} \vdash^{*}
\config{\theta'}{\mathtt{observe}(\phi)}{\hat{K}}{\sigma'}{\theta'_K}{n+n'-1}{w'} \vdash
\config{\theta'}{\modownarrow}{[]}{\failure}{\theta'_K}{n+n'}{w'}$,
where $\sigma' \neq \failure$
and $\sigma'(\phi) = \mathtt{false}$.
By Lemma~\ref{lemma:add-k},
$\config{\theta}{C}{K@ K''}{\sigma}{\theta_K}{n}{w} \vdash^{*}
\config{\theta'}{\mathtt{observe}(\phi)}{\hat{K} @ K''}{\sigma'}{\theta'_K}{n+n'-1}{w'}$.
By applying  (condition-false) again, we get 
$\config{\theta'}{\mathtt{observe}(\phi)}{\hat{K} @ K''}{\sigma'}{\theta'_K}{n+n'-1}{w'} \vdash
\config{\theta'}{\modownarrow}{[]}{\failure}{\theta'_K}{n+n'}{w'}$,
as required.
\qed \end{proof}

\begin{lemma}  \label{lemma:seq-add-k-fail}
If $C_1 \neq C_1'; C_1''$ and 
$\config{\pi_L(\theta)}{C_1}{[]}{\sigma}{\pi_R(\theta) \mathrel{::} \theta_K}{0}{1} \vdash^{*}
\config{\theta'}{\modownarrow}{[]}{\failure}{\theta_K'}{n }{w}$,
then
$\config{\theta}{C_1;C_2}{[]}{\sigma}{\theta_K}{0}{1} \vdash^{*}
\config{\theta'}{\modownarrow}{[]}{\failure}{\theta_K'}{n+1}{w}$.
\end{lemma}
\begin{proof}
By Lemma~\ref{lemma:add-k-fail} 
$\config{\pi_L(\theta)}{C_1}{[C_2]}{\sigma}{\pi_R(\theta) \mathrel{::} \theta_K}{0}{1} \vdash^{*}
\config{\theta'}{\modownarrow}{[]}{\failure}{\theta_K'}{n }{w}$. 
As $\config{\theta}{C_1;C_2}{[]}{\sigma}{\theta_K}{0}{1}  \vdash
\config{\pi_L(\theta)}{C_1}{[C_2]}{\sigma}{\pi_R(\theta) \mathrel{::}\theta_K}{1}{1}$ by (seq),
Lemma~\ref{lemma:change-n-w} yields
$\config{\theta}{C_1;C_2}{[]}{\sigma}{\theta_K}{0}{1} \vdash^{*}
\config{\theta'}{\modownarrow}{[]}{\failure}{\theta_K'}{n+1}{w}$.
\qed \end{proof}

\subsubsection{Sequencing}
We now use the above results to relate the final and intermediate configurations in the reduction
of a statement $C_1$ to the intermediate configurations reached when reducing $C_1;C_2$.

\begin{lemma}[Context evaluation for simple sequencing] \label{lemma:context-simple}
If $C_1 \neq C_1'; C_1''$ and 
$\config{\theta}{C_1}{[]}{\sigma}{\theta_K}{n}{w}
\vdash^{*}_{\mathtt{min}}
\config{\theta'}{\modownarrow}{[]}{\sigma'}{\theta_K}{n + n'}{w'}$
and $\sigma' \neq \failure$,
then \newline
$\config{\theta \mathcal{::} \pi_L(\theta_K)}{C_1;C_2}{[]}{\sigma}{\pi_R(\theta_K)}{n}{w}
\vdash^{*}
\config{\pi_L(\theta_K)}{C_2}{[]}{\sigma'}{\pi_R(\theta_K)}{n + n' + 2}{w'}$.

\end{lemma}
\begin{proof}

By (seq): 
$\config{\theta \mathcal{::} \pi_L(\theta_K)}{C_1;C_2}{[]}{\sigma}{\pi_R(\theta_K)}{n}{w}
\vdash
\config{\theta}{C_1}{[C_2]}{\sigma}{\theta_K}{n+1}{w}
$.

By Lemma \ref{lemma:add-k-final} (and the fact that we can change $n$):
$\config{\theta}{C_1}{[C_2]}{\sigma}{\theta_K}{n+1}{w}
\vdash^{*}
\config{\theta'}{\modownarrow}{[C_2]}{\sigma'}{\theta_K}{(n + 1) + n'}{w'}$.

By (pop),
$\config{\theta'}{\modownarrow}{[C_2]}{\sigma'}{\theta_K}{(n + 1) + n'}{w'}
\vdash \config{\pi_L(\theta_K)}{C_2}{[]}{\sigma'}{\pi_R(\theta_K)}{(n + 1) + n' + 1}{w'}
$, as required.
\qed \end{proof}

\begin{lemma} \label{lemma:add-seq}
If $C_1 \neq C_1'; C_1''$ and 
$\config{\theta}{C_1}{[]}{\sigma}{\theta_K}{n}{w}
\vdash^{*}
\config{\theta'}{C'}{K}{\sigma'}{\theta'_K}{n + n'}{w'}$
and $\sigma' \neq \failure$
and $(C', K') \neq (\modownarrow, [])$,
then
$\config{\theta \mathcal{::} \pi_L(\theta_K)}{C_1;C_2}{[]}{\sigma}{\pi_R(\theta_K)}{n}{w}
\vdash^{*}
\config{\theta' }{C'}{K @ [C_2]}{\sigma'}{\theta'_K}{n + n' + 1}{w'}$.
\end{lemma}
\begin{proof}
By (seq), we have 
$\config{\theta \mathcal{::} \pi_L(\theta_K)}{C_1;C_2}{[]}{\sigma}{\pi_R(\theta_K)}{n}{w} \vdash
\config{\theta}{C_1}{[C_2]}{\sigma}{\theta_K}{n+1}{w}$. Then, by
Corollary~\ref{corr:change-n-w-k}, 
$\config{\theta}{C_1}{[C_2]}{\sigma}{\theta_K}{n+1}{w}
\vdash^{*}
\config{\theta'}{C'}{K@[C_2]}{\sigma'}{\theta'_K}{n + n'+1}{w'}$,
as required.
\qed \end{proof}

\subsubsection{Splitting a sequence evaluation}
We now show that if a sequence $C_1;C_2$ of statements evaluates under entropy $\theta$ to a proper
state, then $C_1$ in itself must evaluate under $\pi_L(\theta)$, and that if the evaluation of
$C_1;C_2$ results in an error, then $C_1$ cannot diverge. These properties will be needed to show
compositionality of the semantics.

To prove the first of the above properties, we first prove that if a configuration with an empty continuation reduces completely, then the 
continuation entropy $\theta_K$ in the final configuration will be identical to the original one (intermediate
steps may extend $\theta_{K}$, but all sub-entropies added to $\theta_K$ will subsequently be removed).
In the following lemma, we write $|K|$ for the length of list $K$.

\begin{lemma} \label{lemma:recover-theta-k}
If $\config{\theta}{C}{K}{\sigma}{\hat{\theta_K}}{n}{w} \vdash^{*}
\config{\theta'}{\modownarrow}{[]}{\sigma'}{\theta_K'}{n + n'}{w'}$
and $\sigma' \neq \failure$
and $\pi_R^{|K|}(\hat{\theta_K}) = \theta_K$,
then $\theta_K' = \theta_K$.
\end{lemma}
\begin{proof}
By induction on $n'$:
\begin{itemize}
\item Base case: $n'=0$: then obviously $|K| = 0$ and $\hat{\theta_K} = \theta_K$,
so the result follows trivially.
\item Induction step: if $n' > 0$, then
$\config{\theta}{C}{K}{\sigma}{\hat{\theta_K}}{n}{w} \vdash
\config{\theta''}{C''}{K'}{\sigma''}{\theta_K''}{n+1}{w''} \vdash^{*}
\config{\theta'}{\modownarrow}{[]}{\sigma'}{\theta_K'}{n + n'}{w'}$.

Now we need to split on the first rule in this derivation chain.

If the first transition was derived with (seq), then $|K'| = |K| + 1$
and $\theta_K'' = \pi_R(\theta) \mathrel{::} \hat{\theta_K}$.
We have $\pi_R^{|K'|}(\theta_K'') = \pi_R^{|K| + 1}(\pi_R(\theta) \mathrel{::} \hat{\theta_K})
= \pi_R^{|K|}(\pi_R(\pi_R(\theta) \mathrel{::} \hat{\theta_K})) = \pi_R^{|K|}(\hat{\theta_K}) = \theta_K$,
so by induction hypothesis, $\theta_K' = \theta_K$.

If the first transition was derived with (pop), then  $|K'| = |K| - 1$
and $\theta_K'' = \pi_R(\hat{\theta_K}) $.
Thus, $\pi_R^{|K'|}(\theta_K'') = \pi_R^{|K|-1}(\pi_R(\hat{\theta_K})) =  \pi_R^{|K|}(\hat{\theta_K}) = \hat{\theta_K}$,
so by induction hypothesis, $\theta_K' = \theta_K$.

Otherwise, we have $K' = K$ (note that $\sigma' \neq \failure$ implies $\sigma'' \neq \failure$) and $\theta_K'' = \hat{\theta_K}$, so $\pi_R^{|K'|}(\theta_K'') = \theta_K$.
By induction hypothesis, $\theta_K' = \theta_K$.%, and so $\theta_K' = \theta_K$.
\end{itemize}
\qed \end{proof}

\begin{corollary} \label{corr:theta-k-unchanged}
If $\config{\theta}{C}{[]}{\sigma}{\theta_K}{n}{w} \vdash^{*}
\config{\theta'}{\modownarrow}{[]}{\sigma'}{\theta_K'}{n + n'}{w'}$
and $\sigma' \neq \failure$,
then $\theta_K' = \theta_K$.
\end{corollary}

We now prove that if $C_1;C_2$ successfully evaluates with entropy $\theta$, then
$C_1$ also successfully evaluates with entropy $\pi_L(\theta)$.

\begin{lemma}[Interpolation for Continuations] \label{lemma:interpolation-cont}
If  $\config{\theta}{C}{K_1 @ K_2}{\sigma}{\theta_K}{n}{w}
\vdash^{*}
\config{\theta'}{\modownarrow}{[]}{\sigma'}{\theta_K'}{n + n'}{w'}$
and $\sigma' \neq \failure$,
then
$\config{\theta}{C}{K_1}{\sigma}{\theta_K}{n}{w}
\vdash^{*}
\config{\theta''}{\modownarrow}{[]}{\sigma''}{\theta_K''}{n + n''}{w''}$,
where $\sigma'' \neq \failure$.
\end{lemma}
\begin{proof}
By induction on $n'$.
\begin{itemize}
\item Base case: $n' = 0$: in this case, $C = \modownarrow$ and $K_1 = K_2 = []$, so the result follows trivially.
\item Induction step: suppose
$\config{\theta}{C}{K_1 @ K_2}{\sigma}{\theta_K}{n}{w}
\vdash
\config{\hat{\theta}}{\hat{C}}{\hat{K}}{\hat{\sigma}}{\hat{\theta_K}}{n+1}{\hat{w}}
\vdash^{*}
\config{\theta'}{\modownarrow}{[]}{\sigma'}{\theta_K'}{n + n'}{w'}$.

If $\config{\theta}{C}{K_1 @ K_2}{\sigma}{\theta_K}{n}{w}
\vdash
\config{\hat{\theta}}{\hat{C}}{\hat{K}}{\hat{\sigma}}{\hat{\theta_K}}{n+1}{\hat{w}}$
was derived with (seq), then $C = C_1; C_2$, $C_1 \neq C_1';C_1''$, 
$\hat{K} = C_2 \mathrel{::} K_1 @ K_2$, $\hat{\theta} = \pi_L(\theta)$, $\hat{w} = w$
and $\hat{\theta_K} = \pi_R(\theta) \mathrel{::} \theta_K$. By (seq), we have
$\config{\theta}{C_1;C_2}{K_1}{\sigma}{\theta_K}{n}{w}
\vdash
\config{ \pi_L(\theta)}{C_1}{C_2 \mathrel{::} K_1}{\sigma}{\pi_R(\theta) \mathrel{::} \theta_K}{n+1}{w}$.
By induction hypothesis, 
$\config{ \pi_L(\theta)}{C_1}{C_2 \mathrel{::} K_1}{\sigma}{\pi_R(\theta) \mathrel{::} \theta_K}{n+1}{w}
\vdash^{*}
\config{\theta''}{\modownarrow}{[]}{\sigma''}{\theta_K''}{n + n''}{w''}$ and $\sigma'' \neq \failure$.
Hence, 
$\config{\theta}{C_1;C_2}{K_1}{\sigma}{\theta_K}{n}{w}
\vdash^{*}
\config{\theta''}{\modownarrow}{[]}{\sigma''}{\theta_K''}{n + n''}{w''}$,
as required.

If $\config{\theta}{C}{K_1 @ K_2}{\sigma}{\theta_K}{n}{w}
\vdash
\config{\hat{\theta}}{\hat{C}}{\hat{K}}{\hat{\sigma}}{\hat{\theta_K}}{n+1}{\hat{w}}$
was derived with (pop), then $C = \modownarrow$, $K_1 @ K_2 = \hat{C}  \mathrel{::}\hat{K}$, 
$\hat{w} = w$, $\hat{\theta} = \pi_L(\theta_K)$
and $\hat{\theta_K} = \pi_R(\theta_K) $.
\begin{itemize}
\item If $K_1 \neq []$, then $K_1  = \hat{C}  \mathrel{::}\hat{K_1}$ and $\hat{K} = \hat{K_1} @ K_2$
and we have $\config{\theta}{\modownarrow}{\hat{C}  \mathrel{::}\hat{K_1}}{\sigma}{\theta_K}{n}{w}
\vdash \config{ \pi_L(\theta_K)}{\hat{C} }{\hat{K_1}}{\sigma}{\pi_R(\theta_K)}{n+1}{w}$.
By induction hypothesis,
$\config{ \pi_L(\theta_K)}{\hat{C} }{\hat{K_1}}{\sigma}{\pi_R(\theta_K)}{n+1}{w}
\vdash^{*}
\config{\theta''}{\modownarrow}{[]}{\sigma''}{\theta_K''}{n + n''}{w''}$ and $\sigma'' \neq \failure$.
Hence, we have
$\config{\theta}{\modownarrow}{\hat{C}  \mathrel{::}\hat{K_1}}{\sigma}{\theta_K}{n}{w} \vdash^{*}
\config{\theta''}{\modownarrow}{[]}{\sigma''}{\theta_K''}{n + n''}{w''}$.

\item If $K_1 = []$, then trivially
$\config{\theta}{\modownarrow}{[]}{\sigma}{\theta_K}{n}{w} \vdash^{*}
\config{\theta}{\modownarrow}{[]}{\sigma}{\theta_K}{n}{w}$ in zero steps.

\end{itemize}

Otherwise, $\hat{K} = K_1 @ K_2$ and $\hat{\theta_K} = \theta_K$ and by inspection of the reduction rules, 
$\config{\theta}{C}{K_1}{\sigma}{\theta_K}{n}{w}
\vdash
\config{\hat{\theta}}{\hat{C}}{K_1}{\hat{\sigma}}{\theta_K}{n+1}{\hat{w}}$. Hence,
by induction hypothesis,
$\config{\theta}{C}{K_1}{\sigma}{\theta_K}{n}{w}
\vdash
\config{\hat{\theta}}{\hat{C}}{K_1}{\hat{\sigma}}{\theta_K}{n+1}{\hat{w}}
\vdash^{*}
\config{\theta''}{\modownarrow}{[]}{\sigma''}{\theta_K''}{n + n''}{w''}$ and $\sigma'' \neq \failure$, as required.
\end{itemize}
\qed \end{proof}

\begin{lemma}[Interpolation] \label{lemma:interpolation}
If  $C_1 \neq C_1'; C_1''$ and $\config{\theta}{C_1;C_2}{[]}{\sigma}{\theta_K}{n}{w}
\vdash^{*}
\config{\theta'}{\modownarrow}{[]}{\sigma'}{\theta_K}{n + n'}{w'}$
and $\sigma' \neq \failure$,
then
$\config{\pi_L(\theta)}{C_1}{[]}{\sigma}{\pi_R(\theta) \mathrel{::} \theta_K}{n}{w}
\vdash^{*}
\config{\theta''}{\modownarrow}{[]}{\sigma''}{\theta_K}{n + n''}{w''}$,
where $\sigma'' \neq \failure$.
\end{lemma}
\begin{proof}
The first rule applied in the derivation of $\config{\theta}{C_1;C_2}{[]}{\sigma}{\theta_K}{n}{w}
\vdash^{*}
\config{\theta'}{\modownarrow}{[]}{\sigma'}{\theta_K}{n + n'}{w'}$ is (seq),
which gives 
 $\config{\theta}{C_1;C_2}{[]}{\sigma}{\theta_K}{n}{w}
\vdash
\config{\pi_L(\theta)}{C_1}{[C_2]}{\sigma}{\pi_R(\theta) \mathrel{::} \theta_K}{n+1}{w}$.
Hence, $\config{\pi_L(\theta)}{C_1}{[C_2]}{\sigma}{\pi_R(\theta) \mathrel{::} \theta_K}{n+1}{w}
\vdash^{*}
\config{\theta'}{\modownarrow}{[]}{\sigma'}{\theta_K}{n + n'}{w'}$. By applying Lemma~\ref{lemma:interpolation-cont}
with $K_1= []$ and Corollary~\ref{corr:theta-k-unchanged}, we get
$\config{\pi_L(\theta)}{C_1}{[]}{\sigma}{\pi_R(\theta) \mathrel{::} \theta_K}{n+1}{w}
\vdash^{*}
\config{\theta''}{\modownarrow}{[]}{\sigma''}{\pi_R(\theta) \mathrel{::} \theta_K}{n + n''}{w''}$, where $\sigma'' \neq \failure$,
as required.
\qed \end{proof}

Finally, we show that if the evaluation of $C_1;C_2$ with entropy $\theta$ yields an error, then
the evaluation of $C_1$ under $\pi_L(\theta)$ either terminates successfully or also results in an error
(depending on where the error in the evaluation of $C_1;C_2$ occurred)---at any rate, $C_1$ does not diverge.

\begin{lemma} \label{lemma:c1-c2-stuck}
If $C_1 \neq  C_1'; C_2'$ and $\config{\theta}{C_1;C_2}{[]}{\sigma}{\theta_K}{0}{1} \vdash^{*}
\config{\theta'}{C'}{K}{\sigma'}{\theta_K'}{n}{w} \nvdash$, then 
either $\config{\pi_L(\theta)}{C_1}{[]}{\sigma}{\pi_R(\theta) \mathrel{::}\theta_K}{0}{1}
\vdash^{*} \config{\theta''}{\modownarrow}{[]}{\sigma''}{\theta_K}{n'}{w'}$
or  $\config{\pi_L(\theta)}{C_1}{[]}{\sigma}{\pi_R(\theta) \mathrel{::}\theta_K}{0}{1}
\vdash^{*} \config{\theta''}{C_1''}{K''}{\sigma''}{\theta_K}{n'}{w'} \nvdash$.
\end{lemma}
\begin{proof}
The statement in the lemma is equivalent to saying that it is \emph{not} the case that
for all $k$,
$\config{\pi_L(\theta)}{C_1}{[]}{\sigma}{\pi_R(\theta) \mathrel{::} \theta_K}{0}{1}
\vdash^{*} \config{\theta''}{C_1''}{K''}{\sigma''}{\pi_R(\theta) \mathrel{::} \theta_K}{k}{w'}$ 
%TODO: Make it explicit that the components all depend on k?
with
$(C_1'', K'') \neq (\modownarrow, [])$. Suppose for contradiction that the negation of this statement holds.
By (seq), we have 
$\config{\theta}{C_1;C_2}{[]}{\sigma}{\theta_K}{0}{1} \vdash
\config{\pi_L(\theta)}{C_1}{[C_2]}{\sigma}{\pi_R(\theta) \mathrel{::} \theta_K}{1}{1}$,
so $\config{\pi_L(\theta)}{C_1}{[C_2]}{\sigma}{\pi_R(\theta) \mathrel{::} \theta_K}{1}{1}
\vdash^{*}  \config{\theta'}{C'}{K}{\sigma'}{\theta_K'}{n}{w}$. 

Take $k = n - 1$. Then we have 
$\config{\pi_L(\theta)}{C_1}{[]}{\sigma}{\pi_R(\theta) \mathrel{::} \theta_K}{0}{1}
\vdash^{*} \config{\theta''}{C_1''}{K''}{\sigma''}{\pi_R(\theta) \mathrel{::} \theta_K}{n-1}{w'}
\vdash  \config{\hat{\theta}}{\hat{C_1}}{\hat{K}}{\hat{\sigma}}{\hat{\theta_K}}{n}{\hat{w}}$,
where $\sigma'' \neq \failure$ (otherwise the middle configuration would not reduce)
and $(C_1'', K'') \neq (\modownarrow, [])$.
By Corollary~\ref{corr:change-n-w-k}, we have
$\config{\pi_L(\theta)}{C_1}{[C_2]}{\sigma}{\pi_R(\theta) \mathrel{::} \theta_K}{1}{1}
\vdash^{*} \config{\theta''}{C_1''}{K''@[C_2]}{\sigma''}{\pi_R(\theta) \mathrel{::} \theta_K}{n}{w'}$.
Hence, 
$\config{\theta}{C_1;C_2}{[]}{\sigma}{\theta_K}{0}{1} \vdash^{*}
\config{\theta''}{C_1''}{K''@[C_2]}{\sigma''}{\pi_R(\theta) \mathrel{::} \theta_K}{n}{w'}$
and $\config{\theta''}{C_1''}{K''@[C_2]}{\sigma''}{\pi_R(\theta) \mathrel{::} \theta_K}{n}{w'}
= \config{\theta'}{C'}{K}{\sigma'}{\theta_K'}{n}{w}$, since reduction is deterministic. By Lemma~\ref{lemma:add-k-step},
this implies that $\config{\theta'}{C'}{K}{\sigma'}{\theta_K'}{n}{w}$ reduces, contradicting the assumption.
\qed \end{proof}
\begin{corollary} \label{corr:c1-c2-stuck}
If $C_1 \neq  C_1'; C_2'$ and $\config{\theta}{C_1;C_2}{[]}{\sigma}{\theta_K}{0}{1} \vdash^{*}
\config{\theta'}{C'}{K}{\sigma'}{\theta_K'}{n}{w} \nvdash$, then 
$\mathbf{O}_{C_1}^\sigma(\pi_L(\theta)) \neq \diverge$.
\end{corollary}

\subsection{Properties of the semantic functions}
\label{app:proofs-sem-fun}
 
\subsubsection{Compositionality of sequencing.} 
A desirable and useful property of the semantic functions is compositionality with respect to sequencing, i.e., the ability to define $\mathbf{O}_{C_1;C_2}^{\sigma}$ in terms of $\mathbf{O}_{C_1}^{\sigma_1}$ and $\mathbf{O}_{C_2}^{\sigma_2}$ for some states $\sigma_1$ and $\sigma_2$.
Similarly for $\mathbf{SC}_{C_1;C_2}^{\sigma}$.
We can easily express the semantics of $C_1; C_2$ in terms of the semantics of $C_1$ and $C_2$ if $C_1$ is not a sequence of statements.
(Recall the explanation of the rule (seq).)
\begin{proposition}[Simple sequencing for final states] \label{lemma:o-sc-seq}
If $C_1 \neq  C_1'; C_2'$, then:
$$
\mathbf{O}_{C_1;C_2}^{\sigma}(\theta) =  \mathbf{O}_{C_2}^{\tau}(\pi_R(\theta))
\quad \mbox{and} \quad
\mathbf{SC}_{C_1;C_2}^{\sigma}(\theta) = \mathbf{SC}_{C_1}^\sigma (\pi_L(\theta)) \cdot \mathbf{SC}_{C_2}^{\tau}(\pi_R(\theta))
$$
where $\tau$ stands for the state $\mathbf{O}_{C_1}^\sigma(\pi_L(\theta))$.
\end{proposition}
%\begin{proof}
%The proof is based on straightforward reasoning with reduction rules. 
%The key observation
%is that when $C_1; C_2$ (where $C_1 \neq  C_1'; C_2'$) is evaluated with entropy $\theta$, the (seq) rule assigns the sub-entropy $\pi_L(\theta)$ to $C_1$ and $\pi_R(\theta)$ to $C_2$.
%A full proof can be found in Appendix~\ref{app:proofs-sem-fun}.
% \qed \end{proof}

Below, we prove Proposition~\ref{lemma:o-sc-seq}. To simplify presentation, we split it into two separate lemmas,
one concerning final states and one concerning scores.

\begin{lemma}[Simple sequencing for final states] \label{lemma:o-seq}
If $C_1 \neq  C_1'; C_2'$, then $\mathbf{O}_{C_1;C_2}^{\sigma}(\theta) =  \mathbf{O}_{C_2}^{\mathbf{O}_{C_1}^\sigma(\pi_L(\theta)) }(\pi_R(\theta))$
\end{lemma}
\begin{proof}
If $\sigma = \diverge$, then $LHS = RHS = \diverge$ directly by definition.

If $\sigma = \failure$, the result also follows trivially, so let us suppose $\sigma \neq \failure$ and $\sigma \neq \diverge$. We need to consider several cases:

\begin{itemize}
\item If $\mathbf{O}_{C_1}^\sigma(\pi_L(\theta)) = \failure$, then
$ \config{\pi_L(\theta)}{C_1}{[]}{\sigma}{\pi_R(\theta) \mathrel{::}\theta_K}{0}{1} \vdash^{*}
\config{\theta'}{C_1'}{K}{\tau}{\theta_K'}{n }{w} \nvdash $.
By (seq), we have 
$\config{\theta}{C_1;C_2}{[]}{\sigma}{\theta_K}{0}{1}  \vdash
\config{\pi_L(\theta)}{C_1}{[C_2]}{\sigma}{\pi_R(\theta) \mathrel{::}\theta_K}{1}{1}$.

If $\tau \neq \failure$, then by Lemmas~\ref{lemma:add-k} and \ref{lemma:change-n-w},
$\config{\pi_L(\theta)}{C_1}{[C_2]}{\sigma}{\pi_R(\theta) \mathrel{::}\theta_K}{1}{1} \vdash^{*}
\config{\theta'}{C_1'}{K@[C_2]}{\tau}{\theta_K'}{n+1}{w} \nvdash $.
Moreover, $\config{\theta'}{C_1'}{K}{\tau}{\theta_K'}{n }{w} \nvdash $ implies
$C_1' \neq \modownarrow$ (because otherwise the configuration would reduce by (final) or (pop)), so by inspection,
$\config{\theta'}{C_1'}{K@[C_2]}{\tau}{\theta_K'}{n +1}{w} \nvdash $.
Thus, $\mathbf{O}_{C_1; C_2}^\sigma(\theta) = \failure$.

If $\tau = \failure$, then $C_1'  = \modownarrow$, $K = []$ and by Lemmas~\ref{lemma:add-k-fail} and \ref{lemma:change-n-w} we have
$\config{\pi_L(\theta)}{C_1}{[C_2]}{\sigma}{\pi_R(\theta) \mathrel{::}\theta_K}{1}{1} \vdash^{*}
\config{\theta'}{\modownarrow}{[]}{\failure}{\theta_K'}{n+1 }{w} \nvdash $. Hence, 
$\mathbf{O}_{C_1; C_2}^\sigma(\theta) = \failure$.

\item If $\mathbf{O}_{C_1}^\sigma(\pi_L(\theta)) = \diverge$, then $RHS = \diverge$.
Moreover, we have neither $\config{\pi_L(\theta)}{C_1}{[]}{\sigma}{\pi_R(\theta) \mathrel{::} \theta_K}{0}{1} \vdash^{*}
\config{\theta'}{\modownarrow}{[]}{\tau}{\theta_K}{n }{w}$ nor
$\config{\pi_L(\theta)}{C_1}{[]}{\sigma}{\pi_R(\theta) \mathrel{::}\theta_K}{0}{1} \vdash^{*}
\config{\theta'}{C'}{K}{\tau}{\theta_K'}{n }{w} \nvdash$.

Now, suppose for contradiction that $LHS \neq \diverge $.
Then we have either $\config{\theta}{C_1;C_2}{[]}{\sigma}{\theta_K}{0}{1} \vdash^{*} \config{\theta'}{\modownarrow}{[]}{\tau}{\theta_K}{n }{w}$ (with
$\tau \neq \failure$) or
$\config{\theta}{C_1;C_2}{[]}{\sigma}{\theta_K}{0}{1} \vdash^{*} \config{\theta'}{C'}{K}{\tau}{\theta_K'}{n }{w} \nvdash$.

First, suppose that $\config{\theta}{C_1;C_2}{[]}{\sigma}{\theta_K}{0}{1} \vdash^{*} \config{\theta'}{\modownarrow}{[]}{\tau}{\theta_K}{n }{w}$, 
where $\tau \neq \failure$. 
By Lemma~\ref{lemma:interpolation}, this implies that 
$\config{\pi_L(\theta)}{C_1}{[]}{\sigma}{\pi_R(\theta) \mathrel{::} \theta_K}{0}{1} \vdash^{*} 
\config{\theta''}{\modownarrow}{[]}{\tau'}{ \pi_R(\theta) \mathrel{::}\theta_K}{n' }{w'}$ 
%for some $\tau' \neq \failure$,
and so $\mathbf{O}_{C_1}^\sigma(\pi_L(\theta)) = \tau' \neq \diverge$, contradicting the assumption.

If $\config{\theta}{C_1;C_2}{[]}{\sigma}{\theta_K}{0}{1} \vdash^{*} \config{\theta'}{C'}{K}{\tau}{\theta_K'}{n }{w} \nvdash$,
then by Corollary~\ref{corr:c1-c2-stuck}, we get a contradiction.

\item If $\mathbf{O}_{C_1}^\sigma(\pi_L(\theta)) \notin \{\failure, \diverge \}$, but 
$ \mathbf{O}_{C_2}^{\mathbf{O}_{C_1}^\sigma(\pi_L(\theta)) }(\pi_R(\theta)) = \failure$,
we have 
$\config{\pi_L(\theta)}{C_1}{[]}{\sigma}{\pi_R(\theta) \mathrel{::} \theta_K}{0}{1} \vdash^{*}_{\mathtt{min}}
\config{\theta'}{\modownarrow}{[]}{\tau'}{ \pi_R(\theta) \mathrel{::}\theta_K}{n}{w}$ for
some $\tau' \neq \failure$, where $\mathbf{O}_{C_1}^\sigma(\pi_L(\theta)) = \tau'$,
and 
$\config{\pi_R(\theta)}{C_2}{[]}{\tau'}{\theta_K}{0}{1} \vdash^{*}
\config{\theta''}{C''}{K'}{\tau}{\theta_K'}{n' }{w'} \nvdash$.
By Lemma~\ref{lemma:context-simple}, 
$\config{\theta}{C_1;C_2}{[]}{\sigma}{\theta_K}{0}{1} \vdash^{*} 
\config{\pi_R(\theta)}{C_2}{[]}{\tau'}{\theta_K}{n+2}{w}$.
%By Corollary~\ref{corr:change-n-w-k} WRONG , 
By Lemma~\ref{lemma:change-n-w}, 
$\config{\pi_R(\theta)}{C_2}{[]}{\tau'}{\theta_K}{n+2}{w} \vdash^{*}
\config{\theta''}{C''}{K'}{\tau}{\theta_K'}{n + 2 + n' }{ww'}$, where
the last configuration clearly does not reduce, as changing the last two
components cannot make any rule apply. Hence,
$\mathbf{O}_{C_1;C_2}^{\sigma}(\theta) = \failure$, as required.

\item If $\mathbf{O}_{C_1}^\sigma(\pi_L(\theta)) \notin \{\failure, \diverge \}$, but 
$ \mathbf{O}_{C_2}^{\mathbf{O}_{C_1}^\sigma(\pi_L(\theta)) }(\pi_R(\theta)) = \diverge$,
we have again
$\config{\pi_L(\theta)}{C_1}{[]}{\sigma}{\pi_R(\theta) \mathrel{::} \theta_K}{0}{1} \vdash^{*}_{\mathtt{min}}
\config{\theta'}{\modownarrow}{[]}{\tau'}{ \pi_R(\theta) \mathrel{::}\theta_K}{n}{w}$ for
some $\tau' \neq \failure$.
Again, by Lemma~\ref{lemma:context-simple}, we have
$\config{\theta}{C_1;C_2}{[]}{\sigma}{\theta_K}{0}{1} \vdash^{*} 
\config{\pi_R(\theta)}{C_2}{[]}{\tau'}{\theta_K}{n+2}{w}$, but we have
neither $\config{\pi_R(\theta)}{C_2}{[]}{\tau'}{\theta_K}{0 }{1} \vdash^{*}
\config{\theta''}{\modownarrow}{[]}{\tau''}{\theta_K}{n'}{w'}$ nor
$\config{\pi_R(\theta)}{C_2}{[]}{\tau'}{\theta_K}{0}{1} \vdash^{*}
\config{\theta''}{C''}{K'}{\tau}{\theta_K'}{n' }{w'} \nvdash$.

Suppose for contradiction that $LHS \neq \diverge $.
Then we have either $\config{\theta}{C_1;C_2}{[]}{\sigma}{\theta_K}{0}{1} \vdash^{*} 
\config{\theta'}{\modownarrow}{[]}{\tau}{\theta_K}{\hat{n}}{\hat{w}}$ (with
$\tau \neq \failure$) or
$\config{\theta}{C_1;C_2}{[]}{\sigma}{\theta_K}{0}{1} \vdash^{*} \config{\theta'}{C'}{K}{\tau}{\theta_K'}{\hat{n}}{\hat{w}} \nvdash$.

In the former case, the determinicity of reduction implies
$\config{\pi_R(\theta)}{C_2}{[]}{\tau'}{\theta_K}{n+2}{w}
\vdash^{*} \config{\theta'}{\modownarrow}{[]}{\tau}{\theta_K}{\hat{n}}{\hat{w}}$, 
so by %Corollary~\ref{corr:change-n-w-k}, 
Lemma~\ref{lemma:change-n-w}, 
$\config{\pi_R(\theta)}{C_2}{[]}{\tau'}{\theta_K}{0}{1}
\vdash^{*} \config{\theta'}{\modownarrow}{[]}{\tau}{\theta_K}{\hat{n}-n-2}{\hat{w} / w}$, 
which contradicts the assumption.
%TODO (LP) state this as a formal lemma?
%OK, we don't need minimal reduction here

Similarly, in the latter case, $\config{\pi_R(\theta)}{C_2}{[]}{\tau'}{\theta_K}{n+2}{w}\vdash^{*} 
\config{\theta'}{C'}{K}{\tau}{\theta_K'}{\hat{n}}{\hat{w}} \nvdash$, which violates the assumption.

Hence, $\mathbf{O}_{C_1;C_2}^{\sigma}(\theta) = \diverge $.

\item \sloppy Finally, suppose that $\mathbf{O}_{C_1}^\sigma(\pi_L(\theta)) \notin \{\failure, \diverge \}$ and 
$\mathbf{O}_{C_2}^{\mathbf{O}_{C_1}^\sigma(\pi_L(\theta)) }(\pi_R(\theta))\notin \{\failure, \diverge \}$.
Then we have again $\config{\pi_L(\theta)}{C_1}{[]}{\sigma}{\pi_R(\theta) \mathrel{::} \theta_K}{0}{1} \vdash^{*}_{\mathtt{min}} 
\config{\theta'}{\modownarrow}{[]}{\tau'}{ \pi_R(\theta) \mathrel{::}\theta_K}{n' }{w'}$ for
some $\tau' \neq \failure$
and $\config{\theta}{C_1;C_2}{[]}{\sigma}{\theta_K}{0}{1} \vdash^{*} 
\config{\pi_R(\theta)}{C_2}{[]}{\tau'}{\theta_K}{n' }{w'}$ by Lemma~\ref{lemma:context-simple}.
Since $\mathbf{O}_{C_1}^\sigma(\pi_L(\theta)) = \tau'$ and 
$\mathbf{O}_{C_2}^{\mathbf{O}_{C_1}^\sigma(\pi_L(\theta)) }(\pi_R(\theta)) = \tau'' \neq \failure$,
 we have $\config{\pi_R(\theta)}{C_2}{[]}{\tau'}{\theta_K}{n' }{w'} \vdash^{*}
\config{\theta''}{\modownarrow}{[]}{\tau''}{\theta_K}{n''}{w''}$.
This also implies that

$\config{\theta}{C_1;C_2}{[]}{\sigma}{\theta_K}{0}{1} \vdash^{*} 
\config{\theta''}{\modownarrow}{[]}{\tau''}{\theta_K}{n''}{w''}$,
and so 
$\mathbf{O}_{C_1;C_2}^{\sigma}(\theta) = \tau'' = \mathbf{O}_{C_2}^{\mathbf{O}_{C_1}^\sigma(\pi_L(\theta)) }(\pi_R(\theta))$.

\end{itemize}
\qed \end{proof}

\begin{lemma}[Simple sequencing for scores] \label{lemma:sc-seq}
If $C_1 \neq  C_1'; C_2'$ 
then 
$\mathbf{SC}_{C_1;C_2}^{\sigma}(\theta) = \mathbf{SC}_{C_1}^\sigma (\pi_L(\theta)) \cdot \mathbf{SC}_{C_2}^{\mathbf{O}_{C_1}^\sigma(\pi_L(\theta)) }(\pi_R(\theta))$
\end{lemma}
\begin{proof}
If $\sigma= \failure$ or $\sigma = \diverge$, the property holds trivially, so let us assume $\sigma \notin \{\failure, \diverge\}$.
We need to consider three cases:
\begin{itemize}
\item If $\mathbf{O}_{C_1}^\sigma(\pi_L(\theta)) = \sigma' \notin \{ \failure, \diverge \}$, then
$\config{\pi_L(\theta)}{C_1}{[]}{\sigma}{\pi_R(\theta) \mathrel{::} \theta_K}{0}{1} \vdash^{*}_{\mathtt{min}} 
\config{\theta'}{\modownarrow}{[]}{\sigma'}{\pi_R(\theta)\mathrel{::} \theta_K}{n}{w}$
and $\mathbf{SC}_{C_1}^\sigma (\pi_L(\theta)) = w$.

By Lemma~\ref{lemma:context-simple},
$\config{\theta}{C_1;C_2}{[]}{\sigma}{\theta_K}{0}{1} \vdash^{*}
\config{\pi_R(\theta)}{C_2}{[]}{\sigma'}{ \theta_K}{n+2}{w}$.

Now, fix a $k \geq 0$.
\begin{itemize}
\item \sloppy If $\config{\pi_R(\theta)}{C_2}{[]}{\sigma'}{ \theta_K}{0}{1} \vdash^{*}
\config{\theta''}{C_2'}{K}{\sigma''}{\theta'_K}{k}{w'}$,
then
$\mathbf{SC}_{C_2}^{\mathbf{O}_{C_1}^\sigma(\pi_L(\theta)) }(\pi_R(\theta), k) = w'$.
By Lemma~\ref{lemma:change-n-w}, 
$\config{\pi_R(\theta)}{C_2}{[]}{\sigma'}{ \theta_K}{n+2}{w} \vdash^{*}
\config{\theta''}{C_2'}{K}{\sigma''}{\theta'_K}{n+2 +k}{w w'}$,
which implies $\config{\theta}{C_1;C_2}{[]}{\sigma}{\theta_K}{0}{1} \vdash^{*}
\config{\theta''}{C_2'}{K}{\sigma''}{\theta'_K}{n+2 +k}{w w'}$,
and so $\mathbf{SC}_{C_1;C_2}^{\sigma}(\theta, n+2 +k) = ww' = 
\mathbf{SC}_{C_1}^\sigma (\pi_L(\theta)) \mathbf{SC}_{C_2}^{\mathbf{O}_{C_1}^\sigma(\pi_L(\theta)) }(\pi_R(\theta), k)$.

\item If there is no configuration $\config{\theta''}{C_2'}{K}{\sigma''}{\theta'_K}{k}{w'}$ such that
$\config{\pi_R(\theta)}{C_2}{[]}{\sigma'}{ \theta_K}{0}{1} \vdash^{*}
\config{\theta''}{C_2'}{K}{\sigma''}{\theta'_K}{k}{w'}$,
then 
$\mathbf{SC}_{C_2}^{\mathbf{O}_{C_1}^\sigma(\pi_L(\theta)) }(\pi_R(\theta), k) = 0$.
If we had 
$\config{\theta}{C_1;C_2}{[]}{\sigma}{\theta_K}{0}{1} \vdash^{*}
\config{\theta''}{C_2'}{K}{\sigma''}{\theta'_K}{n+2 +k}{w w'}$, then,
by determinacy of reduction, 
$\config{\pi_R(\theta)}{C_2}{[]}{\sigma'}{ \theta_K}{n+2}{w} \vdash^{*}
\config{\theta''}{C_2'}{K}{\sigma''}{\theta'_K}{n+2 +k}{w w'}$.
By Lemma~\ref{lemma:change-n-w} and Lemma~\ref{lemma:w-greater-0} (which ensures $w>0$), 
$\config{\pi_R(\theta)}{C_2}{[]}{\sigma'}{ \theta_K}{0}{1} \vdash^{*}
\config{\theta''}{C_2'}{K}{\sigma''}{\theta'_K}{k}{w'}$,
which contradicts the assumption.
Hence, there is no configuration $\config{\theta''}{C_2'}{K}{\sigma''}{\theta'_K}{n+2 +k}{w w'}$
such that $\config{\theta}{C_1;C_2}{[]}{\sigma}{\theta_K}{0}{1} \vdash^{*}
\config{\theta''}{C_2'}{K}{\sigma''}{\theta'_K}{n+2 +k}{w w'}$,
and so $\mathbf{SC}_{C_1;C_2}^{\sigma}(\theta, n+2 +k) =0$.
\end{itemize}

In either case,
$\mathbf{SC}_{C_1;C_2}^{\sigma}(\theta, n+2+k) = 
\mathbf{SC}_{C_1}^\sigma (\pi_L(\theta)) \cdot \mathbf{SC}_{C_2}^{\mathbf{O}_{C_1}^\sigma(\pi_L(\theta)) }(\pi_R(\theta), k)$ 
for all $k \geq 0$. Thus, we have
\begin{eqnarray*}
\mathbf{SC}_{C_1;C_2}^{\sigma}(\theta) &=& \lim_{n \rightarrow \infty} \mathbf{SC}_{C_1;C_2}^{\sigma}(\theta, n)\\
&=& \lim_{k \rightarrow \infty} \mathbf{SC}_{C_1;C_2}^{\sigma}(\theta, n+2+k)\\
&=&  \lim_{k \rightarrow \infty}
\mathbf{SC}_{C_1}^\sigma (\pi_L(\theta)) \cdot \mathbf{SC}_{C_2}^{\mathbf{O}_{C_1}^\sigma(\pi_L(\theta)) }(\pi_R(\theta), k)\\
&=&\mathbf{SC}_{C_1}^\sigma (\pi_L(\theta)) \lim_{k \rightarrow \infty}
\mathbf{SC}_{C_2}^{\mathbf{O}_{C_1}^\sigma(\pi_L(\theta)) }(\pi_R(\theta), k)\\
&=&\mathbf{SC}_{C_1}^\sigma (\pi_L(\theta))  \mathbf{SC}_{C_2}^{\mathbf{O}_{C_1}^\sigma(\pi_L(\theta)) }(\pi_R(\theta))
\end{eqnarray*}

\item If $\mathbf{O}_{C_1}^\sigma(\pi_L(\theta)) = \failure$, then 
$\mathbf{SC}_{C_2}^{\mathbf{O}_{C_1}^\sigma(\pi_L(\theta)) }(\pi_R(\theta)) = 0$, 
so $RHS = 0$.
Moreover, we have
$\config{\pi_L(\theta)}{C_1}{[]}{\sigma}{\pi_R(\theta) \mathrel{::} \theta_K}{0}{1} \vdash^{*}
\config{\theta'}{C'}{K}{\tau}{\theta_K'}{n }{w}\nvdash$. 
If $\tau = \failure$, then $C' = \failure$ and $K = []$ (as the last rule applied must have been (condition-false)),
so by Lemma~\ref{lemma:seq-add-k-fail}, $\config{\theta}{C_1;C_2}{[]}{\sigma}{\theta_K}{0}{1} \vdash^{*}
\config{\theta'}{\modownarrow}{[]}{\failure}{\theta_K'}{n+1}{w}$.
Hence, $\mathbf{SC}_{C_1;C_2}^{\sigma}(\theta, n') = 0$ for all $n' > n+1$, and so
$\mathbf{SC}_{C_1;C_2}^{\sigma}(\theta) = 0$.

\item If $\mathbf{O}_{C_1}^\sigma(\pi_L(\theta)) = \diverge$, then 
$RHS =  \mathbf{SC}_{C_1}^\sigma (\pi_L(\theta)) $ and for all $k$, we
have $\config{\pi_L(\theta)}{C_1}{[]}{\sigma}{\pi_R(\theta) \mathrel{::} \theta_K}{0}{1} \vdash^{*}
\config{\theta'}{C''_1}{K}{\sigma'}{\pi_R(\theta)\mathrel{::} \theta'_K}{k}{w}$,
where $(C''_1, K) \neq (\modownarrow, K)$ and $\sigma' \neq \failure$.
Fix $k \geq 0$. We have  $\mathbf{SC}_{C_1}^\sigma (\pi_L(\theta), k) = w$
and by Lemma~\ref{lemma:add-seq},
$\config{\theta}{C_1;C_2}{[]}{\sigma}{\theta_K}{0}{1} \vdash^{*}
\config{\theta'}{C''_1}{K @ [C_2]}{\sigma'}{\theta'_K}{k+1}{w}$, which implies
$\mathbf{SC}_{C_1;C_2}^\sigma (\theta, k+1) = w$. Hence,
$\mathbf{SC}_{C_1;C_2}^\sigma (\theta, k+1) =
\mathbf{SC}_{C_1}^\sigma (\pi_L(\theta), k) $.

Thus,
\begin{eqnarray*}
\mathbf{SC}_{C_1;C_2}^{\sigma}(\theta) &=& \lim_{n \rightarrow \infty} \mathbf{SC}_{C_1;C_2}^{\sigma}(\theta, n)\\
&=& \lim_{k \rightarrow \infty} \mathbf{SC}_{C_1;C_2}^{\sigma}(\theta, k+1)\\
&=&  \lim_{k \rightarrow \infty}
\mathbf{SC}_{C_1}^\sigma (\pi_L(\theta), k) \\
&=&\mathbf{SC}_{C_1}^\sigma (\pi_L(\theta)) 
\end{eqnarray*}
\noindent as required.
\end{itemize}
\qed \end{proof}

\begin{restate}{Proposition~\ref{lemma:o-sc-seq}}
\sloppy If $C_1 \neq  C_1'; C_2'$, then $\mathbf{O}_{C_1;C_2}^{\sigma}(\theta) =  \mathbf{O}_{C_2}^{\mathbf{O}_{C_1}^\sigma(\pi_L(\theta)) }(\pi_R(\theta))$
and $\mathbf{SC}_{C_1;C_2}^{\sigma}(\theta) = \mathbf{SC}_{C_1}^\sigma (\pi_L(\theta)) \cdot \mathbf{SC}_{C_2}^{\mathbf{O}_{C_1}^\sigma(\pi_L(\theta)) }(\pi_R(\theta))$
\end{restate}
\begin{proof}
This is a combination of Lemma~\ref{lemma:o-seq} and Lemma~\ref{lemma:sc-seq}.
\qed \end{proof}

Proposition~\ref{lemma:o-sc-seq} is not applicable when $C_1$ is not a sequence of statements,
as we cannot know what part of the entropy $\theta$ will be used in the evaluation of which expression without knowing the length of the statement list in $C_1$. 
However, the above result can be generalised using \emph{finite shuffling functions}, as defined by~\cite{DBLP:journals/pacmpl/WandCGC18}.

\begin{definition}[\cite{DBLP:journals/pacmpl/WandCGC18}]
\begin{itemize}
\item A \emph{path} is a function $[d_1, \dots, d_n] \colon \mathbb{S} -> \mathbb{S}$ parametrised by a list of directions
$d_1, \dots, d_n \in \{L, R\}$, such that  $[d_1, \dots, d_n](\theta) = (\pi_{d_1} \circ \dots \circ \pi_{d_n})(\theta)$.
\item A \emph{finite shuffling function} (FSF) is a function $\phi \colon \mathbb{S} -> \mathbb{S}$ such that either $\phi$ is a path
or $\phi(\theta) = \phi_1(\theta) \mathrel{::} \phi_2(\theta)$, where $\phi_1$ and $\phi_2$ are FSFs.
\item A sequence of paths is \emph{non-duplicating} if no path in the sequence is a suffix of another path.
\item A FSF $\phi$ is non-duplicating if the sequence of all paths appearing in its definition is non-duplicating.
%A FSF $\phi$ is non-duplicating if $\phi(\theta)$ is non-duplicating for each $\theta$.
%%\marginpar{Can you make this more precise?}
\end{itemize}
\end{definition}
The following key result shows that entropy rearrangements via FSFs have no effect under integration:
\begin{lemma}[\cite{DBLP:journals/pacmpl/WandCGC18}, Th.~7.6] \label{lemma:fsf-meas-pres}
Any non-duplicating FSF $\phi$ is \emph{measure-preserving}, i.e., for any measurable\footnote{The result in \cite{DBLP:journals/pacmpl/WandCGC18} considers $g$ with co-domain $[0, \infty)$ rather than $\extposreals$. It is however, not difficult to check that their result extends to the latter case.} $g \colon \mathbb{S} ->  \extposreals$: %\mathbb{R}_{+}$, we have:
$$
\int g(\phi(\theta)) \, \mu(d\theta) \ = \ \int g(\theta) \, \mu(d\theta).
$$
\end{lemma}
%TODO: check if this applies to functions admitting infinite values?
%Answer: yes, everything's OK, the original definitions in Culpepper17 also work for functions admitting infinite values
%%
We now have everything in place to define a version of Proposition~\ref{lemma:o-sc-seq} for an arbitrary split of a sequencing statement:
\begin{proposition}[Sequencing for final states]\label{lemma:rearrange-entropy}
If $C = C_1; C_2$, there exists
%a measure-preserving function $\phi$
a non-duplicating FSF $\psi$ such that:
$$\mathbf{O}_C^\sigma(\theta) = \mathbf{O}_{C_2}^{\tau}(\pi_R(\psi(\theta))) 
\quad \mbox{and} \quad
\mathbf{SC}_C^\sigma(\theta) = 
  \mathbf{SC}_{C_1}^\sigma (\pi_L(\psi(\theta))) \cdot 
  \mathbf{SC}_{C_2}^{\tau}(\pi_R(\psi(\theta)))
  $$
with $\tau$ denoting $\mathbf{O}_{C_1}^\sigma(\pi_L(\psi(\theta)))$.
\end{proposition}
\begin{proof}%[of Proposition~\ref{lemma:rearrange-entropy}]
By induction on the structure of $C$. 

\begin{itemize}
\item Base case: $C_1 \neq C_1'; C_1''$: the equality holds trivially for $\psi = \mathit{Id}$ by Lemma~\ref{lemma:o-seq}.
\item Induction step: If $C_1$ is a sequence of statements, then $C_1 = C_1'; C_1''$ for some $C_1'$ such that
$C_1' \neq \hat{C}_1' \hat{C}_1''$.

We have:

\begin{eqnarray*}
\mathbf{O}_{C_1'; C_1'';C_2}^\sigma(\theta) \text{\tiny   (by Lemma~\ref{lemma:o-seq})}&=& 
\mathbf{O}_{C_1'';C_2}^{ \mathbf{O}_{C_1'}^\sigma(\pi_L(\theta))}(\pi_R(\theta))\\
\text{\tiny (by induction hypothesis)} &=&
\mathbf{O}_{C_2}^{ \mathbf{O}_{C_1''}^ {\mathbf{O}_{C_1'}^\sigma(\pi_L(\theta))}(\pi_L(\psi(\pi_R(\theta))))}(\pi_R(\psi(\pi_R(\theta))))\\
\end{eqnarray*}
\noindent for some non-duplicating FSF $\psi$.

Thus, if $\theta = \theta_1 \mathrel{::} \theta_2$, then
$$\mathbf{O}_{C_1'; C_1'';C_2}^\sigma( \theta_1 \mathrel{::} \theta_2)=
\mathbf{O}_{C_2}^{ \mathbf{O}_{C_1''}^ {\mathbf{O}_{C_1'}^\sigma(\theta_1)}(\pi_L(\psi(\theta_2)))}(\pi_R(\psi(\theta_2)))$$

Now, take $\hat\psi$ such that $\hat\psi(\theta_1 \mathrel{::} \theta_2) = 
(\theta_1 \mathrel{::} \pi_L(\psi(\theta_2))) \mathrel{::} \pi_R(\psi(\theta_2))$.

Then 
\begin{eqnarray*}
\mathbf{O}_{C_2}^{ \mathbf{O}_{C_1'; C_1''}^\sigma(\pi_L(\hat{\psi}(\theta_1 \mathrel{::} \theta_2)))}(\pi_R(\hat{\psi}(\theta_1 \mathrel{::} \theta_2)))&=&
\mathbf{O}_{C_2}^{ \mathbf{O}_{C_1'; C_1''}^\sigma(\theta_1 \mathrel{::} \pi_L(\psi(\theta_2)))}(\pi_R(\psi(\theta_2)))\\
\text{\tiny (by Lemma~\ref{lemma:o-seq})}&=&
\mathbf{O}_{C_2}^{ \mathbf{O}_{C_1''}^ {\mathbf{O}_{C_1'}^\sigma(\theta_1)} \pi_L(\psi(\theta_2))}(\pi_R(\psi(\theta_2)))\\
&=&\mathbf{O}_{C_1'; C_1'';C_2}^\sigma(\theta_1 \mathrel{::} \theta_2)
\end{eqnarray*}
\noindent as required. 

For $\mathbf{SC}$, we have:
\begin{eqnarray*}
  \mathbf{SC}_{C_1'; C_1''; C_2}^\sigma(\theta) \text{\tiny (by Lemma \ref{lemma:sc-seq})}&=&
  \mathbf{SC}_{C_1'}(\pi_L(\theta)) %\mathbf{SC}_{}^\sigma(\theta) 
    \mathbf{SC}_{C_1''; C_2}^{ \mathbf{O}_{C'_1}^\sigma(\pi_L(\theta))}(\pi_R(\theta)) \\
    \text{\tiny (by induction hypothesis)}&=&
  \mathbf{SC}_{C_1'}(\pi_L(\theta)) \mathbf{SC}_{C_1''}^{ \mathbf{O}_{C'_1}^\sigma(\pi_L(\theta))}(\pi_L(\psi(\pi_R(\theta))))
  \\ && \qquad
    \mathbf{SC}_{C_2}^{ \mathbf{O}_{C''_1}^{\mathbf{O}_{C'_1}^\sigma(\pi_L(\theta))}(\pi_L(\psi(\pi_R(\theta))))}(\pi_R(\psi(\pi_R(\theta)))) \\
\end{eqnarray*}

\noindent for the same $\psi$. Thus, for $\hat{\psi}$ defined above, we have:

\begin{eqnarray*}
&& \mathbf{SC}_{C_1'; C_1''}(\pi_L(\hat{\psi}(\theta_1 \mathrel{::} \theta_2)))
    \mathbf{SC}_{C_2}^{ \mathbf{O}_{C_1'; C_1''}^\sigma(\pi_L(\hat{\psi}(\theta_1 \mathrel{::} \theta_2)))}
    (\pi_R(\hat{\psi}(\theta_1 \mathrel{::} \theta_2))) \\
    &=&  \mathbf{SC}_{C_1'; C_1''}(\theta_1 \mathrel{::} \pi_L(\psi(\theta_2)))
    \mathbf{SC}_{C_2}^{ \mathbf{O}_{C_1'; C_1''}^\sigma(\theta_1 \mathrel{::} \pi_L(\psi(\theta_2)))}
    \pi_R(\psi(\theta_2)) \\
    %\text{ \tiny (By Lemmas \ref{lemma:o-seq} and \ref{lemma:sc-seq})} &=& 
    \text{\tiny (*)} &=& 
 \mathbf{SC}_{C_1'}(\theta_1) \mathbf{SC}_{C_1''}^{\mathbf{O}_{C_1'}^\sigma(\theta_1) }( \pi_L(\psi(\theta_2)))
%\mathbf{SC}_{C_2}^{ \mathbf{O}_{C_1'; C_1''}^\sigma(\theta_1 \mathrel{::} \pi_L(\psi(\theta_2)))}
   % \pi_R(\psi(\theta_2))
\mathbf{SC}_{C_2}^{ \mathbf{O}_{C''_1}^{\mathbf{O}_{C'_1}^\sigma(\theta_1)}(\pi_L(\psi(\theta_2)))}(\pi_R(\psi(\theta_2))) \\
&=& \mathbf{SC}_{C_1'; C_1''; C_2}^\sigma(\theta_1 \mathrel{::} \theta_2)
\end{eqnarray*}

\noindent as required, where the equality (*) follows from Lemmas \ref{lemma:o-seq} and \ref{lemma:sc-seq}.

Now we only need to show that $\hat\psi$ is a non-duplicating FSF. 

First, let us show that $\hat\psi$ is indeed a FSF. To this end, we need to show that if $\psi$ is a FSF, then
$\psi'(\theta) =  \psi(\pi_R(\theta))$ is also a FSF. We prove this by induction on the structure of $\psi$:

\begin{itemize}
\item Base case: if $\psi$ is a path $[d_1, \dots, d_n]$, then $\psi \circ \pi_R$ is 
the path $[d_1, \dots, d_n, R]$, so it is a FSF.
\item Induction step: Suppose that $\psi(\theta) = \psi_1(\theta) \mathrel{::} \psi_2(\theta)$
and that $\psi_1 \circ \pi_R$ and  $\psi_2 \circ \pi_R$ are FSFs.  Then we
have $\psi(\pi_R(\theta)) = \psi_1(\pi_R(\theta)) \mathrel{::} \psi_2(\pi_R(\theta)) = 
 (\psi_1 \circ \pi_R)(\theta) \mathrel{::} (\psi_2 \circ \pi_R)(\theta)$,
so $\psi \circ \pi_R$ is a FSF by definition.
\end{itemize}

Now, we show that $\psi''(\theta) =  \pi_L(\psi(\pi_R(\theta))) = \pi_L(\psi'(\theta))$ is a FSF:
if $\psi'$ is a path $[d_1, \dots, d_n]$, then $\psi''$ is a path $[L, d_1, \dots, d_n]$, 
and if $\psi' = \psi'_1 \mathrel{::} \psi'_2$, then $\pi_L(\psi'(\theta))
= \pi_L(\psi'_1(\theta) \mathrel{::} \psi'_2(\theta)) = \psi'_1(\theta)$. Similarly, we can show that $ \pi_R(\psi(\pi_R(\theta)))$ is a FSF.
Hence, $\hat\psi$ is a FSF by definition.

Finally, we need to show that $\hat\psi$ is non-duplicating.

We can show by a simple induction that for any $\psi$, the set of paths
$\mathcal{P}_{\psi \circ \pi_R}$
in $\psi \circ \pi_R$ is $\{ pR\ |\ p \in \mathcal{P}_{\psi} \}$,
where $\mathcal{P}_{\psi}$ is the set of paths in $\psi$
and juxtaposition denotes concatenation.

If $\psi$ is a path $p$, then $\pi_L \circ \psi \circ \pi_R$ and $\pi_R \circ \psi \circ \pi_R$
are paths $LpR$ and $RpR$. Hence, the set of paths
in $\hat{\psi}$ is $\{[L], LpR, RpR \}$. It is instantly clear that no path is a suffix of another, so $\hat{\psi}$ is non-duplicating.

If $\psi(\theta)  = \psi_1(\theta) \mathrel{::} \psi_2(\theta)$, then 
$(\pi_L \circ \psi \circ \pi_R)(\theta) = \pi_L(\psi_1(\pi_R(\theta)) \mathrel{::} \psi_2(\pi_R(\theta)))
= \psi_1(\pi_R(\theta))$,
so the set of paths in $\pi_L \circ \psi \circ \pi_R$ is $\{pR\ |\ p \in \mathcal{P}_{\psi_1} \}$, where
$\mathcal{P}_{\psi_1} $ is the set of paths in $\psi_1$.
Similarly, the set of paths in $\pi_R \circ \psi \circ \pi_R$ is $\{pR\ |\ p \in \mathcal{P}_{\psi_2} \}$, where
$\mathcal{P}_{\psi_2} $ is the set of paths in $\psi_2$.
Since $\mathcal{P}_{\psi} = \mathcal{P}_{\psi_1}  \cup \mathcal{P}_{\psi_2} $,
the set of paths in the entire definition of $\hat{\psi}$ is 
$\{ [L] \} \cup \{pR\ |\ p \in \mathcal{P}_{\psi} \}$. It is clear that $[L]$ is not a suffix of any
path of the form $pR$ (as all such paths end with $R$). Moreover, if there were 
paths $p_1, p_2 \in \mathcal{P}_{\psi}$ such that $p_1R$ was a suffix of $p_2R$, then
$p_1$ would be a suffix of $p_2$, which would contradict the assumption.

Hence, $\hat{\psi}$ is non-duplicating, which ends the proof.
\end{itemize}
\qed \end{proof}

\subsection{Approximating while-loops}
To simplify reasoning about $\mathtt{while}$-loops, it is useful---and common in program semantics---to consider finite approximations of loops in which the maximal number of iterations is bounded. 
To that end, we define the $n$-th unfolding of a guarded loop inductively as follows:
\begin{eqnarray*}
\mathtt{while}^0(\phi)\{C\} &\ = \ & \mathtt{diverge} \\[1ex]
\mathtt{while}^{n+1}(\phi)\{C\} &=&\mathtt{if}(\phi)\{C; \mathtt{while}^n(\phi)\{C\} \}.
\end{eqnarray*}
In the limit, bounded $\mathtt{while}$-loops behave as standard $\mathtt{while}$-loops. 
We use this result to define the evaluation of measurable function $f$ on successful termination states of a $\mathtt{while}$-loop, scaled by its score as a limit of approximations.
As we are interested in $f$ on proper states, we use $\hat{f}$ rather than $f$.
\begin{proposition} \label{lemma:sup-while-o-sc}
Let loop $C = \mathtt{while}(\phi)\{C'\}$ and $C^n = \mathtt{while}^n(\phi)\{C'\}$ its $n$-th approximation. Then:
$$
\hat{f}(\mathbf{O}_{C}^\sigma(\theta)) \cdot \mathbf{SC}_{C}^\sigma(\theta)
\ = \ 
\sup_n \, \hat{f}( \mathbf{O}_{C^n}^\sigma(\theta)) \cdot \mathbf{SC}_{C^n}^\sigma(\theta).
$$
\end{proposition}
%\begin{proof}
%The proof is based on standard reasoning about operational semantics. 
%Using a simulation relation, we show that $\mathtt{while}(\phi)\{C\}$ with initial state $\sigma$ evaluates to some final configuration iff $\mathtt{while}^{n}(\phi)\{C\}$ evaluates
%to the same configuration for some $n$; details in Appendix \ref{app:proofs-sem-fun}.
% \qed \end{proof}
%%
The following monotonicity property is relevant later when proving the relationship between the operational semantics of \ccpgcl and its denotational semantics.
As before let $C^n = \mathtt{while}^n(\phi)\{C'\}$.
\begin{proposition} \label{lemma:o-sc-increasing}
If $n \geq k$ and $\hat{f}(\failure) = \hat{f}(\diverge) = 0$, then  
$\hat{f}(\mathbf{O}_{C^n}^\sigma(\theta)) \cdot \mathbf{SC}_{C^n}^\sigma(\theta)
 \geq \hat{f}(\mathbf{O}_{C^k}^\sigma(\theta)) \cdot \mathbf{SC}_{C^k}^\sigma(\theta)$.
\end{proposition}
%\begin{proof}
%The proof uses similar principles to the proof of Proposition~\ref{lemma:sup-while-o-sc}, details in Appendix \ref{app:proofs-sem-fun}.
% \qed \end{proof}
%%
Similarly, we want to show that the sequence
$\hat{f}(\mathbf{O}_{C^n}^\sigma(\theta)) \cdot \mathbf{SC}_{C^n}^\sigma(\theta)$
approximates
$\check{f}(\mathbf{O}_{C}^\sigma(\theta)) \cdot \mathbf{SC}_{C}^\sigma(\theta)$.
%% \marginpar{Why is this relevant? Can we add a sentence on that.}
This result allows us to express the anticipated value of the function $\hat{f}$ for a given fixed entropy
as a limit of approximations,
and by integrating both sides with respect to the measure on entropies we get that the expected value
of $\hat{f}$ can also be expressed as a limit of approximations. We will use this result in the proof of
Theorem~\ref{thm:wlp-op-equiv}.
Recall that $\check{f}(\tau) = 1$ for $\tau = \ \diverge$.
\begin{proposition} \label{lemma:inf-while-o-sc}
%If $\hat{f}(\failure) = \hat{f}(\diverge) = 0$, then
Let loop $C = \mathtt{while}(\phi)\{C'\}$ and $C^n = \mathtt{while}(\phi)\{C'\}$ its $n$-th approximation.
Take a function $f \leq 1$. Then
$$
\check{f}(\mathbf{O}_{C}^\sigma(\theta)) \cdot \mathbf{SC}_{C}^\sigma(\theta)
\ = \ 
\inf_n \, \check{f}( \mathbf{O}_{C^n}^\sigma(\theta)) \cdot \mathbf{SC}_{C^n}^\sigma(\theta).
$$
\end{proposition}
%\begin{proof}
%By analysis of the operational semantics, details in Appendix~\ref{app:proofs-sem-fun}.
%\qed \end{proof}

\begin{proposition} \label{lemma:o-sc-decreasing}
If $n \geq k$ and $f \leq 1$, then  
$$
\check{f}( \mathbf{O}_{C^n}^\sigma(\theta)) \cdot \mathbf{SC}_{C^n}^\sigma(\theta)
\ \leq \ 
\check{f}( \mathbf{O}_{C^k}^\sigma(\theta)) \cdot \mathbf{SC}_{C^k}^\sigma(\theta).
$$
\end{proposition}
%\begin{proof}
%Follows from the fact that $\mathbf{SC}_{C^n}^\sigma(\theta)$ is decreasing in $n$,
%$\mathbf{O}_{C^n}^\sigma(\theta)$ is increasing in $n$ and $\check{f}$ is antitone 
%($\check{f}(\tau) \leq \check{f}(\diverge) = 1$ for all $\tau \geq \ \diverge$); cf.\ Appendix~\ref{app:proofs-sem-fun}.
% \qed \end{proof}

The rest of this section is the proof of Propositions~\ref{lemma:sup-while-o-sc}, \ref{lemma:o-sc-increasing}, 
\ref{lemma:inf-while-o-sc} and \ref{lemma:o-sc-decreasing}, which will be needed to prove
the case of $\mathtt{while}$-loops in Theorem~\ref{thm:wp-op-equiv} and Theorem~\ref{thm:wlp-op-equiv}.
%
%This section contains proofs of properties of bounded approximations of while-loops, %which are needed to prove TODO
%which are needed
%to prove the case of $\mathtt{while}$-loops in Theorem~\ref{thm:wp-op-equiv} and Theorem~\ref{thm:wlp-op-equiv}.
%
\newcommand{\reln}[1]{\sim^{#1}}
The first key fact that we want to show is that for non-diverging executions, a bounded while-loop of the form
$\mathtt{while}^n(\phi)\{C\}$ behaves just like $\mathtt{while}(\phi)\{C\}$ for
a sufficiently large $n$. We formalise and prove it using two auxiliary relations on configurations.

\subsubsection{Replacing $\mathtt{while}(\phi)\{C\}$ with $\mathtt{while}^n(\phi)\{C\}$}

We first prove that in all non-divering configurations, 
if the expression is of the form $\mathtt{while}(\phi)\{C\}$, we can replace it 
with $\mathtt{while}^n(\phi)\{C\}$ for a large enough $n$, without
changing the final configuration reached after reduction is completed. To this end, we first define an indexed
relation $(\reln{n})$ on configurations. We begin with auxiliary relations $C \reln{n} C'$ and  $K \reln{n} K'$,
defined inductively as follows:

\inference
{ }
{ C \reln{0} C'}

\ 

\noindent For $n > 0$:

\ 

\inference
{ }
{C \reln{n} C }

\ 

\inference
{ }
{\modownarrow \reln{n} \modownarrow }

\

\inference
{k \geq n}
{\mathtt{while}(\phi)\ \{ C \}  \reln{n} 
\mathtt{while}^k(\phi)\ \{ C \} }

\ 

\inference
{k \geq n}
{\mathtt{while}^k(\phi)\ \{ C \}  \reln{n} 
\mathtt{while}(\phi)\ \{ C \} }

\ 

\inference
{k \geq n \quad l\geq n}
{\mathtt{while}^k(\phi)\ \{ C \}  \reln{n} 
\mathtt{while}^l(\phi)\ \{ C \} }

\ 

\inference
{C_2 \reln{n} C_2'}
{C_1;C_2 \reln{n} 
C_1;C'_2 }

\ 

\inference
{\forall i \in 1..n \quad C_i \reln{n} C'_i}
{[C_1, \dots, C_n] \reln{n}  [C_1', \dots, C_n']}

\ 

We then naturally extend the definition to configurations:

\inference
{}
{\config{\theta}{C}{K}{\sigma}{\theta_K}{m}{w} \reln{0}  \config{\theta'}{C'}{K'}{\sigma'}{\theta'_K}{m'}{w'} }

\noindent For $n > 0$:
\ 

\inference
{C \reln{n} C' \quad K \reln{n} K'}
{\config{\theta}{C}{K}{\sigma}{\theta_K}{m}{w} \reln{n}  \config{\theta}{C'}{K'}{\sigma}{\theta_K}{m}{w} }

We can immediately check that if two configurations are related by $(\reln{n})$ for some $n>0$, then if we perform one step of reductions
on both of them, the resulting configurations are guaranteed to be related at least by $(\reln{n-1})$.

\begin{lemma}
$\reln{n}$ is a stratified bisumulation---that is, 
$\config{\theta}{C}{K}{\sigma}{\theta_K}{m}{w} \reln{0}  \config{\theta'}{C'}{K'}{\sigma'}{\theta'_K}{m'}{w'} $ and for $n>0$:
\begin{itemize}
\item if $\config{\theta}{C}{K}{\sigma}{\theta_K}{m}{w} \reln{n}  \config{\theta}{C'}{K'}{\sigma}{\theta_K}{m}{w}$
and 

\noindent $\config{\theta}{C}{K}{\sigma}{\theta_K}{m}{w} \vdash \config{\theta''}{C''}{K''}{\sigma''}{\theta''_K}{m+1}{w''}$, then
$\config{\theta}{C'}{K'}{\sigma}{\theta_K}{m}{w}  \vdash \config{\theta''}{C'''}{K'''}{\sigma''}{\theta''_K}{m+1}{w''}$
and $\config{\theta''}{C''}{K''}{\sigma''}{\theta''_K}{m+1}{w''} \reln{n-1}
\config{\theta''}{C'''}{K'''}{\sigma''}{\theta''_K}{m+1}{w''}$

\item 
if $\config{\theta}{C}{K}{\sigma}{\theta_K}{m}{w} \reln{n}  \config{\theta}{C'}{K'}{\sigma}{\theta_K}{m}{w}$
and 

\noindent $\config{\theta}{C'}{K'}{\sigma}{\theta_K}{m}{w}  \vdash \config{\theta''}{C'''}{K'''}{\sigma''}{\theta''_K}{m+1}{w''}$, then
$\config{\theta}{C}{K}{\sigma}{\theta_K}{m}{w} \vdash \config{\theta''}{C''}{K''}{\sigma''}{\theta''_K}{m+1}{w''}$
and $\config{\theta''}{C''}{K''}{\sigma''}{\theta''_K}{m+1}{w''} \reln{n-1}
\config{\theta''}{C'''}{K'''}{\sigma''}{\theta''_K}{m+1}{w''}$
\end{itemize}
\end{lemma}
\begin{proof}
By inspection.
\qed \end{proof}

This result naturally generalises to multi-step reduction.

\begin{corollary} \label{corr:bisim-n-steps}
If $\config{\theta}{C}{K}{\sigma}{\theta_K}{m}{w} \reln{n}  \config{\theta}{C'}{K'}{\sigma}{\theta_K}{m}{w}$
and $\config{\theta}{C}{K}{\sigma}{\theta_K}{m}{w} \vdash^{*} \config{\theta''}{C''}{K''}{\sigma''}{\theta''_K}{m+n'}{w''}$
and $n' < n$
then $\config{\theta}{C'}{K'}{\sigma}{\theta_K}{m}{w}  \vdash^{*} \config{\theta''}{C'''}{K'''}{\sigma''}{\theta''_K}{m+n'}{w''}$
and $\config{\theta''}{C''}{K''}{\sigma''}{\theta''_K}{m+n'}{w''} \reln{n-n'} \config{\theta''}{C'''}{K'''}{\sigma''}{\theta''_K}{m+n'}{w''}$
(and vice versa).
\end{corollary}

This leads us to the desired result for terminating runs.

\begin{lemma} \label{lemma:while-approx-os-right}
If $\config{\theta}{\mathtt{while}(\phi)\{C\}}{[]}{\sigma}{\theta_K}{n}{w} \vdash^{*}
\config{\theta'}{\modownarrow}{[]}{\sigma'}{\theta_K}{n+n'}{w'}$,
then there exists $k$ such that
$\config{\theta}{\mathtt{while}^k(\phi)\{C\}}{[]}{\sigma}{\theta_K}{n}{w} \vdash^{*}
\config{\theta'}{\modownarrow}{[]}{\sigma'}{\theta_K}{n+n'}{w'}$
\end{lemma}
\begin{proof}
%By induction on $n'$
Take $k = n'+1$. We clearly have
$\mathtt{while}(\phi)\{C\} \reln{n'+1} \mathtt{while}^{n'+1}(\phi)\{C\}$,
and so 
$\config{\theta}{\mathtt{while}(\phi)\{C\}}{[]}{\sigma}{\theta_K}{n}{w} \reln{n'+1}
\config{\theta}{\mathtt{while}^{n'+1}(\phi)\{C\}}{[]}{\sigma}{\theta_K}{n}{w}$.
By Corollary~\ref{corr:bisim-n-steps}, 
$\config{\theta}{\mathtt{while}^{n'+1}(\phi)\{C\}}{[]}{\sigma}{\theta_K}{n}{w} \vdash^{*}
\config{\theta'}{C'}{K'}{\sigma'}{\theta_K}{n+n'}{w'}$, where $\modownarrow \reln{1} C'$
and $[] \reln{1} K'$, which implies $C'= \modownarrow$ and $K' = []$. Thus, the statement always
holds for $k = n'+1$.
\qed \end{proof}

This result leads to the following statement about the $\mathbf{O}_{C}^\sigma$ and 
$\mathbf{SC}_{C}^\sigma$ functions:

\begin{lemma} \label{lemma:while-o-sc-sup}
\sloppy For each $\phi$, $C$, $\sigma$, $\theta$,
such that $\mathbf{O}_{\mathtt{while}(\phi)\{C\}}^\sigma(\theta) \in \statespace$
there is a $k$ such that
$\mathbf{O}_{\mathtt{while}(\phi)\{C\}}^\sigma(\theta)
= \mathbf{O}_{\mathtt{while}^k(\phi)\{C\}}^\sigma(\theta)$
and $ \mathbf{SC}_{\mathtt{while}(\phi)\{C\}}^\sigma(\theta) =  \mathbf{SC}_{\mathtt{while}^k(\phi)\{C\}}^\sigma(\theta)$
\end{lemma}
\begin{proof}
\sloppy If $\mathbf{O}_{\mathtt{while}(\phi)\{C\}}^\sigma(\theta) \in \statespace$,
then by definition of $\mathbf{O}$,
$\config{\theta}{\mathtt{while}(\phi)\{C\}}{[]}{\sigma}{\theta_K}{0}{1}
\vdash^{*} \config{\theta'}{\modownarrow}{[]}{\sigma'}{\theta'_K}{n}{w}$, where $\sigma' \neq \failure$.
This implies   $\mathbf{O}_{\mathtt{while}(\phi)\{C\}}^\sigma(\theta) = \sigma'$
and $\mathbf{SC}_{\mathtt{while}(\phi)\{C\}}^\sigma(\theta) = w$.
By Lemma~\ref{lemma:while-approx-os-right},
there is a $k$ such that
$\config{\theta}{\mathtt{while}^k(\phi)\{C\}}{[]}{\sigma}{\theta_K}{0}{1}
\vdash^{*} \config{\theta'}{\modownarrow}{[]}{\sigma'}{\theta'_K}{n}{w}$.
Thus, $\mathbf{O}_{\mathtt{while}^k(\phi)\{C\}}^\sigma(\theta) = \sigma'$
and 
%$\mathbf{SCT}_{\mathtt{while}^k(\phi)\{C\}}^\sigma(\theta) = w$,
%with the latter implying
$\mathbf{SC}_{\mathtt{while}^k(\phi)\{C\}}^\sigma(\theta) = w$.
%by Lemma~\ref{lemma:sc-eq-sct}.
\qed \end{proof}

We can also show that if the evaluation of $\mathtt{while}(\phi)\{C\}$ gets stuck, 
so does the evaluation of $\mathtt{while}^k(\phi)\{C\}$ for large enough $k$.

\begin{lemma} \label{lemma:while-error-some-k}
If $\config{\theta}{\mathtt{while}(\phi)\{C\}}{[]}{\sigma}{\theta_K}{n}{w} \vdash^{*}
\config{\theta'}{C'}{K}{\sigma'}{\theta_K'}{n+n'}{w'} \nvdash$, 
then there exists $k$ such that
$\config{\theta}{\mathtt{while}^k(\phi)\{C\}}{[]}{\sigma}{\theta_K}{n}{w} \vdash^{*}
\config{\theta'}{C''}{K'}{\sigma'}{\theta_K'}{n+n'}{w'} \nvdash$.
\end{lemma}
\begin{proof}
Again, take $k = n'+1$. We have
$\mathtt{while}(\phi)\{C\} \reln{n'+1} \mathtt{while}^{n'+1}(\phi)\{C\}$,
and so 
$\config{\theta}{\mathtt{while}(\phi)\{C\}}{[]}{\sigma}{\theta_K}{n}{w} \reln{n'+1}
\config{\theta}{\mathtt{while}^{n'+1}(\phi)\{C\}}{[]}{\sigma}{\theta_K}{n}{w}$.
By Corollary~\ref{corr:bisim-n-steps}, 
$\config{\theta}{\mathtt{while}^{n'+1}(\phi)\{C\}}{[]}{\sigma}{\theta_K}{n}{w} \vdash^{*}
\config{\theta'}{C''}{K'}{\sigma'}{\theta_K}{n+n'}{w'}$, where $C' \reln{1} C''$
and $K \reln{1} K'$. By case analysis on the derivation of 
$C' \reln{1} C''$, and using the fact that $K$ and $K'$ must have the same length,
we conclude that $\config{\theta'}{C'}{K}{\sigma'}{\theta_K'}{n+n'}{w'}$
reduces if and only if $\config{\theta'}{C''}{K'}{\sigma'}{\theta_K'}{n+n'}{w'}$ reduces.
\qed \end{proof}

\subsubsection{Replacing $\mathtt{while}^n(\phi)\{C\}$ with $\mathtt{while}(\phi)\{C\}$}

We now prove the converse to the above result---that if $\mathtt{while}^n(\phi)\{C\}$ evaluates
with some entropy $\theta$, the unbounded loop $\mathtt{while}(\phi)\{C\}$ evaluates to the same
configuration. We begin with another relation $\relu$ on configurations, which effectively states
that for two configurations $\kappa_1$ and $\kappa_2$, if $\kappa_1 \relu \kappa_2$ and
$\kappa_1$ evaluates, then $\kappa_2$ is guaranteed to evaluate to the same final configuration. 
This relation is defined inductively as follows:

%We define another relation $\relu$ on statements, continuations and configurations inductively as follows:

\ 

\inference
{}
{C \relu C}

\ 

\inference
{}
{\modownarrow \relu \modownarrow}

\ 

\inference
{ }
{\mathtt{while}^k(\phi)\ \{ C \}  \relu
\mathtt{while}(\phi)\ \{ C' \} }

\ 

\inference
{k \leq l}
{\mathtt{while}^k(\phi)\ \{ C \}  \relu
\mathtt{while}^l(\phi)\ \{ C' \} }

\ 

\inference
{}
{\mathtt{diverge}  \relu
C }

\ 

\inference
{ C_2 \relu C_2'}
{C_1;C_2 \relu
C_1;C'_2 }

\ 

\inference
{\forall i \in 1..n \quad C_i \relu C'_i}
{[C_1, \dots, C_n] \relu  [C_1, \dots, C_n]}

\ 

\inference
{C \relu C' \quad K \relu K'}
{\config{\theta}{C}{K}{\sigma}{\theta_K}{m}{w} \relu  \config{\theta}{C'}{K'}{\sigma}{\theta_K}{m}{w} }

\begin{lemma} \label{lemma:relu-sim}
$\relu$ is a simulation---that is, 
if $\config{\theta}{C}{K}{\sigma}{\theta_K}{m}{w} \relu  \config{\theta}{C'}{K'}{\sigma}{\theta_K}{m}{w}$
and $\config{\theta}{C}{K}{\sigma}{\theta_K}{m}{w} \vdash \config{\theta''}{C''}{K''}{\sigma''}{\theta''_K}{m+1}{w''}$
and $C \neq \mathtt{diverge}$, then
$\config{\theta}{C'}{K'}{\sigma}{\theta_K}{m}{w}  \vdash \config{\theta''}{C'''}{K'''}{\sigma''}{\theta''_K}{m+1}{w''}$
and $\config{\theta''}{C''}{K''}{\sigma''}{\theta''_K}{m+1}{w''} \relu
\config{\theta''}{C'''}{K'''}{\sigma''}{\theta''_K}{m+1}{w''}$
\end{lemma}
\begin{proof}
By case analysis on the reduction rules.
\qed \end{proof}

\begin{corollary} \label{corr:sim-red}
If $\config{\theta}{C}{K}{\sigma}{\theta_K}{m}{w} \relu  \config{\theta}{C'}{K'}{\sigma}{\theta_K}{m}{w}$
and $\config{\theta}{C}{K}{\sigma}{\theta_K}{m}{w} \vdash^{*} \config{\theta''}{C''}{K''}{\sigma''}{\theta''_K}{m+n'}{w''}$
and $C'' \neq \mathtt{diverge}$,
then 

\noindent $\config{\theta}{C'}{K'}{\sigma}{\theta_K}{m}{w}  \vdash^{*} \config{\theta''}{C'''}{K'''}{\sigma''}{\theta''_K}{m+n'}{w''}$
and $\config{\theta''}{C''}{K''}{\sigma''}{\theta''_K}{m+n'}{w''} \relu \config{\theta''}{C'''}{K'''}{\sigma''}{\theta''_K}{m+n'}{w''}$
\end{corollary}

We can now show the desired result for terminating reductions.

\begin{lemma} \label{lemma:while-approx-os-left}
If
$\config{\theta}{\mathtt{while}^k(\phi)\{C\}}{[]}{\sigma}{\theta_K}{n}{w} \vdash^{*}
\config{\theta'}{\modownarrow}{[]}{\sigma'}{\theta_K}{n+n'}{w'}$,
then 
\noindent $\config{\theta}{\mathtt{while}(\phi)\{C\}}{[]}{\sigma}{\theta_K}{n}{w} \vdash^{*}
\config{\theta'}{\modownarrow}{[]}{\sigma'}{\theta_K}{n+n'}{w'}$.
%Works also for \sigma = \failure
\end{lemma}
\begin{proof}
We have
$\config{\theta}{\mathtt{while}^k(\phi)\{C\}}{[]}{\sigma}{\theta_K}{n}{w} \relu
\config{\theta}{\mathtt{while}(\phi)\{C\}}{[]}{\sigma}{\theta_K}{n}{w}$,
so by Corollary~\ref{corr:sim-red}, 
$\config{\theta}{\mathtt{while}(\phi)\{C\}}{[]}{\sigma}{\theta_K}{n}{w} \vdash^{*}
\config{\theta'}{C'}{K'}{\sigma'}{\theta_K}{n+n'}{w'}$
where $\modownarrow \relu C'$ and $[] \relu K'$, which implies $C' = \modownarrow$ and $K' = []$.
\qed \end{proof}

If the evaluation of $\mathtt{while}^k(\phi)\{C\}$ gets stuck, so does the evaluation of 
$\mathtt{while}(\phi)\{C\}$.

\begin{lemma} \label{lemma:relu-both-reduce}
If $\config{\theta}{C}{K}{\sigma}{\theta_K}{n}{w} \vdash \config{\theta'}{C'}{K'}{\sigma'}{\theta_K'}{n'}{w'}$
and $\hat{C} \relu C$ and $ \hat{K} \relu K$, then 
$\config{\theta}{\hat{C}}{\hat{K}}{\sigma}{\theta_K}{n}{w} \vdash \config{\theta''}{C''}{K''}{\sigma''}{\theta_K''}{n''}{w''}$.
\end{lemma}
\begin{proof}
%TODO
By case analysis on the derivation of $\hat{C} \relu C$.
%By simultaneous induction on the derivation of $\hat{C} \relu C$ and $ \hat{K} \relu K$ ?
\qed \end{proof}

\begin{lemma} \label{lemma:while-error}
If 
$\config{\theta}{\mathtt{while}^k(\phi)\{C\}}{[]}{\sigma}{\theta_K}{n}{w} \vdash^{*}
\config{\theta'}{C'}{K}{\sigma'}{\theta_K'}{n+n'}{w'} \nvdash$, then
$\config{\theta}{\mathtt{while}(\phi)\{C\}}{[]}{\sigma}{\theta_K}{n}{w} \vdash^{*}
\config{\theta'}{C''}{K'}{\sigma'}{\theta_K'}{n+n'}{w'} \nvdash$.
\end{lemma}
\begin{proof}
%Almost identical to the proof of Lemma~\ref{lemma:while-k-error}:

If $C' \neq \texttt{diverge}$, then by Corollary~\ref{corr:sim-red}, 
$\config{\theta}{\mathtt{while}(\phi)\{C\}}{[]}{\sigma}{\theta_K}{n}{w} \vdash^{*}
\config{\theta'}{C''}{K'}{\sigma'}{\theta_K}{n+n'}{w'}$
where $C' \relu C''$ and $K \relu K'$. By Lemma~\ref{lemma:relu-both-reduce}, if
$\config{\theta'}{C''}{K'}{\sigma'}{\theta_K}{n+n'}{w'}$ reduces, then 
$\config{\theta'}{C'}{K}{\sigma'}{\theta_K}{n+n'}{w'}$ also reduces, contradicting the assumption.
Hence, $\config{\theta'}{C''}{K'}{\sigma'}{\theta_K}{n+n'}{w'} \nvdash$, as required.

If $C' = \texttt{diverge}$, then $\sigma' = \failure$, as otherwise
$\config{\theta'}{\texttt{diverge}}{K}{\sigma'}{\theta_K'}{n+n'}{w'}$ would reduce by (diverge).
However, $\config{\theta'}{\texttt{diverge}}{K}{\failure}{\theta_K'}{n+n'}{w'}$ is not derivable from
any initial configuration other than itself. Hence, $n' = 0$ and $k=0$ and $\sigma = \failure$.
Since no configuration with state $\failure$ reduces, we have
$\config{\theta}{\mathtt{while}(\phi)\{C\}}{[]}{\failure}{\theta_K}{n}{w} \nvdash$, as required. 
\qed \end{proof}
\begin{corollary} \label{corr:o-while-error}
$\mathbf{O}_{\mathtt{while}(\phi)\{C\}}^\sigma(\theta) \geq \mathbf{O}_{\mathtt{while}^k(\phi)\{C\}}^\sigma(\theta)$
for all $k$.
%(w.r.t. flat CPO with bottom $\diverge$).
\end{corollary}

\subsubsection{Replacing one bounded loop with another}

We now prove that a bounded loop $\mathtt{while}^k(\phi)\{C\}$ can be safely replaced by 
another bounded loop with a higher bound.

\begin{lemma} \label{lemma:o-while-increasing}
If $m \geq k$ and 
$\config{\theta}{\mathtt{while}^k(\phi)\{C\}}{[]}{\sigma}{\theta_K}{n}{w} \vdash^{*}
\config{\theta'}{\modownarrow}{[]}{\sigma'}{\theta_K}{n+n'}{w'}$, then
$\config{\theta}{\mathtt{while}^m(\phi)\{C\}}{[]}{\sigma}{\theta_K}{n}{w} \vdash^{*}
\config{\theta'}{\modownarrow}{[]}{\sigma'}{\theta_K}{n+n'}{w'}$

\end{lemma}
%Works also for \sigma' = \failure
\begin{proof}
We have
$\config{\theta}{\mathtt{while}^k(\phi)\{C\}}{[]}{\sigma}{\theta_K}{n}{w} \relu
\config{\theta}{\mathtt{while}^m(\phi)\{C\}}{[]}{\sigma}{\theta_K}{n}{w}$,
so by Corollary~\ref{corr:sim-red}, 
$\config{\theta}{\mathtt{while}^m(\phi)\{C\}}{[]}{\sigma}{\theta_K}{n}{w} \vdash^{*}
\config{\theta'}{C'}{K'}{\sigma'}{\theta_K}{n+n'}{w'}$
where $\modownarrow \relu C'$ and $[] \relu K'$, which implies $C' = \modownarrow$ and $K' = []$.
\qed \end{proof}

We show the same property for reductions which get stuck.

\begin{lemma} \label{lemma:while-k-error}
If $m \geq k$ and 
$\config{\theta}{\mathtt{while}^k(\phi)\{C\}}{[]}{\sigma}{\theta_K}{n}{w} \vdash^{*}
\config{\theta'}{C'}{K}{\sigma'}{\theta_K'}{n+n'}{w'} \nvdash$, then
$\config{\theta}{\mathtt{while}^m(\phi)\{C\}}{[]}{\sigma}{\theta_K}{n}{w} \vdash^{*}
\config{\theta'}{C''}{K'}{\sigma'}{\theta_K'}{n+n'}{w'} \nvdash$.
\end{lemma}
\begin{proof}
If $C' \neq \texttt{diverge}$, then by Corollary~\ref{corr:sim-red}, 
$\config{\theta}{\mathtt{while}^m(\phi)\{C\}}{[]}{\sigma}{\theta_K}{n}{w} \vdash^{*}
\config{\theta'}{C''}{K'}{\sigma'}{\theta_K}{n+n'}{w'}$
where $C' \relu C''$ and $K \relu K'$. By Lemma~\ref{lemma:relu-both-reduce}, if
$\config{\theta'}{C''}{K'}{\sigma'}{\theta_K}{n+n'}{w'}$ reduces, then 
$\config{\theta'}{C'}{K}{\sigma'}{\theta_K}{n+n'}{w'}$ also reduces, contradicting the assumption.
Hence, $\config{\theta'}{C''}{K'}{\sigma'}{\theta_K}{n+n'}{w'} \nvdash$, as required.

If $C' = \texttt{diverge}$, then $\sigma' = \failure$, as otherwise
$\config{\theta'}{\texttt{diverge}}{K}{\sigma'}{\theta_K'}{n+n'}{w'}$ would reduce by (diverge).
However, $\config{\theta'}{\texttt{diverge}}{K}{\failure}{\theta_K'}{n+n'}{w'}$ is not derivable from
any initial configuration other than itself. Hence, $n' = 0$ and $k=0$ and $\sigma = \failure$.
Since no configuration with state $\failure$ reduces, we have
$\config{\theta}{\mathtt{while}^m(\phi)\{C\}}{[]}{\failure}{\theta_K}{n}{w} \nvdash$, as required.
\qed \end{proof}

The above results lead to the following properties of semantic functions:

%increasing chain
\begin{corollary} \label{corr:o-while-increasing-chain}
If $n \geq k$, then  $\mathbf{O}_{\mathtt{while}^n(\phi)\{C\}}^\sigma(\theta) \geq \mathbf{O}_{\mathtt{while}^k(\phi)\{C\}}^\sigma(\theta)$
(w.r.t. flat CPO with bottom $\diverge$).
\end{corollary}

\begin{lemma} \label{lemma:sc-term-increasing-chain}
\sloppy If  $\mathbf{O}_{\mathtt{while}^k(\phi)\{C\}}^\sigma(\theta) \in \statespace$
and $\mathbf{O}_{\mathtt{while}^l(\phi)\{C\}}^\sigma(\theta) \in \statespace$, then 
$\mathbf{SC}_{\mathtt{while}^k(\phi)\{C\}}^\sigma(\theta) = \mathbf{SC}_{\mathtt{while}^l(\phi)\{C\}}^\sigma(\theta)$.
\end{lemma} 
\begin{proof}
Assume w.l.o.g. that $l \geq k$. Then the result follows directly from Lemma~\ref{lemma:o-while-increasing}.
\qed \end{proof}

\subsubsection{Proofs of Propositions~\ref{lemma:sup-while-o-sc} and~\ref{lemma:o-sc-increasing}}

Having shown the above properties of while-loop approximations, we are now ready to prove 
Propositions~\ref{lemma:sup-while-o-sc} and~\ref{lemma:o-sc-increasing}.

\begin{restate}{Proposition~\ref{lemma:o-sc-increasing}}
\sloppy If $n \geq k$, %and $\hat{f}(\failure) = \hat{f}(\diverge) = 0$, 
then  
$\hat{f}(\mathbf{O}_{\mathtt{while}^n(\phi)\{C\}}^\sigma(\theta)) \allowbreak\mathbf{SC}_{\mathtt{while}^n(\phi)\{C\}}^\sigma(\theta)
 \geq \hat{f}(\mathbf{O}_{\mathtt{while}^k(\phi)\{C\}}^\sigma(\theta)) \mathbf{SC}_{\mathtt{while}^k(\phi)\{C\}}^\sigma(\theta)$.
\end{restate}
\begin{proof}[of Proposition~\ref{lemma:o-sc-increasing}]
If $\mathbf{O}_{\mathtt{while}^k(\phi)\{C\}}^\sigma(\theta) = \failure$
or $\mathbf{O}_{\mathtt{while}^k(\phi)\{C\}}^\sigma(\theta) = \diverge$, then $RHS=0$, so the
inequality holds trivially.

\sloppy If $\mathbf{O}_{\mathtt{while}^k(\phi)\{C\}}^\sigma(\theta) \in \statespace$, then by Corollary~\ref{corr:o-while-increasing-chain}, 
$\mathbf{O}_{\mathtt{while}^n(\phi)\{C\}}^\sigma(\theta) = \mathbf{O}_{\mathtt{while}^k(\phi)\{C\}}^\sigma(\theta) $
and by Lemma~\ref{lemma:sc-term-increasing-chain},
$\mathbf{SC}_{\mathtt{while}^n(\phi)\{C\}}^\sigma(\theta) = \mathbf{SC}_{\mathtt{while}^k(\phi)\{C\}}^\sigma(\theta)$. 
Hence, $f(\mathbf{O}_{\mathtt{while}^n(\phi)\{C\}}^\sigma(\theta)) \allowbreak\mathbf{SC}_{\mathtt{while}^n(\phi)\{C\}}^\sigma(\theta)
 = f(\mathbf{O}_{\mathtt{while}^k(\phi)\{C\}}^\sigma(\theta)) \mathbf{SC}_{\mathtt{while}^k(\phi)\{C\}}^\sigma(\theta)$.
\qed \end{proof}

\begin{restate}{Proposition~\ref{lemma:sup-while-o-sc}}
%\begin{align*}
%\begin{autobreak}
\sloppy $\hat{f}(\mathbf{O}_{\mathtt{while}(\phi)\{C\}}^\sigma(\theta))  \mathbf{SC}_{\mathtt{while}(\phi)\{C\}}^\sigma(\theta)
= \sup_n   \hat{f}( \mathbf{O}_{\mathtt{while}^n(\phi)\{C\}}^\sigma(\theta)) \mathbf{SC}_{\mathtt{while}^n(\phi)\{C\}}^\sigma(\theta)$
%\end{autobreak}
%\end{align*}
\end{restate}
\begin{proof}[of Proposition~\ref{lemma:sup-while-o-sc}]
\sloppy If $\mathbf{O}_{\mathtt{while}(\phi)\{C\}}^\sigma(\theta) \notin \statespace$, then $LHS = 0$. 
If $\mathbf{O}_{\mathtt{while}^n(\phi)\{C\}}^\sigma(\theta) \in \statespace$ 
for some $n$, then we get a contradiction by Lemma~\ref{lemma:while-approx-os-left}, so we have
$\mathbf{O}_{\mathtt{while}^n(\phi)\{C\}}^\sigma(\theta)  \notin \statespace$, which implies $RHS=0$.

Now, assume that $\mathbf{O}_{\mathtt{while}(\phi)\{C\}}^\sigma(\theta) \in \statespace$.
Then by Lemma~\ref{lemma:while-o-sc-sup}, 
there exists $k$ such that 
$\mathbf{O}_{\mathtt{while}(\phi)\{C\}}^\sigma(\theta)
= \mathbf{O}_{\mathtt{while}^k(\phi)\{C\}}^\sigma(\theta)$
and $ \mathbf{SC}_{\mathtt{while}(\phi)\{C\}}^\sigma(\theta) =  \mathbf{SC}_{\mathtt{while}^k(\phi)\{C\}}^\sigma(\theta)$.

By Corollary~\ref{corr:o-while-increasing-chain} we know that
$\mathbf{O}_{\mathtt{while}^l(\phi)\{C\}}^\sigma(\theta) = \mathbf{O}_{\mathtt{while}^k(\phi)\{C\}}^\sigma(\theta)$
for all $l \geq k$ and either
$\mathbf{O}_{\mathtt{while}^k(\phi)\{C\}}^\sigma(\theta) = \mathbf{O}_{\mathtt{while}^{l'}(\phi)\{C\}}^\sigma(\theta)$
or $\mathbf{O}_{\mathtt{while}^{l'}(\phi)\{C\}}^\sigma(\theta) = \diverge$ for all $l' \leq k$.
Hence, for all $l$, either 
$\hat{f}(\mathbf{O}_{\mathtt{while}^l(\phi)\{C\}}^\sigma(\theta)) = \hat{f}(\mathbf{O}_{\mathtt{while}^k(\phi)\{C\}}^\sigma(\theta))$
or
$\hat{f}(\mathbf{O}_{\mathtt{while}^l(\phi)\{C\}}^\sigma(\theta)) = 0$.

By Lemma \ref{lemma:sc-term-increasing-chain}, for all  $l$, either
$\mathbf{O}_{\mathtt{while}^{l}(\phi)\{C\}}^\sigma(\theta) \notin \statespace$ or
$\mathbf{SC}_{\mathtt{while}^l(\phi)\{C\}}^\sigma(\theta) = \mathbf{SC}_{\mathtt{while}^k(\phi)\{C\}}^\sigma(\theta)$.
Hence, for all  $l$, either
$\hat{f}(\mathbf{O}_{\mathtt{while}^l(\phi)\{C\}}^\sigma(\theta)) \mathbf{SC}_{\mathtt{while}^l(\phi)\{C\}}^\sigma(\theta) 
= \hat{f}(\mathbf{O}_{\mathtt{while}^k(\phi)\{C\}}^\sigma(\theta)) \mathbf{SC}_{\mathtt{while}^k(\phi)\{C\}}^\sigma(\theta)$
or $\hat{f}(\mathbf{O}_{\mathtt{while}^l(\phi)\{C\}}^\sigma(\theta)) \mathbf{SC}_{\mathtt{while}^l(\phi)\{C\}}^\sigma(\theta)  = 0$.

\sloppy Thus, $\sup_n \hat{f}( \mathbf{O}_{\mathtt{while}^n(\phi)\{C\}}^\sigma(\theta)) \mathbf{SC}_{\mathtt{while}^n(\phi)\{C\}}^\sigma(\theta) 
= \hat{f}(\mathbf{O}_{\mathtt{while}^k(\phi)\{C\}}^\sigma(\theta)) \allowbreak\mathbf{SC}_{\mathtt{while}^k(\phi)\{C\}}^\sigma(\theta)$, 
and so $\hat{f}(\mathbf{O}_{\mathtt{while}(\phi)\{C\}}^\sigma(\theta)) \mathbf{SC}_{\mathtt{while}(\phi)\{C\}}^\sigma(\theta) = 
\sup_n \hat{f}( \mathbf{O}_{\mathtt{while}^n(\phi)\{C\}}^\sigma(\theta)) \allowbreak\mathbf{SC}_{\mathtt{while}^n(\phi)\{C\}}^\sigma(\theta) $, as required.
\qed \end{proof}

\subsubsection{Proofs of Propositions~\ref{lemma:inf-while-o-sc} and~\ref{lemma:o-sc-decreasing}}

Finally, we prove Propositions~\ref{lemma:inf-while-o-sc} and~\ref{lemma:o-sc-decreasing}, which
are required by Theorem~\ref{thm:wlp-op-equiv}. One final additional result needed for these proofs
is that $\mathbf{SC}_{\mathtt{while}^n(\phi)\{C\}}^\sigma(\theta) $ and 
$\mathbf{SC}_{\mathtt{while}^n(\phi)\{C\}}^\sigma(\theta, l) $ (for any $l$) are decreasing as functions of $n$.

\begin{lemma} \label{lemma:sc-decreasing}
If $n \geq k$, then 
$\mathbf{SC}_{\mathtt{while}^n(\phi)\{C\}}^\sigma(\theta) 
\leq \mathbf{SC}_{\mathtt{while}^k(\phi)\{C\}}^\sigma(\theta)$.
\end{lemma}
\begin{proof}
If $\mathbf{O}_{\mathtt{while}^k(\phi)\{C\} }^\sigma(\theta) = \failure$, then 
$\mathbf{O}_{\mathtt{while}^n(\phi)\{C\} }^\sigma(\theta) = \failure$ by
Corollary~\ref{corr:o-while-increasing-chain}. Hence, 
$\mathbf{SC}_{\mathtt{while}^k(\phi)\{C\} }^\sigma(\theta) =  \mathbf{SC}_{\mathtt{while}^n(\phi)\{C\} }^\sigma(\theta) = 0$.
%$\config{\theta}{\mathtt{while}^k(\phi)\{C\}}{[]}{\sigma}{\theta_K}{0}{1} \vdash^{*}
%\config{\theta'}{C'}{K}{\tau}{\theta_K'}{n }{w} \nvdash$ TODO

Now, suppose that $\mathbf{O}_{\mathtt{while}^k(\phi)\{C\} }^\sigma(\theta) \neq \failure$.
If there exists $l$ such that $\config{\theta}{\mathtt{while}^k(\phi)\{C\}}{[]}{\sigma}{\theta_K}{0}{1} \vdash^{*}_{\mathtt{min}}
\config{\theta'}{\mathtt{diverge}}{K}{\tau}{\theta_K'}{l }{w}$, then by Lemma~\ref{lemma:sim-red-div},
$\config{\theta}{\mathtt{while}^n(\phi)\{C\}}{[]}{\sigma}{\theta_K}{0}{1} \vdash^{*}%_{\mathtt{min}}
\config{\theta'}{C}{K'}{\tau}{\theta_K'}{l }{w}$ and
$\config{\theta'}{\mathtt{diverge}}{K}{\tau}{\theta_K'}{l }{w} \relu 
\config{\theta'}{C}{K'}{\tau}{\theta_K'}{l }{w}$.
Since $\config{\theta'}{\mathtt{diverge}}{K}{\tau}{\theta_K'}{m }{w} \vdash \config{\theta'}{\mathtt{diverge}}{K}{\tau}{\theta_K'}{m+1 }{w}$,
for all $l' \geq l$, we have $\mathbf{SC}_{\mathtt{while}^{l'}(\phi)\{C\} }^\sigma(\theta, l') = w$.
For each $l' \geq l$, we either have 
$\config{\theta}{\mathtt{while}^n(\phi)\{C\}}{[]}{\sigma}{\theta_K}{0}{1} \vdash^{*}
\config{\theta'}{C}{K'}{\tau}{\theta_K'}{l }{w}  \vdash^{*}
\config{\theta''}{C'}{K''}{\tau'}{\theta_K''}{l' }{w'}$, where $w' \leq w'$ by Lemma~\ref{lemma:w-decreasing}, and so 
$\mathbf{SC}_{\mathtt{while}^n(\phi)\{C\}}^\sigma(\theta, l') = w'$
or $\mathtt{while}^n(\phi)\{C\}$ does not reduce in $l'$ steps under $\theta$,
in which case \sloppy$\mathbf{SC}_{\mathtt{while}^n(\phi)\{C\}}^\sigma(\theta, l') = 0$.
In either case, $\mathbf{SC}_{\mathtt{while}^n(\phi)\{C\}}^\sigma(\theta, l') \leq \mathbf{SC}_{\mathtt{while}^k(\phi)\{C\}}^\sigma(\theta, l')$
for all $l' \geq l$, so the result holds by a property of the limit of a sequence.

If there exists no $l$ such that $\config{\theta}{\mathtt{while}^k(\phi)\{C\}}{[]}{\sigma}{\theta_K}{0}{1} \vdash^{*}_{\mathtt{min}}
\config{\theta'}{\mathtt{diverge}}{K}{\tau}{\theta_K'}{l }{w}$, then for all $l$, we have 
$\config{\theta}{\mathtt{while}^k(\phi)\{C\}}{[]}{\sigma}{\theta_K}{0}{1} \vdash^{*}
\config{\theta'}{C}{K}{\tau}{\theta_K'}{l }{w}$, where $C \neq \mathtt{diverge}$.
By Corollary~\ref{corr:sim-red},  $\config{\theta}{\mathtt{while}^n(\phi)\{C\}}{[]}{\sigma}{\theta_K}{0}{1} \vdash^{*}
\config{\theta'}{C'}{K'}{\tau}{\theta_K'}{l }{w}$ for some $C'$, $K'$, and so 
$\mathbf{SC}_{\mathtt{while}^k(\phi)\{C\}}^\sigma(\theta, l) = \mathbf{SC}_{\mathtt{while}^n(\phi)\{C\}}^\sigma(\theta, l)$ for all $l$,
which implies $\mathbf{SC}_{\mathtt{while}^k(\phi)\{C\}}^\sigma(\theta) = \mathbf{SC}_{\mathtt{while}^n(\phi)\{C\}}^\sigma(\theta)$.
\qed \end{proof}

\begin{lemma} \label{lemma:sim-red-div}
If $\config{\theta}{C}{K}{\sigma}{\theta_K}{m}{w} \relu  \config{\theta}{C'}{K'}{\sigma}{\theta_K}{m}{w}$
and $\config{\theta}{C}{K}{\sigma}{\theta_K}{m}{w} \vdash^{*}_{\mathtt{min}} \config{\theta''}{\mathtt{diverge}}{K''}{\sigma''}{\theta''_K}{m+n'}{w''}$
then 

\noindent $\config{\theta}{C'}{K'}{\sigma}{\theta_K}{m}{w}  \vdash^{*} \config{\theta''}{C'''}{K'''}{\sigma''}{\theta''_K}{m+n'}{w''}$
and $\config{\theta''}{\mathtt{diverge}}{K''}{\sigma''}{\theta''_K}{m+n'}{w''} \relu \config{\theta''}{C'''}{K'''}{\sigma''}{\theta''_K}{m+n'}{w''}$
\end{lemma}
\begin{proof}
Follows from Corollary \ref{corr:sim-red} and Lemma~\ref{lemma:relu-sim}.
\qed \end{proof}

\begin{lemma} \label{lemma:sc-decreasing-fixed-l}
If $n \geq k$, then for all $l$,
$ \mathbf{SC}_{\mathtt{while}^n(\phi)\{C\}}^\sigma(\theta, l) 
\leq \mathbf{SC}_{\mathtt{while}^k(\phi)\{C\}}^\sigma(\theta, l)$. 
\end{lemma}
\begin{proof}
If 
$\config{\theta}{\mathtt{while}^k(\phi)\{C\}}{[]}{\sigma}{\theta_K}{0}{1} \vdash^{*}
\config{\theta'}{C'}{K}{\sigma'}{\theta_K'}{l'}{w} \nvdash$ for some $l' < l$, then

\noindent
$\config{\theta}{\mathtt{while}^n(\phi)\{C\}}{[]}{\sigma}{\theta_K}{0}{1} \vdash^{*}
\config{\theta'}{C''}{K'}{\sigma'}{\theta_K'}{l'}{w} \nvdash$
by Lemma~\ref{lemma:while-k-error}, and so 
$ \mathbf{SC}_{\mathtt{while}^n(\phi)\{C\}}^\sigma(\theta, l) 
=  \mathbf{SC}_{\mathtt{while}^k(\phi)\{C\}}^\sigma(\theta, l) = 0$.

\sloppy If 
$\config{\theta}{\mathtt{while}^k(\phi)\{C\}}{[]}{\sigma}{\theta_K}{0}{1} \vdash^{*}
\config{\theta'}{\mathtt{diverge}}{K}{\sigma'}{\theta_K'}{l}{w}$, then  
$\mathbf{SC}_{\mathtt{while}^k(\phi)\{C\}}^\sigma(\theta, l) = w$ and there must exist
a $l' \leq l$ such that
$\config{\theta}{\mathtt{while}^k(\phi)\{C\}}{[]}{\sigma}{\theta_K}{0}{1} \vdash^{*}_{\mathtt{min}}
\config{\theta'}{\mathtt{diverge}}{K}{\sigma'}{\theta_K'}{l'}{w}$. 
Moreover, by Lemma~\ref{lemma:sim-red-div},
$\config{\theta}{\mathtt{while}^n(\phi)\{C\}}{[]}{\sigma}{\theta_K}{0}{1} \vdash^{*}%_{\mathtt{min}}
\config{\theta'}{C''}{K'}{\tau}{\theta_K'}{l' }{w}$ and
$\config{\theta'}{\mathtt{diverge}}{K}{\tau}{\theta_K'}{l' }{w} \relu 
\config{\theta'}{C''}{K'}{\tau}{\theta_K'}{l'}{w}$.
If we have 
$\config{\theta'}{C''}{K'}{\tau}{\theta_K'}{l'}{w} \vdash^{*} 
\config{\theta''}{C'''}{K''}{\tau}{\theta_K''}{l}{w'}$,
then $\mathbf{SC}_{\mathtt{while}^n(\phi)\{C\}}^\sigma(\theta, l) 
= w' \leq w$ by Lemma~\ref{lemma:w-decreasing}. Otherwise, $\mathbf{SC}_{\mathtt{while}^n(\phi)\{C\}}^\sigma(\theta, l) = 0$.
In either case, $ \mathbf{SC}_{\mathtt{while}^n(\phi)\{C\}}^\sigma(\theta, l) 
\leq  \mathbf{SC}_{\mathtt{while}^k(\phi)\{C\}}^\sigma(\theta, l) = 0$.

If 
$\config{\theta}{\mathtt{while}^k(\phi)\{C\}}{[]}{\sigma}{\theta_K}{0}{1} \vdash^{*}
\config{\theta'}{C'}{K}{\sigma'}{\theta_K'}{l}{w}$
and $C' \neq \mathtt{diverge}$, then  by Corollary~\ref{corr:sim-red},
$\config{\theta}{\mathtt{while}^n(\phi)\{C\}}{[]}{\sigma}{\theta_K}{0}{1} \vdash
\config{\theta'}{C''}{K'}{\tau}{\theta_K'}{l }{w}$ and
$\config{\theta'}{C'}{K}{\tau}{\theta_K'}{l }{w} \relu 
\config{\theta'}{C''}{K'}{\tau}{\theta_K'}{l}{w}$.
Thus, $ \mathbf{SC}_{\mathtt{while}^n(\phi)\{C\}}^\sigma(\theta, l) 
\leq  \mathbf{SC}_{\mathtt{while}^k(\phi)\{C\}}^\sigma(\theta, l) = w$.
\qed \end{proof}

\begin{restate}{Proposition~\ref{lemma:inf-while-o-sc}}
%If $\hat{f}(\failure) = \hat{f}(\diverge) = 0$, then
For all $f \leq 1$,

\noindent
$\check{f}(\mathbf{O}_{\mathtt{while}(\phi)\{C\}}^\sigma(\theta))  \mathbf{SC}_{\mathtt{while}(\phi)\{C\}}^\sigma(\theta)
= \inf_n   \check{f}( \mathbf{O}_{\mathtt{while}^n(\phi)\{C\}}^\sigma(\theta)) \mathbf{SC}_{\mathtt{while}^n(\phi)\{C\}}^\sigma(\theta)$
\end{restate}
\begin{proof}[of Proposition~\ref{lemma:inf-while-o-sc}]
If $\mathbf{O}_{\mathtt{while}(\phi)\{C\}}^\sigma(\theta) \in \statespace$, then 
by Lemma~\ref{lemma:while-o-sc-sup}, 
there exists $k$ such that 
$\mathbf{O}_{\mathtt{while}(\phi)\{C\}}^\sigma(\theta)
= \mathbf{O}_{\mathtt{while}^k(\phi)\{C\}}^\sigma(\theta)$
and $ \mathbf{SC}_{\mathtt{while}(\phi)\{C\}}^\sigma(\theta) =  \mathbf{SC}_{\mathtt{while}^k(\phi)\{C\}}^\sigma(\theta)$.
By similar reasoning as in the proof of Proposition~\ref{lemma:sup-while-o-sc}, for all $l$, either 
$\check{f}(\mathbf{O}_{\mathtt{while}^l(\phi)\{C\}}^\sigma(\theta)) = \check{f}(\mathbf{O}_{\mathtt{while}^k(\phi)\{C\}}^\sigma(\theta))$
or
$\check{f}(\mathbf{O}_{\mathtt{while}^l(\phi)\{C\}}^\sigma(\theta)) = 1$, 
so $\check{f}(\mathbf{O}_{\mathtt{while}^l(\phi)\{C\}}^\sigma(\theta)) \geq  \check{f}(\mathbf{O}_{\mathtt{while}^k(\phi)\{C\}}^\sigma(\theta))$
for all $l$.

\ 

\sloppy By Lemma \ref{lemma:sc-term-increasing-chain}, for all  $l$, either
$\mathbf{O}_{\mathtt{while}^{l}(\phi)\{C\}}^\sigma(\theta) \notin \statespace$ or
$\mathbf{SC}_{\mathtt{while}^l(\phi)\{C\}}^\sigma(\theta) = \mathbf{SC}_{\mathtt{while}^k(\phi)\{C\}}^\sigma(\theta)$.
If $\mathbf{O}_{\mathtt{while}^{l}(\phi)\{C\}}^\sigma(\theta) \notin \statespace$, then $l < k$ 
%and  $\mathbf{O}_{\mathtt{while}^{l}(\phi)\{C\}}^\sigma(\theta) = \diverge$, 
because of Corollary~\ref{corr:o-while-increasing-chain}.
Moreover, by Lemma~\ref{lemma:sc-decreasing}, if $l < k$, then $ \mathbf{SC}_{\mathtt{while}^k(\phi)\{C\}}^\sigma(\theta) 
\leq \mathbf{SC}_{\mathtt{while}^l(\phi)\{C\}}^\sigma(\theta)$. 
Hence, $ \mathbf{SC}_{\mathtt{while}^k(\phi)\{C\}}^\sigma(\theta) 
\leq \mathbf{SC}_{\mathtt{while}^l(\phi)\{C\}}^\sigma(\theta)$ for all $l$.
This implies
$\inf_n   \check{f}( \mathbf{O}_{\mathtt{while}^n(\phi)\{C\}}^\sigma(\theta)) \mathbf{SC}_{\mathtt{while}^n(\phi)\{C\}}^\sigma(\theta)
= \check{f}( \mathbf{O}_{\mathtt{while}^k(\phi)\{C\}}^\sigma(\theta)) \mathbf{SC}_{\mathtt{while}^k(\phi)\{C\}}^\sigma(\theta)
= \check{f}(\mathbf{O}_{\mathtt{while}(\phi)\{C\}}^\sigma(\theta))  \mathbf{SC}_{\mathtt{while}(\phi)\{C\}}^\sigma(\theta)$.

\ 

If $\mathbf{O}_{\mathtt{while}(\phi)\{C\}}^\sigma(\theta) = \failure$, then
by Lemma~\ref{lemma:while-error-some-k}, $\mathbf{O}_{\mathtt{while}^k(\phi)\{C\}}^\sigma(\theta) = \failure$
for some $k$. Thus, 
$\inf_n \check{f}( \mathbf{O}_{\mathtt{while}^n(\phi)\{C\}}^\sigma(\theta)) 
\mathbf{SC}_{\mathtt{while}^n(\phi)\{C\}}^\sigma(\theta) = 0 = 
\check{f}(\mathbf{O}_{\mathtt{while}(\phi)\{C\}}^\sigma(\theta))  \mathbf{SC}_{\mathtt{while}(\phi)\{C\}}^\sigma(\theta)$.

\ 

If $\mathbf{O}_{\mathtt{while}(\phi)\{C\}}^\sigma(\theta) = \diverge$, then 
$\check{f}(\mathbf{O}_{\mathtt{while}(\phi)\{C\}}^\sigma(\theta)) = 1$.
By Lemma~\ref{corr:o-while-error}, $\mathbf{O}_{\mathtt{while}^k(\phi)\{C\}}^\sigma(\theta) = \diverge$ for all $k$.
Since $\check{f}(\diverge) = 1$, we only need to show that 
$\mathbf{SC}_{\mathtt{while}(\phi)\{C\}}^\sigma(\theta) = \inf_n \mathbf{SC}_{\mathtt{while}^n(\phi)\{C\}}^\sigma(\theta)$.

First, observe that from Corollary~\ref{corr:bisim-n-steps}, it follows that for all $l$,
%there exists $k$ such that 
for all $k \geq l$, 
$\mathbf{SC}_{\mathtt{while}(\phi)\{C\}}^\sigma(\theta, l) = \mathbf{SC}_{\mathtt{while}^k(\phi)\{C\}}^\sigma(\theta, l)$.
Thus, for such fixed $l$, 
$\mathbf{SC}_{\mathtt{while}(\phi)\{C\}}^\sigma(\theta, l) = \inf_n \mathbf{SC}_{\mathtt{while}^n(\phi)\{C\}}^\sigma(\theta, l)$.
Hence,
\begin{eqnarray*}
\mathbf{SC}_{\mathtt{while}(\phi)\{C\}}^\sigma(\theta) &=&\inf_l\ \mathbf{SC}_{\mathtt{while}(\phi)\{C\}}^\sigma(\theta, l) \\
&=& \inf_l\ \inf_n\ \mathbf{SC}_{\mathtt{while}^n(\phi)\{C\}}^\sigma(\theta, l) \\
%Change order of summation  - possible because both sequences decreasing
&=& \inf_n\ \inf_l\ \mathbf{SC}_{\mathtt{while}^n(\phi)\{C\}}^\sigma(\theta, l) \\
&=& \inf_n\ \mathbf{SC}_{\mathtt{while}^n(\phi)\{C\}}^\sigma(\theta)
\end{eqnarray*}

In the equality $\inf_l\ \inf_n\ \mathbf{SC}_{\mathtt{while}^n(\phi)\{C\}}^\sigma(\theta, l)
%Change order of summation  - possible because both sequences decreasing
= \inf_n\ \inf_l\ \mathbf{SC}_{\mathtt{while}^n(\phi)\{C\}}^\sigma(\theta, l)$,
we used the fact that 
$\inf_l\ \inf_n\ \mathbf{SC}_{\mathtt{while}^n(\phi)\{C\}}^\sigma(\theta, l)
= \lim_{l -> \infty} \ \mathtt{lim}_{n -> \infty}\ \mathbf{SC}_{\mathtt{while}^n(\phi)\{C\}}^\sigma(\theta, l)$
and that $\mathbf{SC}_{\mathtt{while}^n(\phi)\{C\}}^\sigma(\theta, l)$ is decreasing in both $n$ and $l$,
which means that by Theorem 4.2 from \citep{Habil05},
$ \lim_{l -> \infty} \ \mathtt{lim}_{n -> \infty}\ \mathbf{SC}_{\mathtt{while}^n(\phi)\{C\}}^\sigma(\theta, l)
=  \lim_{n-> \infty} \ \mathtt{lim}_{l -> \infty}\ \mathbf{SC}_{\mathtt{while}^n(\phi)\{C\}}^\sigma(\theta, l)$.
\qed \end{proof}

Below, we write
$\config{\theta}{C}{K}{\sigma}{\theta_K}{n}{w}
\vdash^{*}_{\mathtt{min}}
\config{\theta'}{\mathtt{diverge}}{K'}{\sigma'}{\theta'_K}{n + n'}{w'}$
if $\config{\theta}{C}{K}{\sigma}{\theta_K}{n}{w}
\vdash^{*}
\config{\theta'}{\mathtt{diverge}}{K'}{\sigma'}{\theta'_K}{n + n'}{w'}$
and there is no $n'' < n'$ such that
$\config{\theta}{C}{K}{\sigma}{\theta_K}{n}{w}
\vdash^{*}
\config{\theta''}{\mathtt{diverge}}{K''}{\sigma''}{\theta''_K}{n + n''}{w''}$
(or, equivalently,
$\config{\theta}{C}{K}{\sigma}{\theta_K}{n}{w}
\vdash^{*}
\config{\theta'}{\mathtt{diverge}}{K'}{\sigma'}{\theta'_K}{n + n'}{w'}$
was derived without (diverge)).

\begin{restate}{Proposition~\ref{lemma:o-sc-decreasing}}
If $n \geq k$ and $f \leq 1$, then  

\noindent
$\check{f}(\mathbf{O}_{\mathtt{while}^n(\phi)\{C\}}^\sigma(\theta)) \mathbf{SC}_{\mathtt{while}^n(\phi)\{C\}}^\sigma(\theta)
 \leq \check{f}(\mathbf{O}_{\mathtt{while}^k(\phi)\{C\}}^\sigma(\theta)) \mathbf{SC}_{\mathtt{while}^k(\phi)\{C\}}^\sigma(\theta)$.
\end{restate}
\begin{proof}[of Proposition~\ref{lemma:o-sc-decreasing}]
\sloppy By Corollary~\ref{corr:o-while-increasing-chain}, $\mathbf{O}_{\mathtt{while}^n(\phi)\{C\}}^\sigma(\theta)
\geq \mathbf{O}_{\mathtt{while}^k(\phi)\{C\}}^\sigma(\theta)$. Since $\check{f}$ is antitone (we have
$\check{f}(\tau) \leq \check{f}(\diverge) = 1$ for all $\tau \geq \diverge$), this implies
$\check{f}(\mathbf{O}_{\mathtt{while}^n(\phi)\{C\}}^\sigma(\theta))  \leq \check{f}(\mathbf{O}_{\mathtt{while}^k(\phi)\{C\}}^\sigma(\theta))$.
By Lemma~\ref{lemma:sc-decreasing}, $ \mathbf{SC}_{\mathtt{while}^n(\phi)\{C\}}^\sigma(\theta) 
\leq \mathbf{SC}_{\mathtt{while}^k(\phi)\{C\}}^\sigma(\theta)$, so 
$\check{f}(\mathbf{O}_{\mathtt{while}^n(\phi)\{C\}}^\sigma(\theta)) \mathbf{SC}_{\mathtt{while}^n(\phi)\{C\}}^\sigma(\theta)
 \leq \check{f}(\mathbf{O}_{\mathtt{while}^k(\phi)\{C\}}^\sigma(\theta)) \mathbf{SC}_{\mathtt{while}^k(\phi)\{C\}}^\sigma(\theta)$,
as required.
\qed \end{proof}

\section{Proofs of Theorems~\ref{thm:wp-op-equiv} and \ref{thm:wlp-op-equiv}} \label{section:thms-1-2}

%\subsection{Approximating while-loops.} 

%\subsection{Proofs of Theorem 1 and 2}

\begin{restate}{Theorem ~\ref{thm:wp-op-equiv}}
For all measurable functions $f \colon \statespace -> \extposreals$, \ccpgcl programs $C$ and initial states $\sigma \in \statespace$:
$$
\mathtt{wp} |[ C |](f)(\sigma) \ = \ \int f(\tau) |[C|]_{\sigma}(d\tau). 
$$
\end{restate}

\begin{proof}
By Lemma~\ref{lemma:unfold-exp-wrt-sem}, it suffices to prove that for all $f$:
$$
\int \hat{f}(\mathbf{O}_{C}^\sigma(\theta)) \cdot \mathbf{SC}_C^\sigma(\theta)\, \mu_{\mathbb{S}}(d\theta) \ = \mathtt{wp}  |[ C |](f)(\sigma).
$$
This can be proven by induction on the structure of $C$.   
We refrain from treating all cases but restrict ourselves to some interesting cases:
\begin{itemize}
\item Case $C = x :\approx U$.
\begin{eqnarray*}
& &  \int \hat{f}(\mathbf{O}_{x :\approx U}^\sigma(\theta)) \cdot \mathbf{SC}_{x :\approx U}^\sigma(\theta)\, \mu_{\mathbb{S}}(d\theta) \\
&= \ &\int f(\sigma[x \mapsto \pi_U(\pi_L(\theta))])  \, \mu_{\mathbb{S}}(d\theta) 
\\ \text{\tiny (property entropy)} &= \ &\int_{[0,1]} f(\sigma[x \mapsto v])  \, \mu_L(dv) \\ 
\text{\tiny (definition $\mathtt{wp}$)} &=&\mathtt{wp}|[ x :\approx U |](f)(\sigma).
\end{eqnarray*}
\item Case $C = C_1; C_2$ with $C_1 \neq C'_1; C'_2$.
\begin{eqnarray*}
& & \int \hat{f}(\mathbf{O}_{C_1;C_2}^\sigma(\theta)) \cdot \mathbf{SC}_{C_1;C_2}^\sigma(\theta)\, \mu_{\mathbb{S}}(d\theta) \\
\text{\tiny (Proposition~\ref{lemma:o-sc-seq})} &=& 
 \int \hat{f}(\mathbf{O}_{C_2}^{\tau}(\pi_R(\theta))) \cdot
%% \\ && \qquad 
\mathbf{SC}_{C_2}^{\tau}(\pi_R(\theta))\cdot\mathbf{SC}_{C_1}^\sigma(\pi_L(\theta))\, \mu_{\mathbb{S}}(d\theta) \\
\text{\tiny (property entropy)} &=&
\int  \underbrace{\int \hat{f}(\mathbf{O}_{C_2}^{\rho}(\theta_R))\cdot
 %% \\ && \qquad
\mathbf{SC}_{C_2}^{\rho}(\theta_R)  
\, \mu_{\mathbb{S}}(d\theta_R)}_{= g(\rho)} \cdot \mathbf{SC}_{C_1}^\sigma(\theta_L) \mu_{\mathbb{S}}(d\theta_L)
\end{eqnarray*}
where $\tau = \mathbf{O}_{C_1}^\sigma(\pi_L(\theta))$ and $\rho = \mathbf{O}_{C_1}^\sigma(\theta_L)$.
%%Now, let
%%$%\[
%%g(\tau) = \int \hat{f}(\mathbf{O}_{C_2}^{\tau}(\theta_R)) \mathbf{SC}_{C_2}^{\tau}(\theta_R)  \, \mu_{\mathbb{S}}(d\theta_R) 
%% $ %\]
%\noindent 
%% for $\tau \in \statespace$. 
We have:
\begin{eqnarray*}
& & 
\int \hat{g}(\mathbf{O}_{C_1}^\sigma(\theta_L))\cdot\mathbf{SC}_{C_1}^\sigma(\theta_L) \mu_{\mathbb{S}}(d\theta_L)\\
\text{\tiny (induction hypothesis)}&=&
\mathtt{wp}|[ C_1|](g)(\sigma)\\
&\ = \ & \mathtt{wp}|[ C_1|](\lambda \tau . \int \hat{f}(\mathbf{O}_{C_2}^{\tau}(\theta_R)) 
\cdot\mathbf{SC}_{C_2}^{\tau}(\theta_R)  \, \mu_{\mathbb{S}}(d\theta_R)  )(\sigma)\\
\text{\tiny (induction hypothesis)}&=& \mathtt{wp}|[ C_1|](\lambda \tau .\mathtt{wp}|[ C_2|](f)(\tau))(\sigma)\\
&=& \mathtt{wp}|[C_1 |] (\mathtt{wp}|[ C_2|](f)) (\sigma)\\
\text{\tiny (definition $\mathtt{wp}$)} &=&\mathtt{wp}|[ C_1;C_2|](f)(\sigma)
\end{eqnarray*}
\item Case $C = \mathtt{score}(E)$.
By inspecting the reduction rules, it follows:
\[
\mathbf{O}_{ \mathtt{score}(E)}^\sigma(\theta) \ = \
\begin{cases}
\sigma & \text{if}\ \sigma(E) \in (0,1] \\
\failure & \text{otherwise} \\
\end{cases}
\]
\noindent which implies $\hat{f}(\mathbf{O}_{ \mathtt{score}(E)}^\sigma(\theta)) \ = \
[\sigma(E) \in (0,1]] \cdot \hat{f}(\sigma)$
 and
\[
\mathbf{SC}_{ \mathtt{score}(E)}^\sigma(\theta) \ = \
\begin{cases}
\sigma(E) & \text{if}\ \sigma(E) \in (0,1] \\
0 & \text{otherwise} \\
\end{cases}
\ = \ 
[\sigma(E) \in (0,1]] \cdot \sigma(E).
\]
Thus, we have:
\begin{eqnarray*}
& & \int \hat{f}(\mathbf{O}_{ \mathtt{score}(E)}^\sigma(\theta)) \cdot \mathbf{SC}_{ \mathtt{score}(E)}^\sigma(\theta)\, \mu_{\mathbb{S}}(d\theta) \\
& \ = \ & \int [\sigma(E) \in (0,1]] \cdot \hat{f}(\sigma) \cdot \sigma(E)\, \mu_{\mathbb{S}}(d\theta) \\
&=& [\sigma(E) \in (0,1]] \cdot \hat{f}(\sigma) \cdot \sigma(E) \\
\text{\tiny ($\sigma \in \statespace$ by assumption)} &=& 
[\sigma(E) \in (0,1]] \cdot f(\sigma) \cdot \sigma(E)\\
&=& \mathtt{wp}|[ \mathtt{score}(E) |](f)(\sigma).
\end{eqnarray*}
\item Case $C = \mathtt{while}(\phi)\{C'\}$. Let $C^n = \mathtt{while}^n(\phi)\{C'\}$. We derive:
\begin{eqnarray*}
& &  
\int \hat{f}(\mathbf{O}_{C}^\sigma(\theta)) \cdot \mathbf{SC}_{C}^\sigma(\theta)\, \mu_{\mathbb{S}}(d\theta) \\
%\text{\tiny (by Proposition~\ref{lemma:sup-while-o-sc})} &\ = \ & \int \sup_n 
%\hat{f}( \mathbf{O}_{\mathtt{while}^n(\phi)\{C\}}^\sigma(\theta))  
%\mathbf{SC}_{\mathtt{while}^n(\phi)\{C\}}^\sigma(\theta)\, \mu_{\mathbb{S}}(d\theta) \\
\text{\tiny (Proposition~\ref{lemma:sup-while-o-sc})} &=& \int \sup_n 
\hat{f}( \mathbf{O}_{C^n}^\sigma(\theta)) \cdot \mathbf{SC}_{C^n}^\sigma(\theta)\, \mu_{\mathbb{S}}(d\theta) \\
\text{\tiny (Beppo Levi's Theorem)} &=& \sup_n \int \hat{f}( \mathbf{O}_{C^n}^\sigma(\theta))  
\cdot\mathbf{SC}_{C^n}^\sigma(\theta)\, \mu_{\mathbb{S}}(d\theta) \\
\text{\tiny ($*$)} &=& \sup_n {}^{\mathtt{wp}}_{\langle \phi, C' \rangle} \Phi_f^n (0) (\sigma) \\
\text{\tiny (Kleene's Fixpoint Theorem)}&=& \mathtt{wp}|[\mathtt{while}(\phi)\{C'\}|](f)(\sigma).
\end{eqnarray*}
%%
%Note that in the penultimate equality, the assumption of Lemma~\ref{lemma:approx-while-wp} is discharged
%by the induction hypothesis, which states that 
%$\int  \hat{f}(\mathbf{O}_{C}^\sigma(\theta))  \mathbf{SC}_C^\sigma(\theta)
%\, \mu_{\mathbb{S}}(d\theta)  =  \mathtt{wp}  |[ C |](f)(\sigma) $ for \emph{all} measurable functions $f$.
When applying the Beppo Levi's Theorem, we used the fact that the sequence
$\hat{f}( \mathbf{O}_{C^n}^\sigma(\theta))  \cdot \mathbf{SC}_{C^n}^\sigma(\theta)$
is monotonic in $n$ (Proposition~\ref{lemma:o-sc-increasing}).
In order to show that the proof step $(*)$ is correct, we need to show:
$$
\int \hat{f}(\mathbf{O}_{C^n}^\sigma(\theta)) 
\cdot\mathbf{SC}_{C^n}^{\sigma}(\theta)\, \mu_{\mathbb{S}}(d\theta) 
\ = \ 
{}^{\mathtt{wp}}_{\langle \phi, C' \rangle} \Phi_f^n (0) (\sigma) \mbox{ for all } n.
$$
We prove this statement by induction on $n$, using Proposition~\ref{lemma:rearrange-entropy}:
\begin{itemize}
\item Base case: $n=0$:
$$
\int \underbrace{\hat{f}(\mathbf{O}_{\mathtt{diverge}}^\sigma(\theta))}_{=0} \cdot \underbrace{\mathbf{SC}_{\mathtt{diverge}}^{\sigma}(\theta)\, \mu_{\mathbb{S}}(d\theta)}_{=1} \ = \ 0 \ = \ {}^{\mathtt{wp}}_{\langle \phi, C' \rangle} \Phi_f^0 (0) (\sigma)
$$
\item Induction step: we distinguish $\sigma(\phi) = \mathtt{true}$ and $\sigma(\phi) = \mathtt{false}$. For the latter case we have:
$$
\int \hat{f}(\sigma) \cdot 1\, \mu_{\mathbb{S}}(d\theta) \ = \ f(\sigma).
$$
For the case $\sigma(\phi) = \mathtt{true}$ we derive:
\begin{eqnarray*}
\!\!\!\!\!\!\!
& & \int \hat{f}(\mathbf{O}_{C^{n+1}}^\sigma(\theta)) \cdot
  \mathbf{SC}_{C^{n+1}}^{\sigma}(\theta)\, \mu_{\mathbb{S}}(d\theta) \\
&\ = \ & \int \hat{f}(\mathbf{O}_{C'; C^n}^\sigma(\theta)) \cdot
  \mathbf{SC}_{C'; C^n}^{\sigma}(\theta)\, \mu_{\mathbb{S}}(d\theta) \\
\text{ \tiny (Prop.~\ref{lemma:rearrange-entropy})} &=&
 \int \hat{f}(\mathbf{O}_{C^n}^\tau (\pi_R(\psi(\theta))))
{\cdot} \mathbf{SC}_{C'}^\sigma  (\pi_L(\psi(\theta))) {\cdot} 
%% \\ && \qquad \qquad
  \mathbf{SC}_{C^n}^{\tau}  (\pi_R(\psi(\theta)))\, \mu_{\mathbb{S}}(d\theta) \\
%\text{\tiny ($\psi$ is measure-preserving)}&=&
\text{\tiny (Prop\.~\ref{lemma:fsf-meas-pres})}&=&
\int \hat{f}(\mathbf{O}_{C^n}^{\rho}  (\pi_R(\theta))) 
  \cdot \mathbf{SC}_{C'}^\sigma  (\pi_L(\theta)) \cdot 
  %% \\ && \qquad \qquad
  \mathbf{SC}_{C^n}^{\rho}  (\pi_R(\theta))\, \mu_{\mathbb{S}}(d\theta) \\
\text{\tiny (entropy)}&=& 
\int \int \hat{f}(\mathbf{O}_{C^n}^{\rho}  (\theta_R)) \cdot
\mathbf{SC}_{C^n}^{\rho}  (\theta_R)
\, \mu_{\mathbb{S}}(d\theta_R) \cdot 
%% \\ && \qquad \qquad
\mathbf{SC}_{C'}^\sigma (\theta_L)
\, \mu_{\mathbb{S}}(d\theta_L)
\end{eqnarray*}
where $\tau = \mathbf{O}_{C'}^\sigma(\pi_L(\psi(\theta)))$ and $\rho = \mathbf{O}_{C'}^\sigma(\pi_L(\theta))$.

Now let $p(\tau) = \int \hat{f}(\mathbf{O}_{C^n}^{\tau}  (\theta_R)) \cdot \mathbf{SC}_{C^n}^{\tau}  (\theta_R)\, \mu_{\mathbb{S}}(d\theta_R)$
for $\tau \in \statespace$.
Then:
\begin{eqnarray*}
& & \int \hat{p}(\mathbf{O}_{C}^{\sigma}(\theta_L)) \cdot 
\mathbf{SC}_{C}^\sigma  (\theta_L)
\, \mu_{\mathbb{S}}(d\theta_L)\\
\text{\tiny (outer IH)}&=&
\mathtt{wp}  |[ C |]( p)(\sigma) \\
&\ = \ & \mathtt{wp}  |[ C |] \biggl( \lambda \tau . \int \hat{f}(\mathbf{O}_{C^n}^{\tau}  (\theta_R)) %% \\&&\qquad
\cdot \mathbf{SC}_{C^n}^{\tau}  (\theta_R)\, \mu_{\mathbb{S}}(d\theta_R)\biggr)(\sigma)\\
\text{\tiny (inner IH)}&=& \mathtt{wp}  |[ C |] \left ( \lambda \tau .{}^{\mathtt{wp}}_{\langle \phi, C' \rangle} \Phi_f^n (0) (\tau) \right)(\sigma)\\
&=&  \mathtt{wp}  |[ C |] \left ( {}^{\mathtt{wp}}_{\langle \phi, C' \rangle} \Phi_f^n (0) \right)(\sigma)\\
\text{\tiny (definition ${}^{\mathtt{wp}}_{\langle \phi, C \rangle} \Phi_f$)} &=& 
{}^{\mathtt{wp}}_{\langle \phi, C' \rangle} \Phi_f^{n+1} (0)(\sigma).
\end{eqnarray*}
\end{itemize}
\end{itemize}
Hence, the equality $(*)$ is correct, which finishes the proof.
 \qed \end{proof}

The second main theorem of this paper states that the weakest liberal preexpectation of a non-negative function $f$ bounded by $1$ is equivalent to the expected value of $f$ with respect to the distribution defined by the operational semantics plus the probability of divergence weighted by scores.

\begin{restate}{Theorem~\ref{thm:wlp-op-equiv}}
For every measurable non-negative function $f \colon \statespace -> \extposreals$ with $f(\sigma) \leq 1$ for all states $\sigma$, \ccpgcl program $C$ and initial state $\sigma \in \statespace$:
$$
\mathtt{wlp} |[ C |](f)(\sigma) \ = \ 
\int f(\tau) \cdot \restr{|[C|]_{\sigma}}{\statespace}(d\tau) + 
\underbrace{\int [\mathbf{O}_C^{\sigma}(\theta) =\ \diverge] \cdot \mathbf{SC}_C^{\sigma}(\theta)  \, \mu_{\mathbb{S}} (d\theta)}_{\mbox{\footnotesize probability of divergence multiplied by the score}}.
$$
\end{restate}
\begin{proof}
By induction on the structure of $C$. The proof is essentially the same as the proof of Theorem~\ref{thm:wp-op-equiv}, except that in the case of $\mathtt{while}$-loops, we use Proposition~\ref{lemma:inf-while-o-sc} instead of Proposition~\ref{lemma:sup-while-o-sc} to show that the $\mathtt{while}$-loop can be replaced by the limit of its finite approximations.

Similarly to Theorem~\ref{thm:wp-op-equiv}, the equation we want to prove can be rewritten as:
$$
\mathtt{wlp} |[ C |](f)(\sigma) \ = \  \int \check{f}(\mathbf{O}_{C}^\sigma(\theta)) \cdot
\mathbf{SC}_{C}^\sigma(\theta)\, \mu_{\mathbb{S}}(d\theta) 
$$
The proof goes as follows. Let $C = \mathtt{while}(\phi)\{C'\}$ and $C^n = \mathtt{while}^n(\phi)\{C'\}$.
\begin{eqnarray*}
 \int \check{f}(\mathbf{O}_{C}^\sigma(\theta)) \cdot
\mathbf{SC}_{C}^\sigma(\theta)\, \mu_{\mathbb{S}}(d\theta) \\
\text{\tiny (Proposition~\ref{lemma:inf-while-o-sc})} &=& 
\int \inf_n \check{f}( \mathbf{O}_{C^n}^\sigma(\theta)) \cdot \mathbf{SC}_{C^n}^\sigma(\theta)\, \mu_{\mathbb{S}}(d\theta) \\
\text{\tiny (Beppo Levi's Theorem)} &=& \inf_n \int \check{f}( \mathbf{O}_{C^n}^\sigma(\theta)) \cdot 
\mathbf{SC}_{C^n}^\sigma(\theta)\, \mu_{\mathbb{S}}(d\theta) \\
\text{\tiny ($*$)} &=& \inf_n {}^{\mathtt{wlp}}_{\langle \phi, C' \rangle} \Phi_f^n (1) (\sigma) \\
\text{\tiny (Kleene's Fixpoint Theorem)}&=& \mathtt{wlp}|[\mathtt{while}(\phi)\{C'\}|](f)(\sigma)
\end{eqnarray*}
In order to show that step $(*)$ is correct, we need to show that 
$\int \check{f}( \mathbf{O}_{C^n}^\sigma(\theta)) \cdot
\mathbf{SC}_{C^n}^\sigma(\theta)\, \mu_{\mathbb{S}}(d\theta)
= \inf_n {}^{\mathtt{wlp}}_{\langle \phi, C' \rangle} \Phi_f^n (1) (\sigma)$ 
for all $n$.
This can be proven by induction on $n$; the proof is almost identical to the proof of $(*)$ from Theorem~\ref{thm:wp-op-equiv}.
When applying the Beppo Levi's Theorem, we used the fact that the sequence
$\check{f}( \mathbf{O}_{C^n}^\sigma(\theta)) \cdot  \mathbf{SC}_{C^n}^\sigma(\theta)$
is decreasing in $n$ (Proposition~\ref{lemma:o-sc-decreasing}) and that
$\int \check{f}(\mathbf{O}_{C^0}^\sigma(\theta)) \cdot
\mathbf{SC}_{C^0}^\sigma(\theta)\, \mu_{\mathbb{S}}(d\theta) < \infty$, 
which can be checked immediately.
 \qed \end{proof}

\section{Proving measurability}
\label{app:proofs-meas}

The proofs of measurability are similar to \citep{SzymczakPhD}, with the %important 
difference that
we are working with an imperative language. In this section, we sketch the proofs of measurability of
functions $\mathbf{O}_C^\sigma(\cdot)$ and $\mathbf{SC}_C^\sigma(\cdot, n)$, without going into the details, %ed proofs,
which are conceptually the same as in \citep{SzymczakPhD}.

%\subsection{Measurable space of statements}
%
%In addition to states and traces, we also need to define a measurable space of statements $C$. This can be easily done
%by using a straightforward metric on syntactic terms, as in \citep{SzymczakPhD}. We omit the details.

\subsection{Measurability of single-step reduction}

\newcommand{\configtuple}[7]{(#1, #2 ,#3, #4, #5, #6, #7 )}

%Let $\mathbb{C}$ be expressions, $\mathbb{V}$ values, $\mathcal{C}$, $\mathcal{V}$ respective $\sigma$-algebras.
%
%First define a function $g_D :  \mathbb{C} \times \fullstatespace -> \mathbb{V}$ for evaluating deterministic expressions.
%\[
%g_D(E, \sigma)  = \sigma(E)
%\]
%
%Measurability of this would require defining expression evaluation formally. Alternatively we can assume that the language only supports measurable operations.
%
%Now, define:

Let us define:
\[
g\configtuple{\theta}{C}{K}{\sigma}{\theta_K}{n}{w} = 
\begin{cases}
\configtuple{\theta'}{C'}{K'}{\sigma'}{\theta'_K}{n+1}{w'}\\
\qquad  \text{if}\ \config{\theta}{C}{K}{\sigma}{\theta_K}{n}{w} \vdash \config{\theta'}{C'}{K'}{\sigma'}{\theta'_K}{n+1}{w'} \\
 \configtuple{\theta}{C}{K}{\failure}{\theta_K}{n+1}{0} \qquad \text{otherwise}
\end{cases}
\]
%\[
%g(\config{\theta}{C}{K}{\sigma}{\theta_K}{n}{w}) = \config{\theta'}{C'}{K'}{\sigma'}{\theta'_K}{n+1}{w'}
%\]
%\noindent if $\config{\theta}{C}{K}{\sigma}{\theta_K}{n}{w} \vdash \config{\theta'}{C'}{K'}{\sigma'}{\theta'_K}{n+1}{w'}$ and
%$g(\config{\theta}{C}{K}{\sigma}{\theta_K}{n}{w}) = \config{\theta}{C}{K}{\failure}{\theta_K}{n+1}{w} $ otherwise. 

We need to show that $g$ is measurable.
The only interesting cases are (assign), which modifies state (we need to show $g$ is still continuous in this case) and (draw), which modifies
both state and trace, and (seq) and (pop), which modify both the main trace and the trace for continuation.

%IDEA: Like in proof of Lemma 159 in \citep{SzymczakPhD}, treat components separately! - But we still need the space
%of traces to be a separable metric space... Or do we? TODO: check this!

%From this point on the proof should be bery sim

%From this point the rest of the proof is similar to \citep{SzymczakPhD}, section E.3.

%TODO

We can show that $g$ is measurable by considering $g$ as a disjoint union of sub-functions
defined on measurable subsets of combinations corresponding to given reduction rules
(e.g. $g_{\mathit{if-true}}$ and $g_{\mathit{if-false}}$ reducing conditional choices, $g_{\mathit{while-true}}$ and 
$g_{\mathit{while-false}}$ reducing while-loops,
$g_{\mathit{sample}}$ reducing sampling statements etc.)
and showing that each sub-function is measurable. The reasoning is very similar to the one presented in
Appendix E.1 of \citep{SzymczakPhD}, so we omit the full proof and only show measurability of
sub-functions modifying states and infinite traces, which were not present in \citep{SzymczakPhD}.

%TODO: finish this!

\subsubsection{From continuity to measurability}

The easiest way of proving measurability of a function is often proving that this function is continuous
as a function between the metric spaces which gave rise to the domain and codomain measurable 
spaces---by Corollary~\ref{corr:cont-fun-meas}, continuity implies measurability. 
Moreover, Corollary~\ref{corr:cont-fun-prod-meas} states that
if a function $f$ between products of separable metric spaces is continuous with respect to the Manhattan products of
metrics, then it is measurable with respect to products of the given measurable spaces. We will make heavy use of these
results in the proofs below.

%with respect to the metrics which generated the underlying $\sigma$-algebras.
%
%We will use Lemma~\ref{??}  and Corollary~\ref{corr:cont-fun-prod-meas} extensively

%If we want to prove that a function between products of measurable spaces is measurable,
%and each measurable space in question has a Borel $\sigma$-algebra induced by a separable metric space, it 
%suffices to prove that the function is continuous as a function between products of underlying metric spaces 
%(taken with the Manhattan metric, i.e. $(d((x_1,y_1), (x_2, y_2)) = d(x_1, y_1) + d(x_2, y_2)$). This fact 
%follows from standard results in metric and topological spaces and measure theory.
%% 
%We will often use it to prove measurability of subfunctions defining the reduction relation. It is easy
%to show that all metric spaces we are working with are separable.

%\paragraph{Metric space of entropies}  If we use the abstract entropy space, we will not 
%have a metric of entropies which we could use in proving continuity of transformations.
%Hence, to carry out measurability proofs, we need to assume a concrete representation 
%of the entropy space $\mathbb{S}$, such that  IS IT NEEDED?? CHECK!
%
%To carry out the proofs, we would need to use the concrete representation of the entropy space $\mathbb{S}$ as 
%the Hilbert cube $[0,1]^\omega$, as suggested in \citep[Section 4.3]{Culpepper17}, for which the $\sigma$-algebra can be defined as
%a Borel $\sigma$-algebra induced by a metric.

\paragraph{Additional Borel $\sigma$-algebras}

In order to carry out the proofs, we need to define separable metric spaces on statements $C$, expressions $E$ and continuations $K$, 
which will induce Borel $\sigma$-algebras. 
%
%Metrics on statements $C$ and continuations $K$  
These metrics are straightforward metrics on syntactic terms, 
similar to the metrics on lambda-terms in \citep{SzymczakPhD}. We omit the details, but these metrics
 would be defined so 
that $d_C(C_1;C_2, C_1';C_2') = d_C(C_1; C_1') + d_C(C_2; C_2')$ and 
$d_K(C\mathrel{::}K, C'\mathrel{::}K') = d_C(C; C') + d_K(K, K')$ (where $d_K(K, K') = \infty$ if $K$ and $K'$
have different lengths).
 
It is easy to check that all the above metric spaces are separable---for each of them, a dense subset can be
obtained by replacing reals with rationals. All subspaces of separable metric spaces can also be shown to be separable.

We also need to define $\sigma$-algebras on step sizes $n$ and weights $w$---these will be the standard 
discrete $\sigma$-algebra on $\mathbb{Z}_{+}$ and the Borel $\sigma$-algebra on $[0,1]$, respectively.
%The latter two are simply standard metric
%spaces $\mathbb{Z}$ and $\mathbb{R}_{+}$, with the Euclidean metric. 

\subsubsection{Measurability of (assign)}

We define:

\begin{eqnarray*}
% g_{assign} : ??? \\
 g_{\mathit{assign}}\configtuple{\theta}{x := E}{K}{\sigma}{\theta_K}{n}{w} &=&
\configtuple{\theta}{\modownarrow}{K}{\sigma[x \mapsto \sigma(E)]}{\theta_K}{n+1}{w} \\
&=&(g_{\mathit{assign1}}\configtuple{\theta}{x := E}{K}{\sigma}{\theta_K}{n}{w},\\
&& g_{\mathit{assign2}}\configtuple{\theta}{x := E}{K}{\sigma}{\theta_K}{n}{w}, \\
&&\dots,\\
&&g_{\mathit{assign7}}\configtuple{\theta}{x := E}{K}{\sigma}{\theta_K}{n}{w})
\end{eqnarray*}
\noindent where:
\begin{eqnarray*}
 g_{\mathit{assign1}}\configtuple{\theta}{x := E}{K}{\sigma}{\theta_K}{n}{w} &=& \theta \\
 g_{\mathit{assign2}}\configtuple{\theta}{x := E}{K}{\sigma}{\theta_K}{n}{w} &=& \modownarrow \\
%&\dots&\\
 g_{\mathit{assign3}}\configtuple{\theta}{x := E}{K}{\sigma}{\theta_K}{n}{w} &=& K \\
 g_{\mathit{assign4}}\configtuple{\theta}{x := E}{K}{\sigma}{\theta_K}{n}{w} &=& \sigma[x \mapsto \sigma(E)] \\
%&\dots&\\
g_{\mathit{assign5}}\configtuple{\theta}{x := E}{K}{\sigma}{\theta_K}{n}{w} &=& \theta_K \\
g_{\mathit{assign6}}\configtuple{\theta}{x := E}{K}{\sigma}{\theta_K}{n}{w} &=& n+1\\
g_{\mathit{assign7}}\configtuple{\theta}{x := E}{K}{\sigma}{\theta_K}{n}{w} &=& w
%&\dots&
\end{eqnarray*}

\begin{lemma}\label{lemma:assign-measurable}
$g_{\mathit{assign}}$ is measurable.
\end{lemma}
\begin{proof}
The functions $g_{\mathit{assign1}}$, $g_{\mathit{assign3}}$, $g_{\mathit{assign5}}$, $ g_{\mathit{assign7}}$ 
are simple projections, so they are trivially 
measurable.
The function $ g_{\mathit{assign2}}$ is a constant function, so it is also measurable.
Function $ g_{\mathit{assign4}}$ is a composition of a function returning the tuple $(x, \sigma, \sigma(E))$ from the configuration,
which can easily be shown measurable (projections are measurable, the function extracting
$E$ from $x:=E$ can be shown continuous and substitution $\sigma(E)$ is measurable
by assumption), and the state update function, which is measurable
by Lemma~\ref{lemma:subst-measurable}.
Function $ g_{\mathit{assign6}}$ is a composition of a projection (returning the sixth component $n$ from a tuple)
and a function adding $1$ to a number, which is continuous and measurable. 

Hence, $ g_{\mathit{assign}}$ is measurable, as all its components are measurable.
\qed \end{proof}

\subsubsection{Measurability of (draw)}

Let us define:
\begin{eqnarray*}
% g_{assign} : ??? \\
 g_{\mathit{draw}}(\configtuple{\theta}{x :\approx U}{K}{\sigma}{\theta_K}{n}{w}) &=&
\configtuple{\pi_R(\theta)}{\modownarrow}{K}{\sigma[x \mapsto \pi_U(\pi_L(\theta))]}{\theta_K}{n+1}{w} \\
&=&(g_{\mathit{draw1}}\configtuple{\theta}{x := E}{K}{\sigma}{\theta_K}{n}{w},\\
&& g_{\mathit{draw2}}\configtuple{\theta}{x := E}{K}{\sigma}{\theta_K}{n}{w}, \\
&&\dots,\\
&&g_{\mathit{draw7}}\configtuple{\theta}{x := E}{K}{\sigma}{\theta_K}{n}{w})
\end{eqnarray*}
\noindent where:
\begin{eqnarray*}
 g_{\mathit{draw1}}\configtuple{\theta}{x := E}{K}{\sigma}{\theta_K}{n}{w} &=& \pi_R(\theta) \\
 g_{\mathit{draw2}}\configtuple{\theta}{x := E}{K}{\sigma}{\theta_K}{n}{w} &=& \modownarrow \\
%&\dots&\\
g_{\mathit{draw3}}\configtuple{\theta}{x := E}{K}{\sigma}{\theta_K}{n}{w} &=& K \\
 g_{\mathit{draw4}}\configtuple{\theta}{x := E}{K}{\sigma}{\theta_K}{n}{w} &=& \sigma[x \mapsto \pi_U(\pi_L(\theta))] \\
%&\dots&\\
g_{\mathit{draw3}}\configtuple{\theta}{x := E}{K}{\sigma}{\theta_K}{n}{w} &=& \theta_K \\
 g_{\mathit{draw6}}\configtuple{\theta}{x := E}{K}{\sigma}{\theta_K}{n}{w} &=& n+1\\
%&\dots&
g_{\mathit{draw7}}\configtuple{\theta}{x := E}{K}{\sigma}{\theta_K}{n}{w} &=& w\\
\end{eqnarray*}

\begin{lemma}
$g_{\mathit{assign}}$ is measurable.
\end{lemma}
\begin{proof}
We only need to show the measurability of $ g_{\mathit{draw1}}$ and $ g_{\mathit{draw4}}$, as the other functions are identical
to the ones used in the definition of $g_{\mathit{assign}}$.

The function $ g_{\mathit{draw1}}$ is a composition of the projection returning the first component $\theta$ of the
configuration, and the function $\pi_R$, which is measurable by the axiomatisation of the entropy space, so it is measurable.

Function $ g_{\mathit{draw4}}$ is measurable by the same argument as $ g_{\mathit{assign4}}$, except that the measurable evaluation
$\sigma(E)$ is replaced by $\pi_U(\pi_L(\theta))$, which as a composition of two measurable (by assumption) functions 
and the measurable projection returning $\theta$ is also measurable.
\qed \end{proof}

\subsubsection{Measurability of (seq) and (pop)}

Define:
\begin{eqnarray*}
% g_{assign} : ??? \\
 g_{\mathit{seq}}(\configtuple{\theta}{C_1;C_2}{K}{\sigma}{\theta_K}{n}{w}) &=&
\configtuple{\pi_L(\theta)}{C_1}{C_2 \mathrel{::} K}{\sigma}{\pi_R(\theta) \mathrel{::} \theta_K}{n+1}{w} \\
&=&(g_{\mathit{seq1}}\configtuple{\theta}{x := E}{K}{\sigma}{\theta_K}{n}{w},\\
&& g_{\mathit{seq2}}\configtuple{\theta}{x := E}{K}{\sigma}{\theta_K}{n}{w}, \\
&&\dots,\\
&&g_{\mathit{seq7}}\configtuple{\theta}{x := E}{K}{\sigma}{\theta_K}{n}{w})
\end{eqnarray*}
\noindent where:
\begin{eqnarray*}
 g_{\mathit{seq1}}\configtuple{\theta}{C_1;C_2}{K}{\sigma}{\theta_K}{n}{w} &=& \pi_L(\theta) \\
 g_{\mathit{seq2}}\configtuple{\theta}{C_1;C_2}{K}{\sigma}{\theta_K}{n}{w} &=& C_1 \\
 g_{\mathit{seq3}}\configtuple{\theta}{C_1;C_2}{K}{\sigma}{\theta_K}{n}{w} &=& C_2 \mathrel{::} K  \\
%&\dots&\\
 g_{\mathit{seq4}}\configtuple{\theta}{C_1;C_2}{K}{\sigma}{\theta_K}{n}{w} &=& \sigma  \\
 g_{\mathit{seq5}}\configtuple{\theta}{C_1;C_2}{K}{\sigma}{\theta_K}{n}{w} &=&\pi_R(\theta) \mathrel{::} \theta_K \\
%&\dots&\\
 g_{\mathit{seq6}}\configtuple{\theta}{C_1;C_2}{K}{\sigma}{\theta_K}{n}{w} &=& n+1\\
 g_{\mathit{seq7}}\configtuple{\theta}{C_1;C_2}{K}{\sigma}{\theta_K}{n}{w} &=& w\\
%&\dots&
\end{eqnarray*}

\begin{lemma}
$g_{\mathit{seq}}$ is measurable.
\end{lemma}
\begin{proof}
The function $ g_{\mathit{seq1}}$ is measurable as a composition of projection and a function measurable
by assumption. The metrics $d_C$ and $d_K$ on statements and continuations 
(whose formal definitions are omitted) satisfy 
$d_C(C_1;C_2, C_1';C_2') = d_C(C_1; C_1') + d_C(C_2; C_2')$ and 
$d_K(C\mathrel{::}K, C'\mathrel{::}K') = d_C(C; C') + d_K(K, K')$, 
which makes it easy to show that  $g_{\mathit{seq2}}$ and $g_{\mathit{seq3}}$ are measurable,
as compositions of projections and continuous functions. Meanwhile, $g_{\mathit{seq5}}$ is composed from measurable projections and the 
functions $\pi_R$ and $(\mathrel{::})$, measurable by assumption, so it is measurable.
\qed \end{proof}

The proof of measurability of  (pop) is analogous.

\subsection{Measurability of $\mathbf{O}_C^{\sigma}(\cdot) $ and $\mathbf{SC}_C^{\sigma}(\cdot, n) $}
%From this point the rest of the proofs is similar to \citep{SzymczakPhD}, section E.3.

Once we have proven the measurability of state updates, the proof of Lemma~\ref{lemma:o-sc-measurable} (measurability
of $\mathbf{O}_C^{\sigma}(\cdot) $) is analogous to the proof of Lemma~92 in \cite{DBLP:conf/icfp/BorgstromLGS16}.

The proof of measurability of $\mathbf{SC}_C^{\sigma}(\cdot, n)$ is even simpler---for each fixed $n$, we can represent
$\mathbf{SC}_C^{\sigma}(\cdot, n)$ as an $n$-fold composition of $g$, followed by a projection returning
the weight $w$ from the configuration. The projection is obviously continuous, and so measurable. Since a composition of
measurable functions is measurable, this shows that $\mathbf{SC}_C^{\sigma}(\cdot, n)$ is measurable.

%\begin{lemma} \label{lemma:o-sc-measurable}
%For all $C$ and $\sigma \in \fullstatespace$, $\mathbf{O}_C^{\sigma}(\cdot) $ is $\mathcal{S} / \fullstatespace$ measurable.
%\end{lemma}
%\begin{proof}
%Analogous to the proof of Lemma~92 in \cite{DBLP:conf/icfp/BorgstromLGS16}, more detailed discussion in Appendix~\ref{app:proofs-meas}.
% \qed \end{proof}

%\begin{restate}{Lemma~\ref{lemma:o-measurable}}
%For all $C$ and $\sigma \in \fullstatespace$, $\mathbf{O}_C^{\sigma}(\cdot) $ is $\mathcal{S} / \fullstatespace$ measurable.
%\end{restate}
%\begin{proof}
%Similar to \citep{SzymczakPhD}.
%\qed \end{proof}

%\begin{restate}{Lemma~\ref{lemma:scn-measurable}}
%For all $C$ and $\sigma \in \fullstatespace$ and $n$, $\mathbf{SC}_C^{\sigma}(\cdot,n) $ is $\mathcal{S} / \mathcal{R}$ measurable.
%\end{restate}
%\begin{proof}
%Similar to \citep{SzymczakPhD}.
%%Here the difference is that we are only considering a fixed number of steps.
%\qed \end{proof}

\bibliography{bibl}
\end{document}